\documentclass[letterpaper,12pt]{article}
\usepackage{amsmath}
\usepackage{amssymb}
\usepackage{amsfonts}
\usepackage{setspace}
\usepackage{bbm}
\usepackage{color}
\usepackage{enumerate}
\usepackage{mathrsfs}
\usepackage{graphicx}
\usepackage{subcaption}
\usepackage{hyperref}
\usepackage{enumitem}
\usepackage[usenames,dvipsnames,svgnames,table]{xcolor}
\usepackage{multirow}
\usepackage[margin=1.00in]{geometry}
\usepackage{enumitem}
\setlist{nosep}

\newtheorem{theorem}{Theorem}[section]
\newtheorem{lemma}{Lemma}[section]
\newtheorem{proposition}{Proposition}[section]
\newtheorem{assumption}{Assumption}[section]
\newtheorem{condition}{Condition}[section]
\newtheorem{corollary}{Corollary}[section]

\newtheorem{remark}{Remark}[section]
\newenvironment{proof}[1][Proof]{\noindent \textbf{#1.} }{\  \rule{0.5em}{0.5em}}
\newcommand{\mb}[1]{\mathbb{#1}}
\newcommand{\wh}[1]{\widehat{#1}}
\newcommand{\wt}[1]{\widetilde{#1}}
\newcommand{\mf}[1]{\mathbf{#1}}
\newcommand{\mr}[1]{\mathrm{#1}}
\newcommand{\mcr}[1]{\mathscr{#1}}
\newcommand{\mc}[1]{\mathcal{#1}}
\newcommand{\tr}[1]{\mathrm{trace}({#1})}
\newcommand{\ind}{1\!\mathrm{l}}
\newcommand{\ol}[1]{\overline{#1}}

\hypersetup{
     colorlinks   = true,
     citecolor    = Blue,
     urlcolor     = Black,
     linkcolor    = Blue
}

\makeatletter
\renewcommand\paragraph{\@startsection{paragraph}{4}{\z@}%
                                    {0pt \@plus1ex \@minus.2ex}%
                                    {-1em}%
                                    {\normalfont\normalsize\bfseries}}
\makeatother

\usepackage{bibunits}
\usepackage[longnamesfirst]{natbib}

\begin{document}

\defaultbibliography{refsmain}
\defaultbibliographystyle{chicago}
\begin{bibunit}

\author{%
{ Timothy M. Christensen}\thanks{%
Department of Economics, New York University, 19 W. 4th Street, 6th floor, New York, NY 10012, USA. E-mail address: \texttt{timothy.christensen@nyu.edu}
}
}

\title{%
Nonparametric Stochastic Discount Factor Decomposition\thanks{%
This paper is based on my job market paper, which was titled ``Estimating the Long-Run Implications of Dynamic Asset Pricing Models'' and dated November 8, 2013. I am very grateful to my advisors Xiaohong Chen and Peter Phillips for their advice and support. I would like to thank the co-editor Lars Peter Hansen and three anonymous referees for insightful and helpful comments. I have benefited from feedback from Caio Almeida (discussant), Jarda Borovi\v{c}ka, Tim Cogley, Jean-Pierre Florens (discussant), Nour Meddahi, Keith O'Hara, Eric Renault, Guillaume Roussellet, Tom Sargent and participants of seminars at Chicago Booth, Columbia, Cornell, Duke, Indiana, Monash, Montr\'eal, Northwestern, NYU, Penn, Princeton, Rutgers, Sydney, UCL, UNSW, Vanderbilt, Wisconsin, and several conferences. All errors are my own.
}
}

\date{First version: May 2, 2013. Revised: May 19, 2017. }

\maketitle

\begin{abstract}  
\singlespacing
\noindent 
Stochastic discount factor (SDF) processes in dynamic economies admit a permanent-transitory decomposition in which the permanent component characterizes pricing over long investment horizons. This paper introduces an empirical framework to analyze the permanent-transitory decomposition of SDF processes. Specifically, we show how to estimate nonparametrically the solution to the Perron-Frobenius eigenfunction problem of \cite{HS2009}. Our empirical framework allows researchers to (i) recover the time series of the estimated permanent and transitory components and (ii) estimate the yield and the change of measure which characterize pricing over long investment horizons. We also introduce nonparametric estimators of the continuation value function in a class of models with recursive preferences by reinterpreting the value function recursion as a nonlinear Perron-Frobenius problem. We establish consistency and convergence rates of the eigenfunction estimators and asymptotic normality of the eigenvalue estimator and estimators of related functionals. As an application, we study an economy where the representative agent is endowed with recursive preferences, allowing for general (nonlinear) consumption and earnings growth dynamics. 

\medskip 


\end{abstract}

\pagenumbering{arabic}

\newpage
\section{Introduction}

In dynamic asset pricing models, stochastic discount factors (SDF) are stochastic processes that are used to infer equilibrium asset prices. \cite{AJ}, \cite{HS2009} and \cite{Hansen2012} show that SDF processes may be decomposed into permanent and transitory components. The {\it permanent component} is a martingale that induces an alternative probability measure which is used to characterize pricing over long investment horizons. The {\it transitory component} is related to the return on a discount bond of (asymptotically) long maturity. \cite{AJ} and \cite{BakshiChabiYo} have found that SDFs must have nontrivial permanent and transitory components in order to explain several salient features of historical returns data. \cite{QL} show that the permanent-transitory decomposition obtains even in very general semimartingale environments, suggesting that the decomposition is a fundamental feature of arbitrage-free asset pricing models.

The pathbreaking work of \cite{HS2009} links SDF decomposition in Markovian environments to a Perron-Frobenius eigenfunction problem. Specifically, \cite{HS2009} show that the permanent and transitory components are constructed from the SDF process, the Perron-Frobenius eigenfunction, and its eigenvalue. The eigenvalue  determines the average yield on long-horizon payoffs and the eigenfunction characterizes dependence of the price of long-horizon payoffs on the Markov state. The probability measure that is relevant for pricing over long investment horizons may be expressed in terms of the eigenfunction and another eigenfunction that is obtained from a time-reversed Perron-Frobenius problem. See \cite{HS2012,HS2014}, \cite{BackusChernovZin}, \cite{BHS}, and \cite{QL,QL:or} for related theoretical developments. 

This paper complements the existing theoretical literature by providing an empirical framework for extracting the permanent and transitory components of SDF processes. We show how to estimate the solution to the Perron-Frobenius eigenfunction problem of \cite{HS2009} from time-series data on state variables and a SDF process. By estimating directly the eigenvalue and eigenfunction, one can reconstruct the time series of the estimated permanent and transitory components and investigate their properties. The methodology also allows one to estimate both the yield and the change of measure which characterize pricing over long investment horizons. This approach is fundamentally different from existing empirical methods for studying the permanent-transitory decomposition, which produce bounds on various moments of the permanent and transitory components as functions of asset returns \citep{AJ,BakshiChabiYo,BC-YG,BC-YG:rec}.\footnote{Recently, \cite{QLN} used a complementary parametric approach to recover the permanent component in a parametric term structure model.}
Although presented in the context of SDF decomposition, the methodology can be applied to more general processes such as the valuation and stochastic growth processes in \cite{HansenHeatonLi}, \cite{HS2009}, and \cite{Hansen2012}.

The empirical framework is {\it nonparametric}, i.e., it does not place any tight parametric restrictions on the law of motion of state variables or the joint distribution of the state variables and the SDF process. This approach is coherent with the existing literature in which bounds on moments of the permanent and transitory components are derived without placing any parametric restrictions on the joint distribution of the SDF, state variables, and asset returns. This approach is also in line with conventional moment-based estimation methods for asset pricing models, such as GMM \citep{Hansen1982} and its various extensions.\footnote{Examples include conditional moment based estimation methodology of \cite{AiChen2003} which has been applied to estimate asset pricing models featuring habits \citep{ChenLudvigson} and recursive preferences \citep{ChenFavilukisLudvigson} and the extended method of moments methodology of \cite{GGR2011} which is particularly relevant for derivative pricing.}

In structural macro-finance models,  SDF processes (and their permanent and transitory components) are determined by both the preferences of economic agents and the dynamics of state variables. Several works have shown that standard preference and state-process specifications struggle to explain salient features of historical returns data. For instance, \cite{BackusChernovZin} find that certain specifications appear unable to generate a SDF whose permanent component is large enough to explain historical return premia without also generating unrealistically large spreads between long- and short-term yields. \cite{BakshiChabiYo} find that historical returns data  support positive covariance between the permanent and transitory components, but that positive association cannot be replicated by workhorse models such as the long-run risks model of \cite{BansalYaron}. 
Our nonparametric methodology  may be used in conjunction with parametric methods to better understand the roles of dynamics and preferences in building models whose permanent and transitory components have empirically realistic properties.

Of course, if state dynamics are treated nonparametrically then certain forward-looking components, such as the continuation value function under recursive preferences, are not available analytically. We therefore introduce nonparametric estimators of the continuation value function in models with \cite{EpsteinZin1989} recursive preferences with elasticity of intertemporal substitution (EIS) equal to unity. This class of preferences is used in prominent empirical work, such as \cite{HansenHeatonLi}, and  may also be interpreted as risk-sensitive preferences as formulated by \cite{HansenSargent1995} (see \cite{Tallarini2000}). We reinterpret the fixed-point problem solved by the value function as a {\it nonlinear} Perron-Frobenius problem. In so doing, we draw connections with the literature on nonlinear Perron-Frobenius theory following \cite{SolowSamuelson}.

As an application, we study an environment similar to that studied by \cite{HansenHeatonLi}. We assume a representative agent with \cite{EpsteinZin1989} preferences with unit elasticity of intertemporal substitution. However, instead of modeling consumption and earnings using a homoskedastic Gaussian VAR as in \cite{HansenHeatonLi}, we model consumption growth and earnings growth as a general (nonlinear) Markov process.  We recover the time series of the SDF process and its permanent and transitory components without assuming any parametric law of motion for the state. The permanent component is large enough to explain historical returns on equities relative to long-term bonds, strongly countercyclical, and highly correlated with the SDF. We also show that the permanent component induces a probability measure that tilts the historical distribution of consumption and earnings growth towards regions of low earnings and consumption growth. To understand better the role of dynamics, we compare properties of the permanent and transitory components extracted nonparametrically with permanent and transitory components implied by two benchmark parametric models for state dynamics, namely a Gaussian VAR and an AR process with stochastic volatility. We find that for certain values of preference parameters, the nonparametric permanent and transitory components can be positively correlated whereas the permanent and transitory components corresponding to the two parametric models are strongly negatively correlated. Overall, our findings suggest that nonlinear dynamics may have a useful role to play in explaining the long end of the term structure.

The sieve (or projection) estimators of the Perron-Frobenius eigenfunction and eigenvalue that we propose draw heavily on earlier work on nonparametric estimation of Markov diffusions by \cite{CHS2000} and \cite{Gobetetal} and are very simple to implement.\footnote{See \cite{DFR}, \cite{DFG}, and \cite{CFR2007} for kernel-based estimation of conditional expectation operators and \cite{LintonLewbelSrisuma2011} and \cite{EHLLS} for kernel-based estimation of marginal utility in nonparametric Euler equation models via eigenfunction methods.} We also propose sieve estimators of the continuation value function in a class of models with recursive preferences. Implementing the sieve value function estimators requires solving a low-dimensional fixed-point problem for which we propose a computationally simple iterative scheme. Both estimation procedures sidestep nonparametric estimation of the state transition distribution.

The main theoretical contributions of the paper are as follows. First, we establish consistency and convergence rates of the Perron-Frobenius eigenfunction estimators and establish asymptotic normality of the eigenvalue estimator and estimators of related functionals. The large-sample properties are established in sufficient generality that they can accommodate SDF process of either of a known functional form or containing components that have been first estimated from data (such as preference parameters and continuation value functions). Second, semiparametric efficiency bounds for the eigenvalue and related functionals are derived for the case in which the SDF is of a known functional form and the estimators are shown to attain their bounds.  Third, this paper is the first to establish consistency and convergence rates for sieve estimators of the continuation value function for a class of models with recursive preferences. Although the analysis is confined to models in which the state vector is observable, the main theoretical results on eigenfunction and continuation value function estimation apply equally to models in which components of the state are latent.

The remainder of the paper is as follows. Section \ref{s:setup} briefly reviews the theoretical framework in \cite{HS2009} and related literature and discusses the scope of the analysis and identification issues. Section \ref{s:est} introduces the estimators of the Perron-Frobenius eigenvalue, eigenfunctions, and related functionals and establishes their large-sample properties. Nonparametric continuation value function estimation is studied in Section \ref{s:recursive}. Section \ref{s:mc} presents a simulation exercise, Section \ref{s:emp} presents the empirical application and Section \ref{s:conc} concludes. Additional results on estimation and inference are deferred to the Appendix. The Supplemental Material contains proofs of all results in the main text and sufficient conditions for some assumptions appearing in the main text. An Online Appendix contains additional results on identification, further simulation results, and additional proofs.

\section{Setup}\label{s:setup}

\subsection{Theoretical framework}

This subsection summarizes the theoretical framework in \cite{AJ}, \cite{HS2009} (HS hereafter), \cite{Hansen2012}, and \cite{BHS} (BHS hereafter).  We  work in discrete time with $T$ denoting the set of non-negative integers. 

In arbitrage-free environments, there is a positive {\it stochastic discount factor} process $M = \{M_t : t \in T\}$ that satisfies:
\begin{equation} \label{e:ee}
 \mb E \Big[ \frac{M_{t+\tau}}{M_t} R_{t,t+\tau} \Big| \mc I_t \Big] = 1
\end{equation}
where $R_{t,t+\tau}$ is the (gross) return on a traded asset over the period from $t$ to $t+\tau$, $\mc I_t$ denotes the information available to all investors at date $t$, and $\mb E[\,\cdot\,]$ denotes expectation with respect to investors' beliefs (see, e.g., \cite{HansenRenault}). 
Throughout this paper, we impose rational expectations by assuming that investors' beliefs agree with the data-generating probability measure.

 \cite{AJ} introduce the \emph{permanent-transitory decomposition}:
\begin{align} \label{e:decomp}
 \frac{M_{t+\tau}}{M_t} & = \frac{M_{t+\tau}^P}{M_t^P} \frac{M_{t+\tau}^T}{M_t^T}\,.
\end{align}
The permanent component $M^P_{t+\tau}/M^P_t$ is a martingale: $\mb E[M_{t+\tau}^P/M_t^P | \mc I_t] = 1$ (almost surely). HS show that the martingale induces an alternative probability measure which is used to characterize pricing over long investment horizons. The transitory component $M_{t+\tau}^T/M_t^T$ is the reciprocal of the  return to holding a discount bond of (asymptotically) long maturity from date $t$ to date $t + \tau$. \cite{AJ} provide conditions under which the permanent and transitory components exist. \cite{QL} show that the decomposition obtains in very general semimartingale environments.

To formally introduce the framework in HS and BHS, consider a probability space $(\Omega,\mcr F, \mb P)$ on which there is a time homogeneous, strictly stationary and ergodic Markov process $X = \{X_t : t \in T\}$ taking values in $\mc X \subseteq \mb R^d$. Let $Q$ denote the stationary distribution of $X$. Let $\{ \mc F_t : t \in T\} \subseteq \mcr F$ be the filtration generated by the histories of $X$. It is assumed that $X_t$ summarizes all information relevant for asset pricing at date $t$. When we consider payoffs depending only on future values of the state and allow trading at intermediate dates, we may assume the SDF process is a {\it positive multiplicative functional} of $X$. That is, $M_t$ is adapted to $\mc F_t$, $M_t > 0$ for each $t \in T$ (almost surely) and:
\[
 \frac{M_{t + \tau}}{M_t} = M_\tau (\theta_t )
\]
with $\theta_t : \Omega \to \Omega$ the time-shift operator given by $X_\tau(\theta_t(\omega)) = X_{t+\tau}(\omega)$ for each $\tau,t \in T$ (see Section 2 of HS). Thus, $M_\tau$ is a function of $X_0,\ldots,X_\tau$ and $M_\tau(\theta_t)$ is the same function applied to $X_t,\ldots,X_{t+\tau}$. In particular, we have:
\begin{equation}
 \frac{M_{t+1}}{M_t} = m(X_t,X_{t+1})   \label{e:sdf:m}
\end{equation}
for some positive function $m$. For convenience, we occasionally refer to $m$ as the SDF.

Given the Markovianity of $X$, we may define a one-period pricing operator $\mb M$ which assigns date-$t$ prices to state-dependent payoffs at date $t+1$. That is, if $\psi(X_{t+1})$ is a payoff at date $t+1$, then its date-$t$ price is given by:
\begin{align} 
 \mb M \psi(x) & = \mb E \Big[ m(X_t,X_{t+1}) \psi(X_{t+1}) \Big| X_t = x \Big] \label{e:m:def} \,.
\end{align}
Pricing operators may be defined analogously for payoff horizons $\tau \geq 1$. The operator $\mb M_\tau$ assigning date-$t$ prices to date-$(t+\tau)$ payoffs $\psi(X_{t+\tau})$ is given by:
\begin{equation} \label{e:mtau:def}
 \mb M_\tau \psi(x) = \mb E \Big[ \frac{M_{t+\tau}}{M_t} \psi(X_{t+\tau}) \Big| X_t = x \Big]\,.
\end{equation}
It follows by Markovianity of the state and the multiplicative functional property of the SDF process that $\mb M_\tau = \mb M^\tau$ (i.e. $\mb M$ applied $\tau$ times) for each $\tau \geq 1$. Therefore, it suffices to study the one-period operator $\mb M$.

HS introduce and study the Perron-Frobenius eigenfunction problem:
\begin{equation} \label{e:pev}
 \mb M \phi(x) = \rho \phi(x)
\end{equation}
where the eigenvalue $\rho$ is a positive scalar and the eigenfunction $\phi$ is positive.\footnote{We say a function is positive (non-negative) if it is positive (non-negative) $Q$-almost everywhere.} Classical, finite-dimensional Perron-Frobenius theory says that a positive matrix has positive right and left eigenvectors corresponding to its spectral radius.\footnote{See, e.g., Theorem 8.2.2 in \cite{HornJohnson}.} The \cite{KreinRutman} theorem and its well-known extensions generalize finite-dimensional Perron-Frobenius theory to infinite-dimensional Banach spaces. To introduce formally the left eigenfunction of $\mb M$ corresponding to $\rho$, we use a time-reversed version of the Perron-Frobenius problem (\ref{e:pev}). Recall that a first-order Markov process seen in reverse time is also a first-order Markov process \cite[p. 4]{Rosenblatt1971}. Define the time-reversed operator 
\[
 \mb M^* \psi(x) = \mb E[ m(X_t,X_{t+1}) \psi(X_t)  | X_{t+1} = x]\,.
\]
In what follows, we will assume that $\mb M$ is a bounded linear operator on the Hilbert space $L^2 = \{ \psi : \mc X \to \mb R$ such that $\int \psi^2 \, \mr d Q < \infty\}$ in which case $\mb M^*$ is defined formally as the adjoint of $\mb M$. The time-reversed Perron-Frobenius problem is:
\begin{equation} \label{e:pev:star}
 \mb M^* \phi^*(x) = \rho \phi^*(x)
\end{equation}
where $\rho$ is the eigenvalue from (\ref{e:pev}) and the eigenfunction $\phi^*$ is positive. 

Given $\rho$ and $\phi$ which solve the Perron-Frobenius problem (\ref{e:pev}), HS define:
\begin{align}  \label{e:pctc}
 \frac{M_{t+\tau}^P}{M_t^P} & = \rho^{-\tau} \frac{M_{t+\tau}}{M_t} \frac{\phi(X_{t+\tau})}{\phi(X_t)} & \frac{M_{t+\tau}^T}{M_t^T} & = \rho^\tau \frac{\phi(X_t)}{\phi(X_{t+\tau})}\,.
\end{align}
It follows from (\ref{e:pev}) that  $\mb E[M_{t+\tau}^P/M_t^P | \mc F_t] = 1$ (almost surely) for each $\tau,t \in T$.
HS show that although there may exist multiple solutions to the Perron-Frobenius problem, only one solution leads to processes $M^P$ and $M^T$ that may be interpreted correctly as permanent and transitory components. Such a solution has a martingale term that induces a change of measure under which $X$ is \emph{stochastically stable}; see Condition 4.1 in BHS for sufficient conditions. Loosely speaking, stochastic stability requires that conditional expectations under the distorted probability measure converge (as the horizon increases) to an unconditional expectation $\wt{\mb E}[\,\cdot\,]$. The expectation $\wt{\mb E}[\,\cdot\,]$ will typically be different from the expectation ${\mb E}[\,\cdot\,]$ associated with the stationary distribution of $X$. 
Under stochastic stability, the one-factor representation:
\begin{equation} \label{e:lrr}
 \lim_{\tau \to \infty}  \rho^{-\tau} \mb M_\tau \psi(x)  = \wt {\mb E} \left[ \frac{\psi(X_t)}{\phi(X_t)} \right] \phi(x) 
\end{equation}
holds for each $\psi$ for which $ \wt {\mb E} [ {\psi(X_t)}/{\phi(X_t)} ]$ is finite (see, e.g., Proposition 7.1 in HS). When a long-run approximation like (\ref{e:lrr}) holds, we may interpret $M_{t+\tau}^P/M_t^P$ and $M_{t+\tau}^T/M_t^T$ from (\ref{e:pctc}) as the permanent and transitory components of the SDF process. Moreover, the scalar $ -\log \rho$ may be interpreted as the \emph{long-run yield}. The long-run approximation (\ref{e:lrr}) also shows that $\phi$ captures state dependence of long-horizon asset prices. The function $\phi^*$ is itself of interest as it will play a role in characterizing the expectation $\wt{\mb E}[\,\cdot\,]$  and will also appear in the asymptotic variance of estimators of $\rho$.

The theoretical framework of HS may be used to characterize properties of the permanent and transitory components {analytically} by solving the Perron-Frobenius eigenfunction problem. Below, we describe an {empirical} framework to estimate the eigenvalue $\rho$ and eigenfunctions $\phi$ and $\phi^*$ from time series data on $X$ and a SDF process.

\subsection{Scope of the analysis}

The Markov state vector $X_t$ is assumed throughout to be observable to the econometrician.  However,  we do not constrain the transition law of $X$ to be of any parametric form. This approach is similar to that taken by \cite{GGR2011}, who also presume the existence of a stationary, time-homogeneous Markov state process that is observable to the econometrician but do not constrain its transition law to be of any parametric form.

We assume the SDF function $m$ is either observable or known up to some parameter which is first estimated from data on $X$ (and possibly asset returns). 

\bigskip

\paragraph{Case 1: SDF is observable} 
Here the functional form of $m$ is known ex ante. For example, consider the CCAPM with time discount parameter $\beta$ and risk aversion parameter $\gamma$ both pre-specified by the researcher. Here we would simply take $m(X_t,X_{t+1}) = \beta G_{t+1}^{-\gamma}$ provided consumption growth $G_{t+1}$ is of the form $G_{t+1} = G(X_t,X_{t+1})$ for some known function $G$. Other structural examples include models with external habits  and models with durables with pre-specified preference parameters.  

\bigskip

\paragraph{Case 2: SDF is estimated}
In this case we assume that $m(X_t,X_{t+1};\alpha_0)$ where the functional form of $m$ is known up to a parameter $\alpha_0$, which could be of several forms:
\begin{itemize}
\item A finite-dimensional vector of preference parameters in structural models (e.g. \cite{HansenSingleton1982} and \cite{HHY}) or risk-premium parameters in reduced-form models (e.g. \cite{GGR2011}).
\item A vector of parameters $\theta_0$ together with a function $h_0$, so $\alpha_0 = (\theta_0,h_0)$. One example is models with \cite{EpsteinZin1989} recursive preferences, where the continuation value function is not known when the transition law of the Markov state is modeled nonparametrically (see \cite{ChenFavilukisLudvigson} and the application in Section \ref{s:emp}). For such models,  $\theta_0$ would consist of  discount, risk aversion, and intertemporal substitution parameters and $h_0$ would be the continuation value function. Another example is \cite{ChenLudvigson} in which $\theta_0$ consists of time discount and homogeneity parameters and $h_0$ is a nonparametric internal or external habit formation component. 
\item We could also take $\alpha_0$ to be $m$ itself, in which case $\hat \alpha$ would be a nonparametric estimate of the SDF. Prominent examples include \cite{BansalViswanathan}, \cite{AitSahaliaLo}, and \cite{RosenbergEngle}.
\end{itemize}

In Case 2 we consider a two-step approach to SDF decomposition. In the first step $\alpha_0$ is estimated from time-series data on the state and possibly also  asset returns. In the second step we plug the first-stage estimator $\hat \alpha$  into the nonparametric procedure to recover $\rho$, $\phi$, $\phi^*$, and related quantities.

\subsection{Identification}\label{s:id}

In this section we present some sufficient conditions that ensure there is a unique solution to the Perron-Frobenius problems (\ref{e:pev}) and (\ref{e:pev:star}). The conditions also ensure that a long-run approximation of the form (\ref{e:lrr}) holds. Therefore, the resulting $M^P$ and $M^T$ constructed from $\rho$ and $\phi$ as in (\ref{e:pctc}) may be interpreted correctly as the permanent and transitory components. For estimation, all that we require is for the conclusions of Proposition \ref{p:id} below hold. Therefore, the following conditions could be replaced by other sets of sufficient conditions.

HS and BHS established very general identification, existence and long-run approximation results using Markov process theory. The operator-theoretic conditions that we use are more restrictive than the conditions in HS and BHS but they are convenient for deriving the large-sample theory that follows. In particular, the conditions ensure certain continuity properties of $\rho$, $\phi$ and $\phi^*$ with respect to perturbations of the operator $\mb M$. Our results are derived for the specific parameter (function) space that is relevant for estimation, whereas the results in HS and BHS apply to a larger class of functions. Connections between our conditions and the conditions in HS and BHS are discussed in detail in the Online Appendix, which also treats separately the issues of identification, existence, and long-run approximation. 

We take the cone of all positive functions in $L^2$ as the parameter space for $\phi$.  Let $\|\cdot\|$ and $\langle \cdot,\cdot \rangle$ denote the $L^2$ norm and inner product. We say that $\mb M$ is bounded if $\|\mb M\| := \sup\{\|\mb M \psi\| : \psi \in L^2 , \|\psi\| = 1 \} < \infty$ and compact if $\mb M$ maps bounded sets into relatively compact sets. Finally, let $Q \otimes Q$ denote the product measure on $\mc X^2$.

\begin{assumption}\label{a:id:0} Let the operators $\mb M$ in display (\ref{e:m:def}) $\mb M_\tau$ in display (\ref{e:mtau:def}) and  satisfy the following:
\begin{enumerate}
\item[(a)] $\mb M$ is a bounded linear operator of the form:
\[
 \mb M \psi(x_t) = \int \mc K_m (x_t,x_{t+1}) \psi (x_{t+1}) \, \mr d Q(x_{t+1})
\]
for some $\mc K_{m}: \mc X^2 \to \mb R$ that is positive ($Q \otimes Q$-almost everywhere)
\item[(b)] $\mb M_\tau$ is compact for some $\tau \geq 1$.
\end{enumerate}
\end{assumption}

{\bf Discussion of assumptions:} Part (a) are mild boundedness and positivity conditions.  If the unconditional density $f(x_t)$ and the transition density $f(x_{t+1}|x_t)$ of $X$ exist, then $\mc K_m$ is of the form:
\[
 \mc K_m (x_t,x_{t+1}) = m(x_t,x_{t+1}) \frac{f(x_{t+1}|x_t)}{f(x_{t+1})}\,.
\]
Part (b) is weaker than requiring $\mb M$ to be compact. A sufficient condition for compactness of $\mb M$ is the Hilbert-Schmidt condition $\int \int \mc K_m(x_t,x_{t+1})^2 \mr d Q(x_t) \mr d Q(x_{t+1}) < \infty$.

To introduce the identification result, let $\sigma(\mb M) \subset \mb C $ denote the spectrum of $\mb M$.\footnote{See, e.g., \cite{DunfordSchwartz}, Chapter VII.3 for definitions.} We say that $\rho$ is \emph{simple} if it has a unique eigenfunction (up to scale) and \emph{isolated} if there exists a neighborhood $N$ of $\rho$ such that $\sigma(\mb M) \cap N = \{\rho\}$. As $\phi$ and $\phi^*$ are defined up to scale, we say that $\phi$ and $\phi^*$ are unique if they are unique up to scale. Normalizing $\phi$ and $\phi^*$ so that $\mb E[\phi(X_t)\phi^*(X_t)] = 1$, we may define a probability measure $\wt Q$ that is absolutely continuous with respect to $Q$ by the change of measure:
\begin{equation} \label{e:rnderiv}
 \frac{\mr d \wt Q}{\mr d Q} = \phi\phi^* \,.
\end{equation}
Finally, let $\wt{\mb E}[\,\cdot\,]$ denote expectation under $\wt Q$, i.e. for any indicator function $\chi$ we have
\begin{equation} \label{e:test}
 \wt{\mb E}[\chi(X_t)] = \mb E[\chi(X_t) \phi(X_t) \phi^*(X_t)]
\end{equation}
where the expectation on the right-hand side is taken under the stationary distribution $Q$.

\begin{proposition}\label{p:id}
Let Assumption \ref{a:id:0} hold. Then:
\begin{enumerate}
\item[(a)] There exists positive functions $\phi,\phi^* \in L^2$ and a positive scalar $\rho$ such that $(\rho,\phi)$ solves (\ref{e:pev}) and $(\rho,\phi^*)$ solves (\ref{e:pev:star}).
\item[(b)] The functions $\phi$ and $\phi^*$ are the unique positive solutions (in $L^2$) to (\ref{e:pev}) and (\ref{e:pev:star}).
\item[(c)] The eigenvalue $\rho$ is simple and isolated and it is the largest eigenvalue of $\mb M$.
\item[(d)] The representation (\ref{e:lrr}) holds for all $\psi \in L^2$ with $\wt{\mb E}[\,\cdot\,]$ defined in (\ref{e:test}).
\end{enumerate}
\end{proposition}

\bigskip

Parts (a) and (b) are existence and identification results, respectively. This is a well-known extension of the classical Krein-Rutman theorem \cite[Theorems V.5.2 and V.6.6]{Schaefer1974}. Recently, similar operator-theoretic results have been applied to study identification in semi/nonparametric Euler equation models (see \cite{EscancianoHoderlein}, \cite{LintonLewbelSrisuma2011}, \cite{Chenetal2012}, and \cite{EHLLS}).
Identification under slightly weaker but related conditions is studied in \cite{Christensen-idpev}.

Part (c)  guarantees that $\rho$ is isolated and simple, which is used extensively in the derivation of the large sample theory. Part (d) says that $\rho$ and $\phi$ are the relevant eigenvalue-eigenfunction pair for constructing the permanent and transitory components and links the expectation $\wt{\mb E}$ to $\phi^*$. Note, in particular, that $\wt {\mb E} [ \psi(X_t)/\phi(X_t) ] = \mb E[\psi(X_t) \phi^*(X_t)]$. Estimating $\phi$ and $\phi^*$ directly allows one to estimate the Radon-Nikodym derivative of $\wt Q$ with respect to $Q$.

\section{Estimation}\label{s:est}

This section introduces the estimators of the Perron-Frobenius eigenvalue $\rho$ and eigenfunctions $\phi$ and $\phi^*$ and presents the large-sample properties of the estimators.

\subsection{Sieve estimation}

We follow \cite{CHS2000} and \cite{Gobetetal} in using a sieve approach in which the infinite-dimensional eigenfunction problem is approximated by a low-dimensional matrix eigenvector problem. Let $b_{k1},\ldots,b_{kk} \in L^2$ be a dictionary of linearly independent basis functions (e.g. polynomials, splines, wavelets, or a Fourier basis) and let $B_k \subset L^2$ denote the linear subspace spanned by $b_{k1},\ldots,b_{kk}$. The sieve dimension $k < \infty$ is a smoothing parameter chosen by the econometrician and should increase with the sample size. 

To describe the approximation, let $\Pi_k : L^2 \to B_k$ denote the orthogonal projection onto $B_k$. Consider the projected eigenfunction problem: 
\begin{equation} \label{e:symprob} 
 (\Pi_k \mathbb M) \phi_k = \rho_k \phi_k 
\end{equation}
where $\rho_k$ is the largest real eigenvalue of $\Pi_k \mathbb M$ and $\phi_k : \mc X \to \mb R$ is its eigenfunction. Under regularity conditions, the problem (\ref{e:symprob}) has a unique solution for all $k$ large enough (see Lemma \ref{lem:exist}). As the function $\phi_k$ belongs to the space $B_k$, we have that  $\phi_k(x) = b^k(x)' c_k$ for a vector $c_k \in \mb R^k$, where $b^k(x) = (b_{k1}(x),\ldots,b_{kk}(x))'$. The eigenfunction problem (\ref{e:symprob}) may be written in matrix form as:
\[
  \mf G^{-1}_k \mf M_k^{\phantom {1}} c_k = \rho_k  c_k 
\]
where the $k \times k$ matrices $\mf G_k$ and $\mf M_k$ are given by:
\begin{eqnarray}
 \mf G_k &= &  \mb E[b^k(X_t) b^k(X_t)'] \label{e:gmat} \\
 \mf M_k &=&  \mb E[b^k(X_t) m(X_t,X_{t+1}) b^k(X_{t+1})'] \label{e:mmat}
\end{eqnarray}
and where $\rho_k$ is the largest real eigenvalue of $\mf G_k^{-1} \mf M_k^{\phantom{1}}$ and $c_k$ is its eigenvector (we assume throughout that $\mf G_k$ is nonsingular). We refer to $\phi_k(x) = b^k(x)'c_k$ as the \emph{approximate solution} for $\phi$. The approximate solution for $\phi^*$ is $\phi_k^*(x) =  b^k(x)'c_k^{*}$ where $c_k^*$ is the eigenvector of $\mf G^{-1}_k \mf M_k^{\prime}$ corresponding to $\rho_k$. Together, $(\rho_k,c_k^{\phantom *},c_k^*)$ solve the generalized eigenvector problem:\
\begin{align} \label{e:gev}
 \mf M_k c_k & = \rho_k \mf G_k c_k &
 c_k^{*\prime} \mf M_k & = \rho_k c_k^{*\prime} \mf G_k
\end{align}
where $\rho_k$ is the largest real generalized eigenvalue of the pair $(\mf M_k,\mf G_k)$. We suppress dependence of $\mf M_k$ and $\mf G_k$ on $k$ hereafter to simplify notation.

To estimate $\rho$, $\phi$ and $\phi^*$, we solve the sample analogue of (\ref{e:gev}), namely:
\begin{align} \label{e:est}
 \wh{\mf M} \hat c & = \hat \rho  \wh{\mf G} \hat c &
 \hat c^{* \prime} \wh{\mf M} & = \hat \rho \hat c^{* \prime} \wh{\mf G} 
\end{align}
where $\wh{\mf M}$ and $\wh{\mf G}$ are defined below and where $\hat \rho$ is the maximum real generalized eigenvalue of the matrix pair $(\wh{\mf M},\wh{\mf G})$. The estimators $\hat \rho$, $\hat c$ and $\hat c^*$ may be computed simultaneously using, for example, the \texttt{eig} routine in Matlab. The estimators of $\phi$ and $\phi^*$ are:
\begin{align*}
 \hat \phi(x) & = b^k(x)' \hat c &
 \hat \phi^*(x) & = b^k(x)' \hat c^*\,.
\end{align*}
Under the regularity conditions below, the eigenvalue $\hat \rho$ and its right- and left-eigenvectors $\hat c$ and $\hat c^*$ will be unique with probability approaching one (see Lemma \ref{lem:exist:hat}). 

Given a time series of data $\{X_0,X_1,\ldots,X_n\}$, a natural estimator of $\mf G$ is: 
\begin{equation}\label{e:ghat}
 \wh{\mf G} = \frac{1}{n} \sum_{t=0}^{n-1} b^k(X_t)b^k(X_t)'\,.
\end{equation}
We consider two possibilities for estimating $\mf M$.

\bigskip

\paragraph{Case 1: SDF is observable} 
First, consider the case in which the function $m(X_t,X_{t+1})$ is specified by the researcher. In this case:
\begin{equation}\label{e:mhat1}
 \wh{\mf M} = \frac{1}{n} \sum_{t=0}^{n-1} b^k(X_t)m(X_t,X_{t+1})b^k(X_{t+1})'\,.
\end{equation}

\bigskip

\paragraph{Case 2: SDF is estimated} 
Now suppose that the SDF is of the form $m(X_t,X_{t+1};\alpha_0)$ where the functional form of $m$ is known up to the parameter $\alpha_0$ which is to be estimated first from the data on $X$ and possibly also asset returns. Let $\hat \alpha$ denote this first-stage estimator. In this case:
\begin{equation} \label{e:mhat2}
 \wh{\mf M} = \frac{1}{n} \sum_{t=0}^{n-1} b^k(X_t)  m(X_t,X_{t+1};\hat \alpha)b^k(X_{t+1})'\,.
\end{equation}

\subsubsection{Other functionals}

Recall that the {\it long-run yield} is $y \equiv -\log \rho$. We may estimate $y$ using:
\begin{equation} \label{e-rhohat}
 \hat y = - \log \hat \rho \,.
\end{equation}

Another functional of interest is the {\it entropy} of the permanent component, namely $L \equiv \log \mb E[M^P_{t+1}/M^P_t] - \mb E[\log(M^P_{t+1}/M^P_t)]$, which is bounded from below by the expected  excess return of any traded asset relative to a discount bond of (asymptotically) long maturity \cite[Proposition 2]{AJ}. Previous empirical work has estimated this bound from data on equity returns and proxies for holding period returns on long-maturity discount bonds (see, e.g., \cite{AJ} and \cite{BakshiChabiYo}). Here we take a complementary approach by assuming the SDF process is identifiable and estimate the entropy of its permanent component directly.

In Markovian environments, the entropy has the simple form $L = \log \rho - \mb E[ \log m(X_t,X_{t+1})]$ (see \cite{Hansen2012} and \cite{BackusChernovZin}). Given $\hat \rho$, a natural estimator of $L$ is:
\begin{equation} \label{e:Lhat1}
 \hat L =  \displaystyle \log \hat \rho - \frac{1}{n} \sum_{t=0}^{n-1} \log m(X_t,X_{t+1})
\end{equation}
in Case 1; in Case 2 we replace $m(X_t,X_{t+1})$ by $m(X_t,X_{t+1};\hat \alpha)$ in (\ref{e:Lhat1}). 
The size of the permanent component may also be measured by other types of statistical discrepancies besides entropy (e.g. Cressie-Read divergences) which may be computed from the time series of the permanent component recovered empirically using $\hat \rho$ and $\hat \phi$. We confine our attention to entropy because the theoretical literature has typically used entropy to measure the size of SDFs and their permanent components over different horizons (see, e.g., \cite{Hansen2012} and \cite{BackusChernovZin}) and for sake of comparison with the empirical literature on bounds.

\subsection{Consistency and convergence rates}

Here we establish consistency of the estimators and derive the convergence rates of the eigenfunction estimators under mild regularity conditions.

\begin{assumption}\label{a:id}
$\mb M$ is bounded and conclusions (a)--(d) of Proposition \ref{p:id} hold.
\end{assumption}

\begin{assumption}\label{a:bias}
$\|\Pi_k \mb M - \mb M\| =o(1)$. 
\end{assumption}

Let $\mf G^{-1/2}$ denote the inverse of the positive definite square root of $\mf G$ and let $\mf I$ denote the $k \times k$ identity matrix. Define the ``orthogonalized'' matrices $\mf M^o = \mf G^{-1/2} \mf M \mf G^{-1/2}$, $\wh{\mf G}^o = \mf G^{-1/2} \wh{\mf G} \mf G^{-1/2}$, and $\wh{\mf M}^o = \mf G^{-1/2} \wh{\mf M} \mf G^{-1/2}$. Let $\|\cdot\|$ also denote the Euclidean norm when applied to vectors and the operator norm (largest singular value) when applied to matrices. Note that $\wh{\mf G}^o$ and $\wh{\mf M}^o$ are a proof device and do not need to be calculated in practice.

\begin{assumption}\label{a:var}
$\|\wh {\mf G}^o - \mf I\| = o_p(1)$  and $\|\wh{\mf M}^o - {\mf M}^o\| = o_p(1)$.
\end{assumption}

{\bf Discussion of assumptions:} Assumption \ref{a:bias} requires that the space $B_k$ be chosen such that it approximates well the range of $\mb M$ (as $k \to \infty$). This assumption also implicitly requires that $\mb M$ is compact, as has been assumed previously in the literature on sieve estimation of eigenfunctions (see, e.g., \cite{Gobetetal}).\footnote{If $\mb M$ is not compact but $\mb M_{\tau}$ is compact for some $\tau \geq 2$, then one can apply the estimators to $\mb M_{\tau}$ in place of $\mb M$ and estimate the solution $(\rho^\tau,\phi)$ to $\mb M_\tau \phi = \rho^\tau \phi$ and similarly for $\phi^*$. Large-sample properties of the estimators of $\rho^\tau$,  $\phi$ and $\phi^*$ would then follow directly from Theorems \ref{t:rate}--\ref{t:asydist:2b}.} Assumption \ref{a:var} ensures that the sampling error in estimating ${\mf G}^{-1} \mf M$ vanishes asymptotically. This condition implicitly restricts the maximum rate at which $k$ can grow with $n$. Appendix \ref{ax:est:mat} in the Supplementary Material presents some sufficient conditions for Assumption \ref{a:var}.

Before presenting the main result on convergence rates, we first we introduce sequences of constants that bound the approximation bias  and sampling error. As eigenfunctions are only normalized up to scale, impose the normalizations $\|\phi\| = 1$ and $\|\phi^*\| = 1$. Define:
\begin{equation} \label{e:deltas}
 \delta_k^{\phantom *} = \| \Pi_k \phi - \phi\| \quad \mbox{and} \quad \delta_k^* = \| \Pi_k \phi^* - \phi^*\|\,.
\end{equation}
Here $\delta_k^{\phantom *}$ and $\delta_k^*$ measure the bias incurred by approximating $\phi$ and $\phi^*$ by elements of $B_k$.  Bounds for $\delta_k^{\phantom *}$ and $\delta_k^*$ are available for commonly used bases when $\phi$ and $\phi^*$ belong a H\"older, Sobolev or Besov class (see, e.g., \cite{Chen2007}).
Let $\tilde c_k = \mf G^{1/2} c_k$ and $\tilde c_k^* = \mf G^{1/2} c_k^*$ and normalize $c_k$ and $c_k^*$ so that $\|\tilde c_k\| = \|\tilde c_k^*\| = 1$ (this is  equivalent to $\|\phi_k\| =\|\phi_k^*\| = 1$). Under Assumption \ref{a:var}, we may choose positive sequences $\eta_{n,k}$ and $\eta_{n,k}^*$ which are both $o(1)$, so that:
\begin{equation} \label{e:etas}
 \|((\wh{\mf G}^o)^{-1}\wh{\mf M}^o - {\mf M}^o) \tilde c_k^{\phantom *}\| = O_p(\eta_{n,k}^{\phantom *}) \quad \mbox{and} \quad \|((\wh{\mf G}^o)^{-1}\wh{\mf M}^{o\prime} - {\mf M}^{o\prime}) \tilde c_k^*\| = O_p(\eta_{n,k}^{ *})\,.
\end{equation}
Appendix \ref{ax:est:mat} presents bounds on $\eta_{n,k}$ and $\eta_{n,k}^*$.

\begin{theorem}\label{t:rate}
Let Assumptions \ref{a:id}--\ref{a:var} hold. Then: 
\begin{enumerate} 
\item[(a)] $|\hat \rho - \rho| = O_p( \delta_k + \eta_{n,k})$
\item[(b)] $\|\hat \phi - \phi\| = O_p( \delta_k + \eta_{n,k})$
\item[(c)] $\|\hat \phi^* - \phi^*\| = O_p( \delta_k^* + \eta_{n,k}^*)$
\end{enumerate}
where $\delta_k^{\phantom *}$ and $\delta_k^*$ are defined in (\ref{e:deltas}) and $\eta_{n,k}^{\phantom *}$ and $\eta_{n,k}^*$ are defined in (\ref{e:etas}). The convergence rates for $\hat \phi$ and $\hat \phi^*$ should be understood to hold under the scale normalizations $\|\phi\| = 1$, $\|\hat \phi\| = 1$, $\| \phi^* \| = 1$ and $\|\hat \phi^* \| = 1$ and sign normalizations $\langle \phi,\hat \phi \rangle \geq 0$ and $\langle \phi^*, \hat \phi^* \rangle \geq 0$.
\end{theorem}

It is worth noting that Theorem \ref{t:rate} holds for $\hat \rho$, $\hat \phi$ and $\hat \phi^*$ calculated from any estimators $\wh{\mf G}$ and $\wh{\mf M}$ that satisfy Assumption \ref{a:var}. Indeed, Theorem \ref{t:rate} is sufficiently general that it applies to models with latent state vectors without modification: all that is required is that one can construct estimators of $\mf G$ and $\mf M$ that satisfy Assumption \ref{a:var}.

Theorem \ref{t:rate} displays the usual bias-variance tradeoff encountered in nonparametric estimation. The bias terms $\delta_k^{\phantom *}$ and $\delta_k^*$ will be decreasing in $k$ (since $\phi$ and $\phi^*$ are approximated over increasingly rich subspaces as $k$ increases). On the other hand, the variance terms $\eta_{n,k}^{\phantom *}$ and $\eta_{n,k}^*$ will typically be increasing in $k$ (larger matrices) and decreasing in $n$ (more data). Choosing $k$ to balance the bias and variance terms will yield the best convergence rate. As an illustration, we now establish the convergence rates of $\hat \phi$ and $\hat \phi^*$ in Case 1, where $\wh{\mf G}$ and $\wh{\mf M}$ are as in (\ref{e:ghat}) and (\ref{e:mhat1}), under standard conditions from the statistics literature on optimal convergence rates. Although the following conditions are not necessarily appropriate in an asset pricing context, the result is informative about the convergence properties of $\hat \phi$ and $\hat \phi^*$. Let $W^p = \{ f \in L^2 : \sum_{|a| \leq p} \|D^a f\| < \infty\}$ with $D^a f = \frac{\partial^{a_1 + \ldots + a_d}}{\partial^{a_1}x_1 \cdots \partial^{a_d}x_d} f$ and $|a| = a_1 + \ldots + a_d$ denote a Sobolev space of smoothness $p \in \mb N$ equipped with the norm $\|f\|_{W^p} = \sum_{|a| \leq p} \|D^a f\|$.

\begin{corollary}\label{c:rate}
Let Assumption \ref{a:id} and the following conditions hold: (i) $\mc X \subset \mb R^d$ is compact and rectangular; (ii) $Q$ has a continuous density bounded away from zero; (iii) $\phi,\phi^* \in W^p$ and $\mb M$ is a bounded linear operator from $L^2$ into $W^{\bar p}$ for some $p \geq \bar p > 0$; (iv)  $B_k$ is spanned by tensor-product B-splines of order $\nu > p$ with equally spaced interior knots; (v) $\mb E[m(X_0,X_1)^r] < \infty$ for some $r > 2$; (vi) $k^{2+2/r}/n = o(1)$; (vii) $X$ is exponentially rho-mixing. 
Then: Assumptions \ref{a:bias} and \ref{a:var} hold and we may take $\delta_k,\delta_k^*=O(k^{-p/d})$, and $\eta_{n,k},\eta_{n,k}^*=O(k^{(r+2)/2r}/\sqrt n)$. Choosing $k \asymp n^{\frac{rd}{2rp+(2+r)d}}$ yields:
\begin{align*}
 \|\hat \phi - \phi\| & =  O_p(n^{-\frac{rp}{2rp+(2+r)d}}) & \|\hat \phi^* - \phi^*\| & =  O_p(n^{-\frac{rp}{2rp+(2+r)d}})\,.
\end{align*}
If $m$ is bounded, the rates become $n^{-p/(2p+d)}$ which is the optimal $L^2$-norm rate for nonparametric regression estimators when the regression function belongs to $W^p$. 
\end{corollary}

\newpage

Sieve methods may also be used to numerically compute $\rho$, $\phi$, and $\phi^*$  in models for which analytical solutions are unavailable. For such models, the matrices $\mf M$ and $\mf G$ may be computed directly (e.g. via simulation or numerical integration) and $\rho_k$, $\phi_k$ and $\phi_k^*$ can be obtained by solving (\ref{e:gev}). Lemma \ref{lem:bias} gives the rates $|\rho_k - \rho| = O(\delta_k)$, $\|\phi_k - \phi\| = O(\delta_k^{\phantom *})$, and $\|\phi_k^* - \phi^*\| = O(\delta_k^*)$.

We close this subsection with a remark relating $\delta_k$ and $\delta_k^*$ under an additional condition on the sieve basis $B_k$. Assumption \ref{a:bias} implies that $\mb M$ is compact. Therefore, $\mb M$ has a singular value decomposition $\{(\mu_n, \varphi_n, g_n) : n \in \mb N\}$ where $\{\mu_n : n \in \mb N\}$ are the nonzero singular values of $\mb M$ arranged in non-increasing order (i.e. $\mu_n \geq \mu_{n+1} \searrow 0$) and $\{\varphi_n : n \in \mb N\}$ and $\{g_n : n \in \mb N\}$ are orthonormal bases for $L^2$ with $\mb M \varphi_n = \mu_n g_n$ and $\mb M^* g_n = \mu_n \varphi_n$ for each $n \in \mb N$ (see, for example, Chapter 15.4 in \cite{Kress}).

\begin{remark}\label{rmk:delta}
Let Assumption \ref{a:bias} hold and let $B_k$ span the linear subspaces generated by $\{ \varphi_n : 1 \leq n \leq k\}$ and $\{g_n : 1 \leq n \leq k\}$. Then: $\delta_k^{\phantom *}$ and $\delta_k^*$ are both $O(\mu_{k+1})$.
\end{remark} 

For example, if $X$ is a scalar Gaussian AR(1), $m(X_t,X_{t+1})$ is exponentially affine in $(X_t,X_{t+1})$, and the basis functions are Hermite polynomials then $\delta_k^{\phantom *}$ and $\delta_k^*$ are $O(e^{-ck})$ for some $c > 0$. Similar spanning assumptions are often made in the literature on sieve estimation of nonparametric instrumental variables models (see, e.g., \cite{BCK}). 

\subsection{Asymptotic normality}

In this section we establish the asymptotic normality of $\hat \rho$. The semiparametric efficiency bound in Case 1 is also derived and $\hat \rho$ is shown to be efficient in this case. Related results on asymptotic normality and semiparametric efficiency of the estimator of the entropy of the permanent component are presented in Appendix \ref{ax:inf}.

\subsubsection{Asymptotic normality in Case 1}

To establish asymptotic normality of $\hat \rho$, we derive the representation:
\begin{equation}\label{e:ale:1}
 \sqrt n (\hat \rho - \rho) = \frac{1}{\sqrt n} \sum_{t=0}^{n-1} \psi_\rho(X_t,X_{t+1})  + o_p(1)
\end{equation}
where the influence function $\psi_\rho$ is given by: 
\begin{equation} \label{e:inf:def}
 \psi_\rho(x_0,x_1) = \phi^*(x_0) m(x_0,x_1) \phi(x_1) - \rho \phi^*(x_0) \phi(x_0)
\end{equation}
with $\phi$ and $\phi^*$ normalized so that $\|\phi\| = 1$ and $\langle \phi, \phi^* \rangle = 1$.
The process $\{\psi_\rho(X_t,X_{t+1}) : t \in T\}$ is a martingale difference sequence (relative to the filtration $\{\mc F_t : t \in T\}$). Therefore, the asymptotic distribution of $\hat \rho$ follows from (\ref{e:ale:1}) by a central limit theorem for martingale differences. To formalize this argument, we make the following assumption.

\begin{assumption}\label{a:asydist}
Let the following hold:
\begin{enumerate}
\item[(a)] $\delta_k = o(n^{-1/2})$ and $\delta_k^* = o(n^{-1/2})$
\item[(b)] $ \|\wh {\mf G}^o - \mf I\|  = o_p(n^{-1/4})$ and $\|\wh{\mf M}^o - {\mf M}^o\| = o_p(n^{-1/4})$
\item[(c)] $\mb E[(\phi^*(X_t)m(X_t,X_{t+1})\phi(X_{t+1}))^2] < \infty$.
\end{enumerate}
\end{assumption}

{\bf Discussion of assumptions:}  Assumption \ref{a:asydist}(a) is an under-smoothing condition which ensures that the leading bias terms $\sqrt n(\rho_k - \rho)$ and higher-order bias terms involving $\phi_k^{\phantom *}$, $\phi_k^*$, and $\rho_k$ are asymptotically negligible. Assumption 3.4(b) ensures that $\wh{\mf G}$ and $\wh{\mf M}$ converge fast enough that $\sqrt n(\hat \rho - \rho_k)$ may be written in an asymptotically linear form similar to (\ref{e:ale:1})-(\ref{e:inf:def}) but with $\phi_k^{\phantom *}$, $\phi_k^*$, and $\rho_k$ in place of $\phi$, $\phi^*$, and $\rho$. This result, in view of the asymptotic negligibility of the leading and higher-order bias terms under Assumption \ref{a:asydist}(a), leads to the representation (\ref{e:ale:1}). Sufficient conditions for Assumption \ref{a:asydist}(b) are presented in Appendix \ref{ax:est:mat}. Assumption 3.4(c) allows a CLT for square-integrable martingale differences to be applied to the martingale difference sequence $\{\psi_\rho(X_t,X_{t+1}) : t \in T\}$. 
Let $V_\rho = \mb E[\psi_\rho(X_0,X_1)^2]$.

\begin{theorem}\label{t:asydist:1}
Let Assumptions \ref{a:id}--\ref{a:asydist} hold. Then: the asymptotic linear expansion (\ref{e:ale:1}) holds and $\sqrt n (\hat \rho - \rho) \to_d N(0,V_\rho)$.
\end{theorem}

It follows directly from Theorem \ref{t:asydist:1} that $\sqrt n(\hat y - y) \to_d N(0,\rho^{-2} V_\rho)$.

We conclude by deriving the semiparametric efficiency bounds for Case 1. We require a further technical condition to characterize the tangent space (see Appendix \ref{ax:inf}).

\begin{theorem}\label{t:eff}
Let Assumptions \ref{a:id}--\ref{a:asydist} and \ref{a:eff} hold. Then: the semiparametric efficiency bound for $\rho$ is $V_\rho$ and $\hat \rho$ is semiparametrically efficient.
\end{theorem}

\subsubsection{Asymptotic normality in Case 2}

For Case 2, we obtain the following expansion (under regularity conditions):
\begin{equation}\label{e:ale:2}
 \sqrt n (\hat \rho - \rho) = \frac{1}{\sqrt n} \sum_{t=0}^{n-1} \Big( \psi_\rho(X_t,X_{t+1}) + \psi_{\alpha,k}(X_t,X_{t+1}) \Big) + o_p(1)
\end{equation}
where $\psi_\rho$ is from display (\ref{e:inf:def}) with $m(x_0,x_1) = m(x_0,x_1;\alpha_0)$ and where:
\begin{equation} \label{e:inf:alpha:def}
 \psi_{\alpha,k}(x_0,x_1) = \phi^*_k(x_0) \big( m(x_0,x_1;\hat \alpha) - m(x_0,x_1;\alpha_0) ) \phi_k^{\phantom *}(x_1) \,.
\end{equation}
The expansion (\ref{e:ale:2}) shows that the asymptotic distribution of $\hat \rho$ and related functionals will depend on the properties of the first stage estimator $\hat \alpha$. The following regularity conditions are deliberately general so as to accommodate a wide class of estimators.

We first suppose that $\alpha_0$ is a finite-dimensional parameter and the plug-in estimator $\hat \alpha$ is root-$n$ consistent and asymptotically normal.  Let $\psi_{\rho,t} = \psi_\rho(X_t,X_{t+1})$.

\begin{assumption}\label{a:parametric}
Let the following hold:
\begin{enumerate}
\item[(a)]  $\sqrt n (\hat \alpha - \alpha_0) = \frac{1}{\sqrt n} \sum_{t=0}^{n-1} \psi_{\alpha,t} + o_p(1)$ for some $\mb R^{{d_\alpha}}$-valued random process $\{\psi_{\alpha,t} : t \in T\}$
\item[(b)] $\frac{1}{\sqrt n} \sum_{t=0}^{n-1} (\psi_{\rho,t}^{\phantom \prime},\psi_{\alpha,t}')' \to_d N(0,V_{[\mr{2a}]}) $ for some finite matrix $V_{[2a]}$
\item[(c)] $m(x_0,x_1;\alpha)$ is continuously differentiable in $\alpha$ on a neighborhood $N$ of $\alpha_0$ for all $(x_0,x_1) \in \mc X^2$ and there exists some function $\bar m : \mc X^2 \to \mb R$ with $\mb E[\bar m(X_t,X_{t+1})^s ] < \infty$ for some $s\geq 2$ such that:
\[
 \sup_{\alpha \in N} \left\| \frac{\partial m(x_0,x_1;\alpha)}{\partial \alpha} \right\| \leq \bar m(x_0,x_1) \quad \mbox{for all $(x_0,x_1) \in \mc X^2$.}
\]
\item[(d)] $\mb E[(\phi(X_t)\phi^*(X_t))^{s/(s-1)}] < \infty$.
\end{enumerate}
\end{assumption}

Let $h_{[\mr{2a}]} = (1\,,\, \mb E[\phi^*(X_t) \phi(X_{t+1})\frac{\partial m(X_t,X_{t+1};\alpha_0)}{\partial \alpha'}])'$ and define $V_\rho^{[\mr{2a}]} = h_{[\mr{2a}]}'V_{[\mr{2a}]}^{\phantom \prime} h_{[\mr{2a}]}^{\phantom \prime}$.

\begin{theorem}\label{t:asydist:2a}
Let Assumptions \ref{a:id}--\ref{a:parametric} hold. Then: $\sqrt n (\hat \rho - \rho) \to_d  N(0,V_\rho^{[\mr{2a}]})$.
\end{theorem}

We now suppose that $\alpha_0$ is an  infinite-dimensional parameter. The parameter space is $\mc A \subseteq \mb A$ (a Banach space) equipped with some norm $\|\cdot\|_{\mc A}$. This includes the case in which (1) $\alpha$ is a function, i.e. $\alpha = h$ with $\mb A = \mb H$ a function space, and  (2)  $\alpha$ consists of both finite-dimensional and function parts, i.e. $\alpha = (\theta,h)$ with $\mb A = \Theta \times \mb H$ with $\Theta \subseteq \mb R^{\dim (\theta)} $. For example, under recursive preferences the vector $\theta$ could consist of discount, risk-aversion and EIS parameters and $h$ could be the continuation value function.

Inference in this case involves the (typically nonlinear) functional $\ell : \mc A \to \mb R$, given by: 
\[
 \ell(\alpha) = \mb E[\phi^*(X_t) \phi(X_{t+1}) m(X_t,X_{t+1};\alpha)]\,.
\]
We focus on the case in which $\ell(\alpha_0)$ is root-$n$ estimable.
We say the functional $\ell : \mc A \to \mb R$ is \emph{pathwise differentiable} at $\alpha_0$ if  $\lim_{\tau \to 0^+} (\ell(\alpha_0 + \tau[\alpha - \alpha_0]) - \ell(\alpha_0))/\tau$ exists for every fixed $\alpha \in \mc A$. If so, we denote the derivative by $\dot \ell_{\alpha_0}[\alpha - \alpha_0]$. Define $\mc G = \{ g_\alpha : \alpha \in \mc A\}$ where $g_\alpha(x_t,x_{t+1}) =\phi^*(x_t) \phi(x_{t+1}) (m(x_t,x_{t+1};\alpha) - m(x_t,x_{t+1};\alpha_0))$. Let $\mc Z_n$ denote the centered empirical process on $\mc G$. We say $\mc G$ is \emph{Donsker} if $\sum_{t \in \mb Z} \mr{Cov}(g(X_0,X_1),g(X_t,X_{t+1}))$ is absolutely convergent over $\mc G$ to a non-negative quadratic form $\mb K(g,g)$ and there exists a sequence of Gaussian processes $\mc Z^{(n)}$ indexed by $\mc G$ with covariance function $\mb K$ and a.s. uniformly continuous sample paths such that $\sup_{g \in G} | \mc Z_n (g) - \mc Z^{(n)}(g)| \to_p 0$ as $n \to \infty$ (see \cite{DoukhanMassartRio}). Finally, let $\|\cdot\|_p$ denote the $L^p$ norm $\|\psi\|_p = (\int |\psi|^p \, \mr d Q )^{1/p}$ for any $1 \leq p < \infty$ (note that $\|\cdot\|_2 = \|\cdot\|$ in our earlier notation).

\begin{assumption}\label{a:nonpara}
Let the following hold:
\begin{enumerate}
\item[(a)] $\mc G$ is Donsker
\item[(b)] $\ell$ is pathwise differentiable at $\alpha_0$ and $| \ell(\alpha) - \ell(\alpha_0) - \dot \ell_{\alpha_0} [\alpha - \alpha_0]| = O(\|\alpha - \alpha_0\|^2_{\mc A})$
\item[(c)] $\sqrt n \dot \ell_{\alpha_0}[\hat \alpha - \alpha_0] = \frac{1}{\sqrt n } \sum_{t=0}^{n-1} \psi_{\ell,t} + o_p(1)$ for some $\mb R$-valued random process $\{ \psi_{\ell,t} : t \in T\}$, $\|\hat \alpha - \alpha_0\|_{\mc A} = o_p(n^{-1/4})$, and $\mb K(g_{\hat \alpha},g_{\hat \alpha}) = o_p(1)$
\item[(d)] $\frac{1}{\sqrt n} \sum_{t=0}^{n-1} (\psi_{\rho,t},\psi_{\ell,t})'  \to_d N(0,V_{[\mr{2b}]})$ for some finite matrix $V_{[2b]}$
\item[(e)] $\mb E[ \sup_{\alpha \in \mc A} m(X_t,X_{t+1};\alpha)^s] < \infty$ and either $\|\phi_k\|_{2s/(s-2)} = O(1)$ and $\|\phi^*\|_{2s/(s-2)} < \infty$ or $\|\phi_k^*\|_{2s/(s-2)} = O(1)$ and $\|\phi\|_{2s/(s-2)} < \infty$ holds for some $s > 2$.

\end{enumerate}
\end{assumption}

{\bf Discussion of assumptions:}  Sufficient conditions for the class $\mc G$ to be Donsker are well known (see, e.g., \cite{DoukhanMassartRio}). 
Parts (b) and (c) are standard conditions for inference in nonlinear semiparametric models (see, e.g., Theorem 4.3 in \cite{Chen2007}). Part (d) is a mild CLT condition and part (e) is a mild higher-than-second-moments condition.

For the following theorem, let $h_{[\mr{2b}]} = (1,1)'$ and define $V_\rho^{[\mr{2b}]} = h_{[\mr{2b}]}'V_{[\mr{2b}]}^{\phantom \prime} h_{[\mr{2b}]}^{\phantom \prime}$.

\begin{theorem}\label{t:asydist:2b}
Let Assumptions \ref{a:id}--\ref{a:asydist} and \ref{a:nonpara} hold. Then: $\sqrt n (\hat \rho - \rho) \to_d  N(0,V_\rho^{[\mr{2b}]})$.
\end{theorem}

\section{Value function recursion as a nonlinear Perron-Fro-\break benius problem} \label{s:recursive}

This section describes how to estimate nonparametrically the continuation value function and SDF in a class of models with recursive preferences by solving a \emph{nonlinear} Perron-Frobenius eigenfunction problem. We focus on models in which a representative agent has \cite{EpsteinZin1989} recursive preferences with unit elasticity of intertemporal substitution (EIS). This class of preferences may also be interpreted as risk-sensitive preferences as formulated by \cite{HansenSargent1995} (see \cite{Tallarini2000}).
After describing the setup, we present some regularity conditions for local identification. We then introduce the estimators and derive their large-sample properties.

\subsection{Setup}

Under Epstein-Zin preferences, the date-$t$ utility of the representative agent is defined via the recursion:
\[
 V_t = \left\{ (1-\beta) C_t^{1-\theta} + \beta \mb E[V_{t+1}^{1-\gamma}|\mathcal F_t]^{\frac{1-\theta}{1-\gamma}} \right\}^{\frac{1}{1-\theta}}
\]
where $C_t$ is date-$t$ consumption, $1/\theta$ is the EIS, $\beta \in (0,1)$ is the time discount parameter, and $\gamma > 1$ is the relative risk aversion parameter. We maintain the assumption of a Markov state process $X$. Let consumption growth, namely $G_{t+1} =  C_{t+1}/C_t$, be a measurable function of $(X_t,X_{t+1})$. \cite{HansenHeatonLi} show that the scaled continuation value $V_t/C_t$ may be written as $V(X_t) $ where:
\begin{equation} \label{e:ezrecur}
 V(X_t) = \bigg\{ (1-\beta)  + \beta \mb E\Big[  \left(V(X_{t+1})G_{t+1}\right)^{1-\gamma} \Big|X_t \Big]^{\frac{1-\theta}{1-\gamma}} \bigg\}^{\frac{1}{1-\theta}} \,.
\end{equation}
With unit EIS (i.e. $\theta = 1$) the fixed point equation (\ref{e:ezrecur}) reduces to:
\begin{equation} \label{e:fp:prelim}
 v(X_t) =\frac{\beta}{1-\gamma} \log \mb E \Big[e^{  (1-\gamma) (v(X_{t+1}) +\log G_{t+1} ) } \Big|X_t \Big]
\end{equation}
with $v(x)= \log V(x)$. Analytical solutions for $v$ are typically only available when the conditional moment generating function of the Markov state is exponentially affine and $\log G_{t+1}$ is affine in $(X_t,X_{t+1})$. Assuming frictionless markets, the SDF is:
\begin{align} \label{e:efn:nl:0}
 \frac{M_{t+1}}{M_t} & = \beta G_{t+1}^{-1} \frac{(V_{t+1})^{1-\gamma}}{\mb E[(V_{t+1})^{1-\gamma}|X_t]}\,.
\end{align}
The dynamics of $X$ determine both the value function and the conditional expectation in the denominator of the SDF. The value function and conditional expectation are therefore unknown when the dynamics of $X$ are treated nonparametrically.

Consider the following reformulation of the fixed-point problem in display (\ref{e:fp:prelim}) as a nonlinear Perron-Frobenius problem. Setting $h(x) = \exp(\frac{1-\gamma}{\beta} v(x))$ and rearranging, we obtain the fixed-point equation $\mb T h = h
$, where:
\begin{align*}
 \mb T \psi (x) & = \mb E\Big[ G_{t+1}^{1-\gamma} \big|\psi(X_{t+1}) \big|^\beta \Big| X_t = x \Big] \,.
\end{align*}
As we seek a positive solution, taking an absolute value inside the conditional expectation in the preceding display does not change the fixed point. Dividing  $\mb T h = h$ by $\|h\|$ and using the fact that $\mb T$ is positive homogeneous of degree $\beta$, we obtain the nonlinear Perron-Frobenius problem:
\begin{equation} \label{e:efn:nl}
 \mb T \chi (x) = \lambda \chi (x)
\end{equation}
where $\chi(x) = h(x)/\|h\|$ is a positive eigenfunction of $\mb T$ and $\lambda = \|h\|^{1-\beta}$ is its eigenvalue. Throughout this section we normalize the eigenfunction $\chi$ to have unit norm. Unlike with linear operators, here changing the scaling of $h$ changes the corresponding eigenvalue: $c \chi$ is a positive eigenfunction of $\mb T$ with eigenvalue $c^{\beta-1} \lambda$ for any $c > 0$.

Reformulation of the recursion as a nonlinear Perron-Frobenius problem also leads to a convenient representation of the SDF. Rewriting the SDF from display (\ref{e:efn:nl:0}) in terms of $h$, we obtain:
\begin{align*} 
 \frac{M_{t+1}}{M_t} & = \beta G_{t+1}^{-\gamma} \frac{(h(X_{t+1}))^\beta}{\mb T h(X_t)}\,.
\end{align*}
Rescaling by $\|h\|$ and using (\ref{e:efn:nl}) yields:
\begin{align} \label{e:rec:sdf}
 \frac{M_{t+1}}{M_t} 
 & = \frac{\beta}{\lambda} G_{t+1}^{-\gamma} \frac{ \big(\chi(X_{t+1}) \big)^\beta}{  \chi(X_t) } \,.
\end{align} 
In what follows, we show how to estimate $\chi$ and $\lambda$ from time-series data on $X$. The estimates $\hat \chi$ and $\hat \lambda$ can be plugged into (\ref{e:rec:sdf})  to obtain nonparametric estimates of the SDF process (i.e. without assuming a parametric law of motion for $X$).

\subsection{Local identification}

In this section we provide sufficient conditions for local identification of the fixed point $h$ and its corresponding eigenfunction $\chi$. We establish the results for the parameter (function) space $L^2$ because it is convenient for sieve estimation. One cannot establish (global) identification using contraction mapping arguments because $\mb T$ is not a contraction on $L^2$.\footnote{Suppose that $\mb T$ has a positive fixed point $h \in L^2$. The function $\bar h \equiv 0$ is also a fixed point. Therefore, $\mb T$ is not a contraction on $L^2$ (else the Banach contraction mapping theorem would yield a unique fixed point).} Some of the regularity conditions we require for estimation are sufficient for $\mb T$ to satisfy a local ergodicity property which, in turn, is sufficient for local identification.

To describe the local ergodicity property, first choose some (nonzero) function $\psi \in L^2$ and set $\chi_1(\psi) = \psi$. Then consider the sequence defined iteratively by:
\[
 {\chi_{n+1}(\psi)} = \frac{\mb T \chi_n(\psi) }{\| \mb T \chi_n(\psi)\|}
\]
for $n \geq 1$. Proposition \ref{p:nl} below shows that the sequence $\chi_n(\psi)$ converges to $\chi$ for any starting value $\psi$ in a suitably defined region. 
This is similar to various ``stability'' results in the literature on balanced growth following \cite{SolowSamuelson}.\footnote{The literature on infinite-dimensional Perron-Frobenius theory has typically dealt with function spaces for which cone of non-negative functions has nonempty interior (see \cite{Krause} for a recent overview). The non-negative cone in $L^2$ has empty interior.  If $\mc X$ is bounded then these previous results may be used to derive (global) identification conditions in the space $C(\mc X)$. However, bounded support seems inappropriate for common choices of state variable, such as consumption growth and dividend growth.} There, $\mb T : \mb R^K \to \mb R^K$ is a homogeneous input-output system, $\chi_n \in \mb R^K$ lists the proportions of commodities in the economy in period $n$, and $\mb T \chi_n$ is normalized by its $\ell^1$ norm so that $\chi_{n+1} := \mb T \chi_n/\|\mb T \chi_n\|_{\ell^1}$ lists the proportions in period $n+1$. ``Stability'' concerns convergence of the sequence $\chi_n$ to a positive eigenvector $\chi$ of $\mb T$ (representing balanced growth proportions).

Write $\mb T = \mb G \mb F$ where $\mb F$ is the nonlinear operator $\mb F \psi(x) = |\psi(x)|^\beta$ and $\mb G$ is the linear operator:
\[
 \mb G \psi(x) =  \mb E\Big[ G_{t+1}^{1-\gamma} \psi(X_{t+1})  \Big| X_t = x \Big] \,.
\]
The operator $\mb T$ is bounded (respectively, compact) on $L^2$  whenever $\mb G$ is bounded (compact) on $L^2$ (see Chapter 5 of \cite{KZPS}). We say that $\mb G$ is {\it positive} if $\mb G \psi$ is positive for any non-negative $\psi \in L^2$ that is not identically zero. Positivity of $\mb G$ ensures that the sequence $\chi_n(\psi)$ is well defined and that any nonzero fixed point of $\mb T$ is positive.
 We say that $\mb T$ is \emph{Fr\'echet differentiable} at $h$ if there exists a bounded linear operator $\mb D_h : L^2 \to L^2$ such that:
\[
 \| \mb T(h + \psi) - \mb T h - \mb D_h \psi\| = o( \|\psi\| ) \quad \mbox{ as $\|\psi\| \to 0$.}
\]
If it exists, the Fr\'echet derivative $\mb D_h$ of  $\mb T$ is given by:
\[
 \mb D_h \psi(x) = \mb E \Big[ \beta G_{t+1}^{1-\gamma} h(X_t)^{\beta-1} \psi(X_{t+1}) \Big| X_t = x \Big] \,.
\]
Let $r(\mb D_h)$ denote the spectral radius of $\mb D_h$.

\begin{proposition}\label{p:nl}
Let $\mb G$ be positive and bounded and let $\mb T$ be Fr\'echet differentiable at $h$ with $r(\mb D_h) < 1$. Then: there exists finite positive constants $C,c$ and a neighborhood $N$ of $\chi$ such that:
\[
 \| \chi_{n+1}(\psi) - \chi\| \leq C e^{-cn}
\]
for any initial point $\psi$ in the cone $\{ a N : a \in \mb R, a \neq 0\}$. 
\end{proposition}

We say that $\chi$ is {\it locally identified} if there exists a neighborhood $N$ of $\chi$ such that $\chi$ is the unique eigenfunction of $\mb T$ belonging to $N \cap S_1$ where $S_1$ denotes the unit sphere in $L^2$ (recall we normalize eigenfunctions of $\mb T$ to have unit norm). Similarly, we say that $h$ is locally identified if $h$ is the unique fixed point of $\mb T$ belonging to some neighborhood $N'$ of $h$. To see why local identification follows from Proposition \ref{p:nl}, suppose $\bar \chi$ is a positive eigenfunction of $\mb T$ belonging to $N \cap S_1$. Proposition \ref{p:nl} implies that $\| \chi_{n+1}(\bar \chi) - \chi \| = \|\bar \chi - \chi\| \leq C e^{-cn}$ for each $n$, hence $\bar \chi = \chi$. Local identification of $h$ follows similarly.

\begin{corollary}\label{c:local-id}
$h$ and $\chi$ are locally identified under the conditions of Proposition \ref{p:nl}.
\end{corollary}

In fact, local identification of $\chi$ and positive homogeneity of $\mb T$ imply that $h$ is the unique fixed point of $\mb T$ in the cone $\{ a (N \cap S_1) : a \in \mb R , a \neq 0\}$. 

Existence and (global) identification of value functions in models with recursive preferences has been studied previously (see \cite{MarinacciMontrucchio}, \cite{HS2012}, and references therein). The most closely related work to ours is \cite{HS2012}, who study existence and uniqueness of value functions for Markovian environments in $L^1$ spaces (whose cones of non-negative functions also have empty interior). \cite{HS2012} provide conditions under which a fixed point may exist when the EIS is equal to unity but do not establish its uniqueness. Their existence conditions are based, in part, on existence of a positive eigenfunction of the operator $\mb G$. 

There is also a connection between Corollary \ref{c:local-id} and the literature on local identification of nonlinear, nonparametric econometric models. We can write $\mb T h = h$ as the conditional moment restriction:
\[
 \mb E\Big[ G_{t+1}^{1-\gamma} \big|h(X_{t+1}) \big|^\beta - h(X_t) \Big| X_t \Big] = 0
\]
(almost surely). The conditions of Proposition \ref{p:nl} ensure that the above moment restriction is Fr\'echet differentiable at $h$ with derivative $\mb D_h - I$. The condition $r(\mb D_h) < 1$ implies that $\mb D_h - I$ is invertible on $L^2$. The conditions in Proposition \ref{p:nl} are therefore similar to the differentiability and rank conditions that \cite{Chenetal2012} use to study local identification in nonlinear conditional moment restriction models.

\subsection{Estimation}

We again use a sieve approach to reduce the infinite-dimensional problem to a low-dimensional (nonlinear) eigenvector problem.  
Consider the projected fixed-point problem: 
\begin{equation} \label{e:pfpe}
 (\Pi_k \mathbb T) h_k = h_k 
\end{equation}
where $\Pi_k : L^2 \to B_k$ is the orthogonal projection onto the sieve space defined in Section \ref{s:est}. Lemma \ref{lem:fp:exist} in the Appendix guarantees existence of a solution $h_k$ to (\ref{e:pfpe}) on a neighborhood of $h$ for all $k$ sufficiently large.
 As $h_k \in B_k$, we have $h_k = b^k(x)' v_k$ for some vector $v_k \in \mb R^k$ which solves:
\begin{equation} \label{e:fp:vec}
 \mf G_k^{-1} \mf T_k^{\phantom {-1}}\!\!\! v_k = v_k
\end{equation}
where $\mf T_k v = \mb E[ b^k(X_{t}) G_{t+1}^{1-\gamma} |b^k(X_{t+1})'v |^\beta ]$. To simplify notation we drop dependence of $\mf G_k$ and $\mf T_k$ on $k$ hereafter. 
For estimation, we solve a sample analogue of (\ref{e:fp:vec}), namely:
\begin{equation} \label{e:fp:sample}
 \wh{\mf G}^{-1} \wh{\mf T} \hat v = \hat v
\end{equation}
where $\wh{\mf G}$ is defined in display (\ref{e:ghat}) and $\wh{\mf T} : \mb R^k \to \mb R^k$ is given by: 
\[
 \wh{\mf T}v = \frac{1}{n} \sum_{t=0}^{n-1}  b^k(X_{t}) G_{t+1}^{1-\gamma} |b^k(X_{t+1})'v |^\beta \,.
\]
Under the regularity conditions below, a solution $\hat v$ on a neighborhood of $v_k$ necessarily exists wpa1 (see Lemma \ref{lem:fphat:exist} in the Appendix).
The estimators of $h$, $\chi$ and $\lambda$ are:
\begin{align} \label{e:fpest}
 \hat h(x) & = b^k(x)' \hat v &
 \hat \chi(x) & = \frac{b^k(x)' \hat v}{(\hat v' \wh{\mf G} \hat v)^{1/2}} &
 \hat \lambda & = (\hat v' \wh{\mf G} \hat v)^{\frac{1-\beta}{2}} \,.
\end{align}
The estimators $\hat \chi$ and $\hat \lambda$ can then be plugged into display (\ref{e:rec:sdf}) to obtain an estimate of the SDF consistent with preference parameters $(\beta,\gamma)$ and the observed law of motion of the state.

\begin{assumption} \label{a:fp:exist}
Let the following hold:
\begin{enumerate}
\item[(a)] $\mb T$ has a unique positive fixed point $h \in L^2$
\item[(b)] $\mb G$ is positive and compact
\item[(c)] $\mb T$ is Fr\'echet differentiable at $h$ with $r(\mb D_h) < 1$.
\end{enumerate}
\end{assumption}

\begin{assumption} \label{a:fp:bias} Let the following hold:
\begin{enumerate}
\item[(a)] $\|\Pi_k \mb D_h - \mb D_h\| = o(1)$  
\item[(b)] $\sup_{\psi \in L^2 : \|\psi\| \leq c} \| \Pi_k \mb T\psi - \mb T\psi\| = o(1)$ for each $c > 0$.
\end{enumerate}
\end{assumption}

Let $\wh{\mf G}^o$ be as in Assumption \ref{a:var}. Let $\mf T^o v = \mf G^{-1/2} \mf T (\mf G^{-1/2} v)$ and $\wh{\mf T}^o v = \mf G^{-1/2} \wh{\mf T} (\mf G^{-1/2} v)$. Note that $\wh{\mf G}^o$ and $\wh{\mf T}^o$ are a proof device and do not need to be calculated in practice.

\begin{assumption} \label{a:fp:var}
$\| \wh{\mf G}^o - \mf I \| = o_p(1)$ and $\sup_{v \in \mb R^k : \|  v\| \leq c} \| \wh{\mf T}^ov - \mf T^o v\| = o_p(1)$ for each $c > 0$.
\end{assumption}

{\bf Discussion of assumptions:} Assumption \ref{a:fp:exist}(a) is a global identification assumption, parts (b) and (c) imposes some mild structure on $\mb T$ which ensures that fixed points of $\mb T$ are continuous under perturbations. Assumption \ref{a:fp:bias}(a)(b) are analogous to Assumption \ref{a:bias}.  Assumption \ref{a:fp:var} is similar to Assumption \ref{a:var} and restricts the rate at which the sieve dimension $k$ can grow with $n$; sufficient conditions are presented in Appendix \ref{ax:est:mat:fp}.

Let $\tau_k = \|\Pi_k h - h\|$ denote the bias in approximating $h$ by an element of the sieve space. Assumption \ref{a:fp:bias}(b) implies that $\tau_k = o(1)$. To control the sampling error, fix any small $\varepsilon > 0$. By Assumption \ref{a:fp:var} we may choose a sequence of positive constants $\nu_{n,k}$ with  $\nu_{n,k} = o(1)$ such that: 
\begin{equation} \label{e:nudef}
 \sup_{v \in \mb R^k : \|  v'b^k- h\| \leq \varepsilon} \| (\wh{\mf G}^o)^{-1} \wh{\mf T}^ov - \mf T^o v\| = O_p(\nu_{n,k})\,.
\end{equation}
Appendix \ref{ax:est:mat:fp} presents bounds on $\nu_{n,k}$.

\begin{theorem} \label{t:fpest}
Let Assumptions \ref{a:fp:exist}--\ref{a:fp:var} hold. Then:
\begin{enumerate}
\item[(a)] $|\hat \lambda - \lambda| = O_p( \tau_k + \nu_{n,k})$
\item[(b)] $\|\hat \chi - \chi\| = O_p( \tau_k + \nu_{n,k})$
\item[(c)] $\|\hat h - h\| = O_p( \tau_k + \nu_{n,k})$.
\end{enumerate}
\end{theorem}

The convergence rates obtained in Theorem \ref{t:fpest} again exhibit a bias-variance tradeoff. The bias terms $\tau_k$ are decreasing in $k$, whereas the variance term $\nu_{n,k}$ is typically increasing in $k$ but decreasing in $n$. Choosing $k$ to balance the terms will lead to the best convergence rate.

For implementation, we propose the following iterative scheme based on Proposition \ref{p:nl}. Set $z_1 = \wh{\mf G}^{-1} (\frac{1}{n} \sum_{t=0}^{n-1} b^k(X_t))$, then  calculate:
\begin{align*}
 a_{k} & = \frac{z_k}{(z_k' \wh{\mf G} z_k^{\phantom \prime})^{1/2}} &
 z_{k+1} & = \wh{\mf G}^{-1} \wh{\mf T} a_k
\end{align*}
for $k \geq 1$. If the sequence $\{(a_k,z_k) : k \geq 1\}$ converges to $(\hat a,\hat z)$ (say), we then set: 
\begin{align*}
 \hat h(x) & = \hat \lambda^{\frac{1}{1-\beta}}b^k(x)' \hat a & 
 \hat \chi(x) & = b^k(x)' \hat a & 
 \hat \lambda = (\hat z' \wh{\mf G} \hat z^{\phantom \prime})^{1/2}\,.
\end{align*}
This iterative scheme proved to be a computationally efficient procedure for solving the sample fixed-point problem (\ref{e:fp:sample}) in the simulations and empirical application.

\section{Simulation evidence}\label{s:mc}

The following Monte Carlo experiment illustrates the performance of the estimators in consumption-based models with power utility and recursive preferences. The state variable is log consumption growth, i.e. $X_t = g_t$, which evolves as a Gaussian AR(1) process:
\[
 g_{t+1} - \mu = \kappa (g_t - \mu) + \sigma e_{t+1}\,,  \quad e_t \sim \mbox{ i.i.d. N$(0,1)$.}
\]
The parameters for the simulation are $\mu = 0.005$, $\kappa = 0.6$, and $\sigma = 0.01$. The data are constructed to be somewhat representative of quarterly growth in U.S. real per capita consumption of nondurables and services (for which $\kappa \approx 0.3$ and $\sigma \approx 0.005$). However, we make the consumption growth process twice as persistent  to produce more nonlinear eigenfunctions and twice as volatile   to produce a more challenging estimation problem. 

We consider a power utility design in which $m(X_t,X_{t+1}) = \beta G_{t+1}^{-\gamma}$ and a design with recursive preferences with unit EIS, whose SDF is presented in display (\ref{e:rec:sdf}). For both designs we set $\beta = 0.994$ and $\gamma = 15$. The parameterization $\beta = 0.994$ and $\gamma = 10$ is typically used in calibrations of long-run risks models; here we take $\gamma = 15$ so that the eigenfunctions and continuation value function are more nonlinear. For each design we generate 50000 samples of length 400, 800, 1600, and 3200. Results reported in this section use a Hermite polynomial basis of dimension $k=8$.  Further experimentation with other sieve dimensions showed that the results were reasonably insensitive to the dimension of the sieve space. Similar results were obtained using B-splines (see Appendix \ref{s:mc:supp} in the Online Appendix). 

We estimate $\phi$, $\phi^*$, $\rho$, $y$, and $L$ for both designs and $\chi$ and  $\lambda$ for the recursive preference design. We use the estimator $\wh{\mf G}$ in (\ref{e:ghat}) for both preference specifications. For power utility we use the estimator $\wh{\mf M}$ in (\ref{e:mhat1}). For recursive preferences we first estimate $(\lambda,\chi)$ using the method described in the previous section, then construct the estimator $\wh{\mf M}$ as in display (\ref{e:mhat2}), using:
\[
 m(X_t,X_{t+1};\hat \lambda,\hat \chi) = \frac{\beta}{\hat \lambda} G_{t+1}^{-\gamma} \frac{ \big(\hat \chi(X_{t+1}) \big)^\beta}{ \hat \chi(X_t) } 
\]
based on the first-stage estimators $(\hat \lambda,\hat \chi)$ of $(\lambda,\chi)$.
We impose the scale normalizations  $\frac{1}{n} \sum_{t=0}^{n-1} \hat \phi(X_t)^2 = 1$,  $\frac{1}{n} \sum_{t=0}^{n-1} \hat \phi(X_t)\hat\phi^*(X_t) = 1$, and  $\frac{1}{n} \sum_{t=0}^{n-1} \hat \chi(X_t)^2 = 1$.

The bias and RMSE of the estimators are presented in Tables \ref{tab:mc1} and \ref{tab:mc2}.\footnote{To calculate the RMSE of $\hat \phi$, $\hat \phi^*$, and $\hat \chi$, for each replication we calculate the $L^2$ distance between the estimators and their population counterparts, then take the average over the MC replications. To calculate the bias we take the average of the estimators across the simulations to produce $\bar \phi(x)$, $\bar \phi^*(x)$, and $\bar \chi(x)$ (say), then compute the $L^2$ distance between $\bar \phi$, $\bar \phi^*$ and $\bar \chi$ and the true $\phi$, $\phi^*$ and $\chi$. The use of the ``bias'' here is not to be confused with the bias term in the convergence rate calculations: here ``bias'' of an estimator refers to the distance between the parameter and the average of its estimates across the simulations. Bias for $\hat \rho$, $\hat y$, $\hat L$, and $\hat \lambda$ is the average of the estimates across simulations minus the true parameter values.} 
Table \ref{tab:mc1} shows that $\phi$, $\phi^*$ and $\chi$ may be estimated with small bias and RMSE using a reasonably low-dimensional sieve. Table \ref{tab:mc2} presents similar results for $\hat \rho$, $\hat y$, $\hat L$ and $\hat \lambda$. The RMSEs for $\hat \phi$ and $\hat \rho$ under recursive preferences are typically smaller than the RMSEs for $\hat \phi$ and $\hat \rho$ under power utility, even though with recursive preferences the continuation value must be first estimated nonparametrically. In contrast, the RMSE for $\hat \phi^*$ is larger under recursive preferences, which is likely due to the fact that $\phi^*$ is much more curved for that design (as evident from comparing the vertical scales Figures \ref{fig:mc:sfig2} and \ref{fig:mc:sfig4}). The results in Table \ref{tab:mc1} also show that $\chi$ may be estimated with a reasonably small degree of bias and RMSE in moderate samples. 

Figures \ref{fig:mc:sfig1}--\ref{fig:mc:sfig5} also present (pointwise) confidence intervals for $\phi$, $\phi^*$ and $\chi$ computed across simulations of different sample sizes. For each figure, the true function lies approximately in the center of the pointwise confidence intervals, and the widths of the intervals shrink noticeably as the sample size $n$ increases. Corresponding plots using a B-spline basis are presented in the Online Appendix and are very similar to Figures \ref{fig:mc:sfig1}--\ref{fig:mc:sfig5}.

\begin{table}[t]
\small
\centering 
\begin{tabular}{cc|cc|ccc} 
\hline \hline
 & & \multicolumn{2}{c|}{Power Utility} & \multicolumn{3}{c}{Recursive Preferences} \\
 & $n$ &$\hat \phi$  &$\hat \phi^*$& $\hat \phi$& $\hat \phi^*$& $\hat \chi$ \\ \hline 
\multirow{4}{*}{Bias} 	& 400 & 0.0144 & 0.0129 & 0.0027 & 0.0247 & 0.0119 \\ 
 						& 800 & 0.0115 & 0.0129 & 0.0020 & 0.0187 & 0.0090 \\
 						& 1600& 0.0084 & 0.0104 & 0.0016 & 0.0128 & 0.0062 \\ 
 						& 3200& 0.0058 & 0.0077 & 0.0014 & 0.0095 & 0.0040 \\ \hline 
\multirow{4}{*}{RMSE} 	& 400 & 0.1136 & 0.1683 & 0.0458 & 0.4068 & 0.1034 \\ 
 						& 800 & 0.0872 & 0.1060 & 0.0413 & 0.3513 & 0.0760 \\ 
 						& 1600& 0.0681 & 0.0837 & 0.0361 & 0.1763 & 0.0577 \\ 
 						& 3200& 0.0552 & 0.0677 & 0.0317 & 0.1591 & 0.0455  \\ \hline \hline
\end{tabular} \vskip 4pt
\parbox{5.0in}{\caption{\label{tab:mc1} \small
Simulation results for $\hat \phi$, $\hat \phi^*$ and $\hat \chi$ with a Hermite Polynomial sieve of dimension $k =8$.}} 
\end{table}

\begin{table}[t]
\small
\centering 
\begin{tabular}{cc|ccc|cccc}
\hline \hline
 & & \multicolumn{3}{c|}{Power Utility} & \multicolumn{4}{c}{Recursive Preferences} \\
 & $n$ & $\hat \rho$ & $\hat y$    & $\hat L$    & $\hat \rho$ & $\hat y$    & $\hat L$ & $\hat \lambda$ \\ \hline 
\multirow{4}{*}{Bias} 	& 400 & 0.0035 &-0.0029 & 0.0029 & 0.0010 &-0.0008 & 0.0034 & 0.0040 \\
						& 800 & 0.0027 &-0.0024 & 0.0024 & 0.0011 &-0.0010 & 0.0027 & 0.0022 \\
						& 1600& 0.0020 &-0.0018 & 0.0018 & 0.0010 &-0.0008 & 0.0020 & 0.0014 \\ 
						& 3200& 0.0014 &-0.0013 & 0.0012 & 0.0010 &-0.0009 & 0.0016 & 0.0009 \\ \hline 
\multirow{4}{*}{RMSE} 	& 400 & 0.0358 & 0.0338 & 0.0282 & 0.0216 & 0.0179 & 0.0420 & 0.1005 \\ 
 						& 800 & 0.0264 & 0.0251 & 0.0214 & 0.0217 & 0.0172 & 0.0299 & 0.0318 \\ 
 						& 1600& 0.0204 & 0.0192 & 0.0168 & 0.0190 & 0.0151 & 0.0227 & 0.0179 \\ 
 						& 3200& 0.0159 & 0.0149 & 0.0133 & 0.0192 & 0.0155 & 0.0204 & 0.0123 \\ \hline \hline
\end{tabular} \vskip 4pt
\parbox{5.0in}{\caption{\label{tab:mc2} \small
Simulation results for $\hat \rho$, $\hat y$, $\hat L$ and $\hat \lambda$ with a Hermite Polynomial sieve of dimension $k =8$.}} 
\end{table}

\begin{figure}[p]
\centering
\begin{subfigure}{.5\textwidth}
  \centering
  \includegraphics[width=\linewidth]{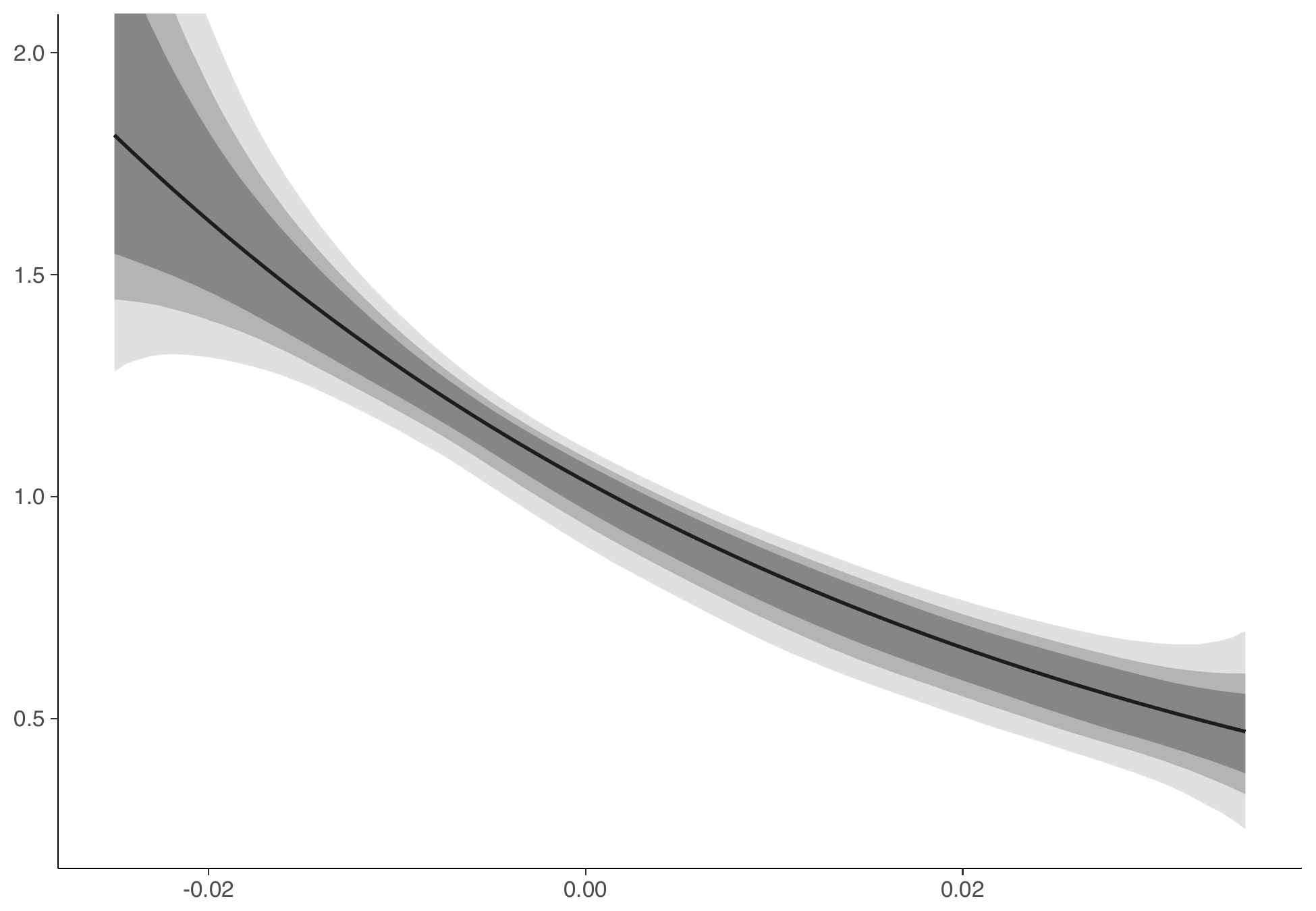}
  \caption{\footnotesize $\hat \phi(x)$ for power utility}
  \label{fig:mc:sfig1}
\end{subfigure}%
\begin{subfigure}{.5\textwidth}
  \centering
  \includegraphics[width=\linewidth]{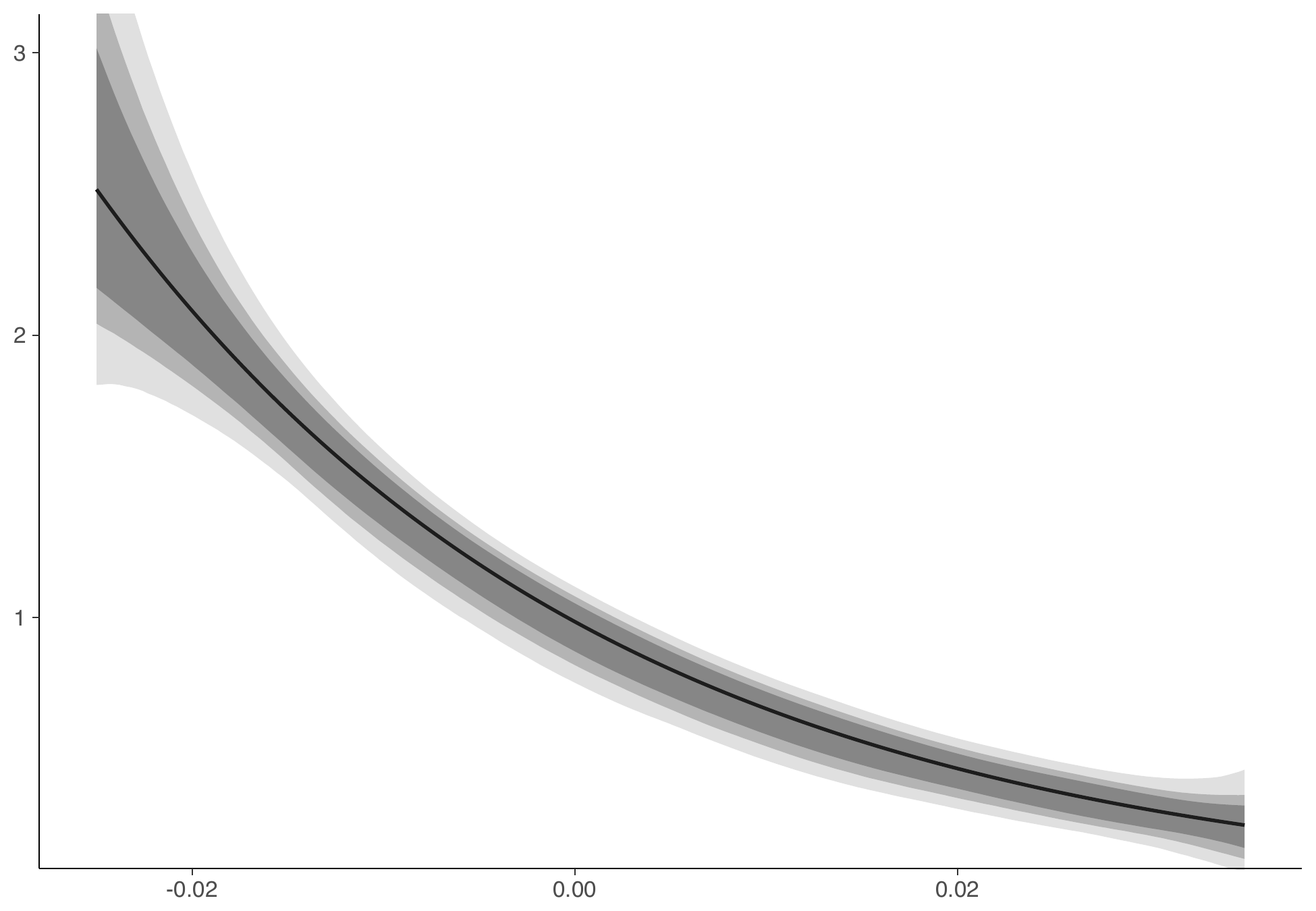}
  \caption{\footnotesize $\hat \phi^*(x)$ for power utility}
  \label{fig:mc:sfig2}
\end{subfigure} \\[10pt]
\begin{subfigure}{.5\textwidth}
  \centering
  \includegraphics[width=\linewidth]{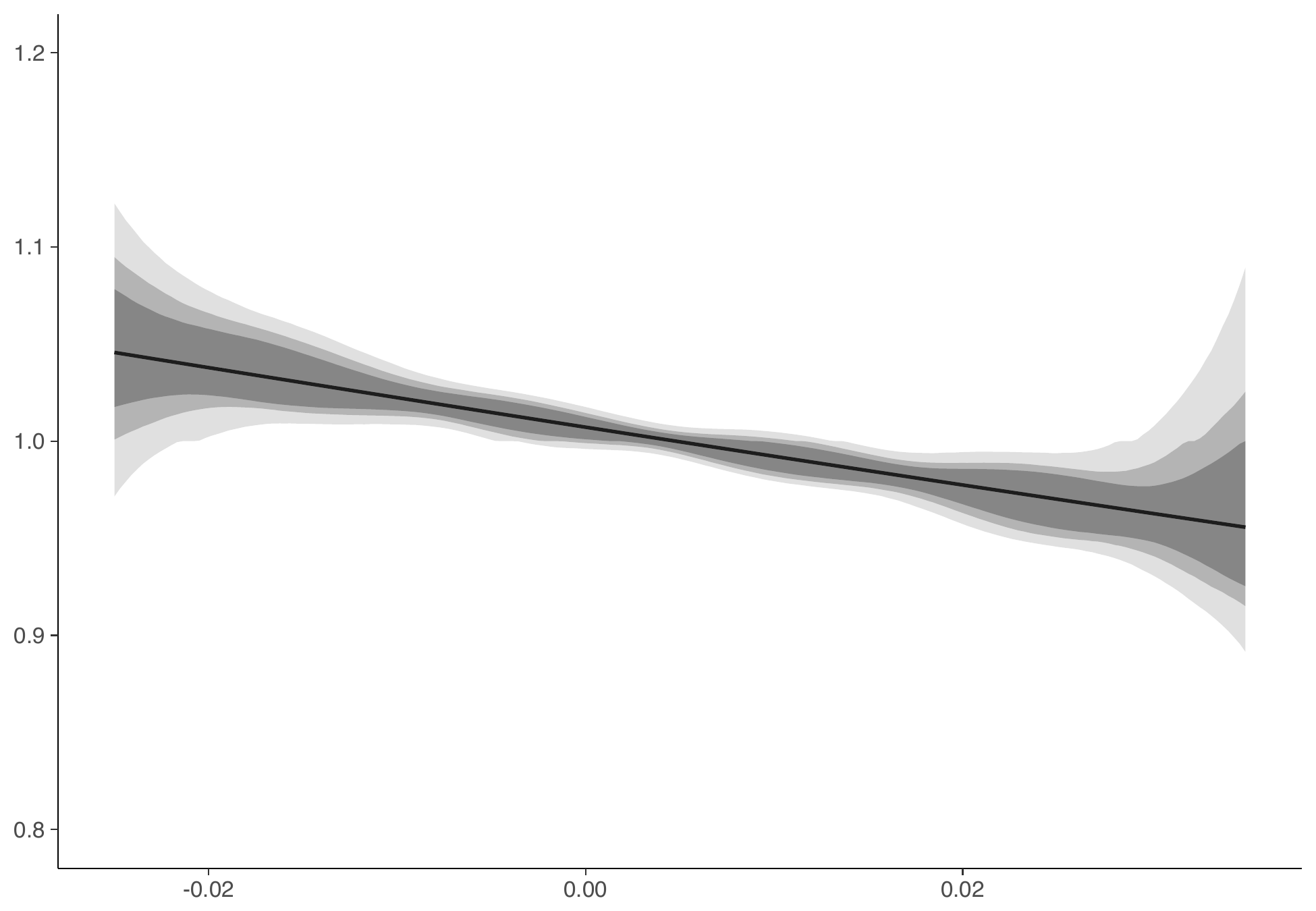}
  \caption{\footnotesize $\hat \phi(x)$ for recursive preferences}
  \label{fig:mc:sfig3}
\end{subfigure}%
\begin{subfigure}{.5\textwidth}
  \centering
  \includegraphics[width=\linewidth]{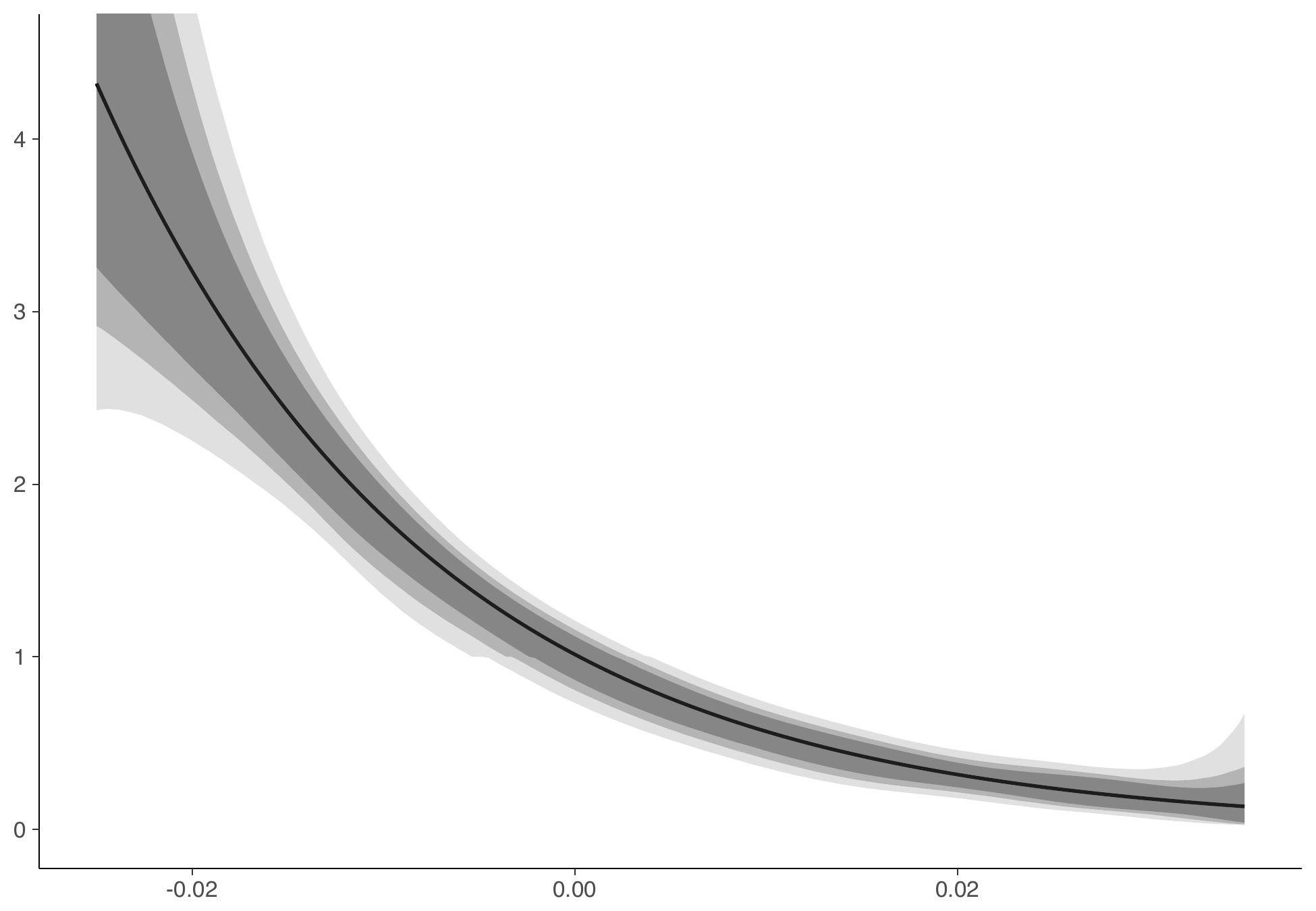}
  \caption{\footnotesize $\hat \phi^*(x)$ for recursive preferences}
  \label{fig:mc:sfig4}
\end{subfigure} \\[10pt]
\begin{subfigure}{.5\textwidth}
  \centering
  \includegraphics[width=\linewidth]{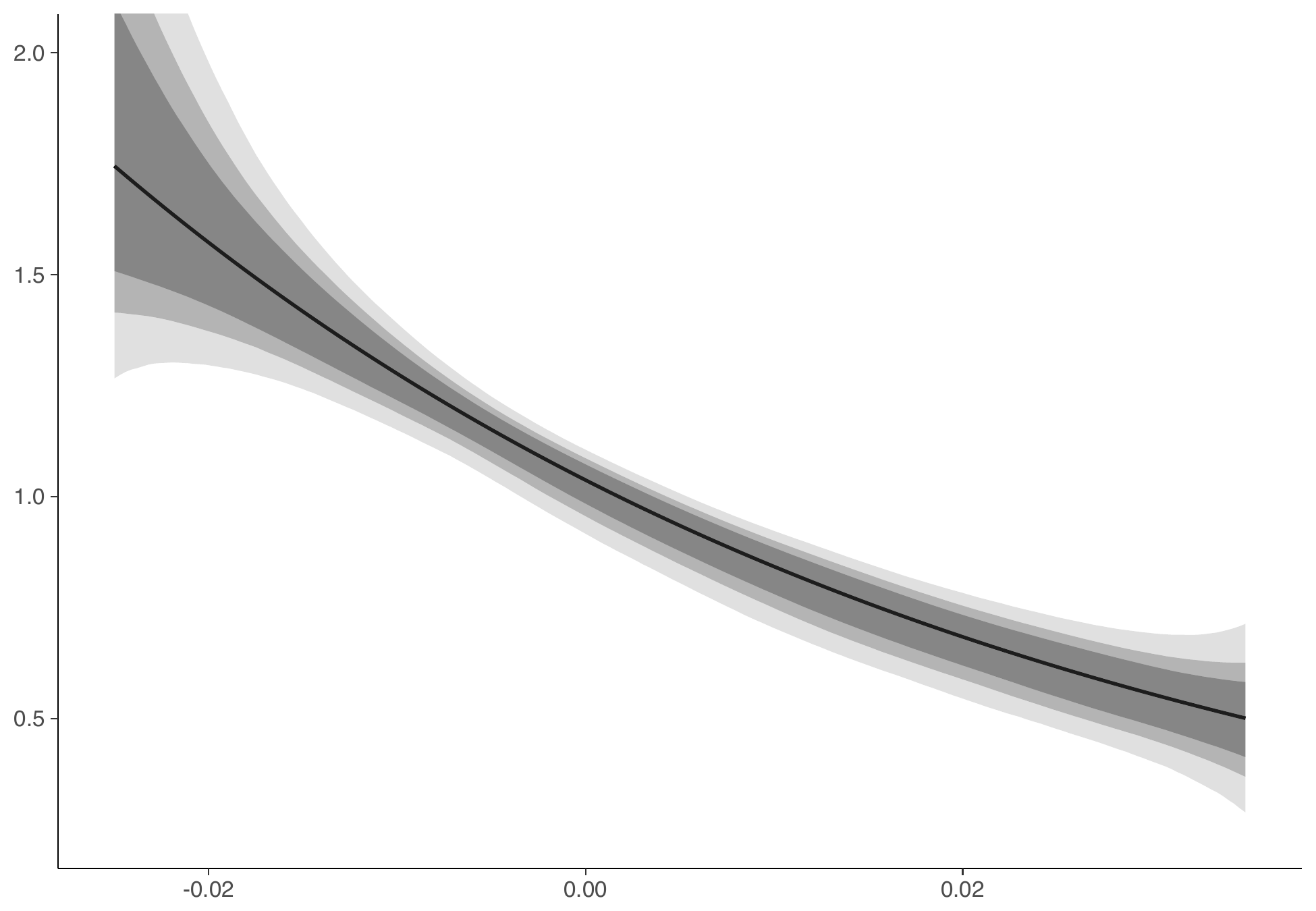}
  \caption{\footnotesize $\hat \chi(x)$ for recursive preferences}
  \label{fig:mc:sfig5}
\end{subfigure}
\begin{center}
\parbox{5.0in}{\caption{ \small Simulation results for a Hermite polynomial basis with $k=8$. Panels (a)--(d) display pointwise 90\% confidence intervals for $\phi$ and $\phi^*$ across simulations (light, medium and dark correspond to $n= 400$, $800$, and $1600$ respectively; the true $\phi$ and $\phi^*$ plotted as solid lines). Panel (e) displays results for the positive eigenfunction $\chi$ of the continuation value operator.}}\label{fig:mc}
\end{center}
\end{figure}

\section{Empirical application}\label{s:emp}

In this section we study an economy similar to that in \cite{HansenHeatonLi}. We assume a representative agent with \cite{EpsteinZin1989} recursive preferences with unit EIS and specify a two-dimensional state process in consumption and earnings growth. Our analysis may be summarized as follows. First, with discount and risk aversion parameters estimated from asset returns data ($\hat \beta \approx 0.985$ and $\hat \gamma \approx 24.5$), we show that this bivariate specification is able to generate a permanent component which implies a long-run equity premium (i.e. return on assets relative to long-term discount bonds) of approximately 2\% per quarter. Second, we document the business cycle properties of the permanent and transitory components. Third, we describe the wedge required to tilt the distribution of the state to that which is relevant for long-run pricing. Finally, we show that, unlike the linear-Gaussian case, allowing for flexible treatment of the state process can lead to different behavior of long-run yields and different signs of correlation between the permanent and transitory components for different preference parameters. This suggests that nonlinearities in dynamics can be important in explaining the long end of the yield curve.

All data are quarterly and span the period 1947:Q1 to 2016:Q1 (277 observations). Data on consumption, dividends, inflation, and population are sourced from the National Income and Product Accounts (NIPA) tables. Real per capita consumption  and dividend growth series are formed by taking seasonally adjusted consumption of nondurables and services (Table 2.3.5, lines 8 plus 13) and dividends (Table 1.12, line 16), deflating by the personal consumption implicit price deflator  (Table 2.3.4, line 1), then converting to per capita growth rates using population data (Table 2.1, line 40). The resulting state variable is $X_t = (g_t,d_t)$ where $g_t$ and $d_t$ are real per capita consumption and dividend growth in quarter $t$, respectively. Similar results are obtained replacing $d_t$ with real per capita growth in corporate earnings (using after-tax profits from line 15 of Table 1.12) and with real per capita growth in a four-quarter geometric moving average of dividends, as in  \cite{HansenHeatonLi}.

We also use data on seven asset returns, namely the returns on the six value-weighted portfolios sorted on size and book-to-market values (sourced from Kenneth French's website) and the 90-day Treasury bill rate. All asset returns series are converted to real returns using the implicit price deflator for personal consumption expenditures.

We estimate the preference parameters $(\beta,\gamma)$ and the pair $(\lambda,\chi)$ using the data on $X_t$ and the time series of seven asset returns. This falls into the setup of ``Case 2'' with $\alpha = (\beta,\gamma,\lambda,\chi)$. We estimate the parameters $(\beta,\gamma)$ using a series conditional moment estimation procedure \citep{AiChen2003}. This methodology was used recently in a similar context by \cite*{ChenFavilukisLudvigson}.\footnote{The differences between our estimator and that of \cite{ChenFavilukisLudvigson} are as follows. First, we focus on the EIS $=1$ case whereas \cite{ChenFavilukisLudvigson} treat the EIS as a free parameter. Second, we exploit the eigenfunction representation of the continuation value recursion. Third, we ``profile out'' continuation value function estimation by solving for $(\lambda,\chi)$ separately from estimating the preference parameters. Therefore, our criterion function depends only on $(\beta,\gamma)$. In contrast, \cite{ChenFavilukisLudvigson} jointly estimate the preference parameters and the continuation value function. Fourth, the continuation value is a function of the Markov state in our analysis whereas the continuation value function in \cite{ChenFavilukisLudvigson} depends on contemporaneous consumption growth and the lagged continuation value.} For each $(\beta,\gamma)$, we estimate the solution to the nonlinear eigenfunction problem, namely $(\hat \lambda_{(\beta,\gamma)},\hat \chi_{(\beta,\gamma)})$, using the procedure introduced in Section \ref{s:recursive}. Here we make explicit the dependence of $(\lambda,\chi)$  on $\beta $ and $\gamma$, since different preference parameters will correspond to different continuation value functions. We then form:
\[
 m(X_t,X_{t+1};(\beta,\gamma,\hat \lambda_{(\beta,\gamma)},\hat \chi_{(\beta,\gamma)})) = \frac{\beta}{\hat \lambda_{(\beta,\gamma)}} G_{t+1}^{-\gamma} \frac{ \big(\hat \chi_{(\beta,\gamma)}(X_{t+1}) \big)^\beta}{  \hat \chi_{(\beta,\gamma)}(X_t) }\,.
\]
Let $\mf R_{t+1}$ denote a vector of (gross) asset returns from time $t$ to $t+1$ and $\mf 1$ and $\mf 0$ denote conformable vectors of ones and zeros. As the Euler equation $\mb E[m(X_t,X_{t+1}) \mf R_{t+1} - \mf 1 | X_t] = \mf 0$ holds conditionally, we instrument the generalized residuals, namely:
\[
 m(X_t,X_{t+1};(\beta,\gamma,\hat \lambda_{(\beta,\gamma)},\hat \chi_{(\beta,\gamma)})) \mf R_{t+1} - \mf 1\,,
\]
by basis functions of $X_t$ to form a criterion function which exploits the conditional nature of the Euler equation. This leads to the criterion function:
\[
 L_n(\beta,\gamma)  =  \frac{1}{n} \sum_{t=0}^{n-1} \| l_n(X_t,\beta,\gamma)\|^2 
\]
where
\begin{align*}
 l_n(x,\beta,\gamma) & = \left( \frac{1}{n} \sum_{t=0}^{n-1} \Big( m \big(X_t,X_{t+1};(\beta,\gamma,\hat \lambda_{(\beta,\gamma)},\hat \chi_{(\beta,\gamma)})\big) \mf R_{t+1} - \mf 1 \Big)b^k(X_t) ' \right)   \wh{\mf G}^{-} b^k(x) \,.
\end{align*}
We minimize $L_n(\beta,\gamma)$ to obtain $(\hat \beta,\hat \gamma)$ and we set $\hat \alpha = (\hat \beta,\hat \gamma,\hat \lambda_{(\hat \beta,\hat \gamma)},\hat \chi_{(\hat \beta,\hat \gamma)})$. We then estimate $\rho$, $\phi$, $\phi^*$ and related quantities using the estimator $\wh{\mf M}$ in display (\ref{e:mhat2}) for this choice of $\hat \alpha$. 

To implement the procedure, we use the same basis functions for estimation of  $(\rho$, $\phi$, $\phi^*)$  and $(\lambda,\chi)$. We form fifth-order univariate Hermite polynomial bases for the $g_t$ and $d_t$ series. We then construct a tensor product basis from the univariate bases, discarding any tensor-product polynomials whose total degree is order six or higher. The resulting sparse basis has dimension 15 whereas a tensor product basis would have dimension 25. We sometimes compare with the univariate state process $X_t = g_t$ for which we use an eighth-order Hermite polynomial basis. We instrument the generalized residuals with a lower-dimensional vector of basis functions to form the criterion function $L_n$ (dimension 6 in the univariate case and 10 in the bivariate case). Similar results are obtained with different dimensional bases.

\begin{table}[t]
\centering 
\begin{tabular}{c|cc|ccc} 
\hline \hline
 & $X_t = (g_t,d_t)$ & $X_t = g_t$ & \phantom{$X_t = (g_t,d_t)$} & $X_t = (g_t,d_t)$ & \phantom{$X_t = (g_t,d_t)$} \\ \hline & & &  \\[-12pt]
$\hat \rho$   & $\underset{( 0.9733, 0.9902)}{ 0.9812}$ & $\underset{(0.9726, 0.9893)}{ 0.9817}$ & $\underset{(0.9851,0.9872)}{0.9859}$ & $\underset{(0.9850,0.9881)}{0.9861}$ & $\underset{(0.9842,0.9913)}{0.9860}$ \\
$\hat y $     & $\underset{( 0.0098, 0.0270)}{ 0.0190}$ & $\underset{(0.0107, 0.0277)}{ 0.0184}$ & $\underset{(0.0129,0.0150)}{0.0142}$ & $\underset{(0.0120,0.0151)}{0.0140}$ & $\underset{(0.0087,0.0159)}{0.0141}$ \\
$\hat L $     & $\underset{( 0.0000, 0.0381)}{ 0.0193}$ & $\underset{(0.0000, 0.0426)}{ 0.0215}$ & $\underset{(0.0090,0.0185)}{0.0128}$ & $\underset{(0.0146,0.0295)}{0.0203}$ & $\underset{(0.0198,0.0435)}{0.0297}$ \\ 
\hline & & &  \\[-12pt] 
$\hat \beta$  & $\underset{( 0.9784, 0.9926)}{ 0.9851}$ & $\underset{(0.9771, 0.9921)}{ 0.9853}$ & 0.99 & 0.99 & 0.99 \\
$\hat \gamma$ & $\underset{( 0.6850, 44.7570)}{24.4712}$ & $\underset{(0.0000, 50.4619)}{27.4838}$ & 20 & 25 & 30  \\ 
$\hat \lambda$& $\underset{( 0.8146, 0.9922)}{ 0.8999}$ & $\underset{(0.7927, 0.9888)}{ 0.8872}$ & $\underset{(0.9008,0.9324)}{0.9154}$ & $\underset{(0.8789,0.9205)}{0.8983}$ & $\underset{(0.8579,0.9111)}{0.8834}$ \\ \hline  \hline
\end{tabular} \vskip 4pt
\parbox{5.0in}{\caption{\label{tab:emp} \small
Left panel: Estimates of $\rho$, $y$ and $L$ corresponding to $(\hat \beta,\hat \gamma,\hat \lambda,\hat \chi)$. Right panel: estimates  of $\rho$, $y$ and $L$ corresponding to pre-specified $(\beta,\gamma)$ and estimated $(\hat \lambda,\hat \chi)$. 90\% bootstrap confidence intervals are in parentheses.}} 
\end{table}

Table \ref{tab:emp} presents the estimates and bootstrapped 90\% confidence intervals.\footnote{We resample the data 1000 times using the stationary bootstrap with an expected block length of six quarters. In the left panel we re-estimate $\beta$, $\gamma$, $\lambda$, $\chi$, $\rho$, $y$, and $L$ for each bootstrap replication. We discard the tiny fraction of replications in which the estimator of $(\beta,\gamma)$  failed to converge. In the right panel we fix $\beta$ and $\gamma$ and re-estimate $\lambda$, $\chi$, $\rho$, $y$, and $L$ for each bootstrap replication. } There are several notable aspects. First, both state specifications generate a permanent component whose entropy is consistent with a return premium of around 2\% per quarter relative to the long bond, which is in the ballpark of empirically reasonable estimates. Second, the estimated long-run yield of around 1.9\% per quarter is too large, which is explained by the low value of $\hat \beta$. Third, the estimated entropy of the SDF, namely $\log (\frac{1}{n} \sum_{t=0}^{n-1} m(X_t,X_{t+1};\hat \alpha)) - \frac{1}{n} \sum_{t=0}^{n-1}  \log (m(X_t,X_{t+1};\hat \alpha))$ is $0.0191$ for the bivariate specification and $0.0208$ for the univariate specification.
Therefore, the estimated horizon dependence (the difference between the entropy of the permanent component and the entropy of the SDF) is within the bound of $\pm$0.1\% per month that \cite{BackusChernovZin} argue is required to match the spread in average yields between short- and long-term  bonds. 
Finally, the estimates of $\gamma$ are quite imprecise, in agreement with previous studies (e.g., \cite{ChenFavilukisLudvigson}). The confidence intervals for $\rho$, $y$ and $L$ in the left panel of Table \ref{tab:emp} are rather wide which reflects, in large part, the uncertainty in estimating $\beta$ and $\gamma$. Experimentation with different sieve dimensions resulted in estimates of $\gamma$ between $20$ and $30$.\footnote{\cite{ChenFavilukisLudvigson} obtain $\hat \gamma \approx 60$ using aggregate consumption data and $\hat \gamma \approx 20$ using stockholder consumption data. Further, with stockholder data their estimated EIS is not significantly different from zero. This suggests that our estimates of $\gamma$ and maintained assumption of a unit EIS are empirically plausible.} The right panel of Table \ref{tab:emp} presents estimates of $\rho$, $y$ and $L$ fixing $\beta = 0.99$ and $\gamma = 20$, $25$, and $30$ ($\chi$ and $\lambda$ are still estimated nonparametrically). It is clear that the resulting confidence intervals are much narrower once the uncertainty from estimating $\beta$ and $\gamma$ is shut down.

We now turn to analyzing the time-series properties of the permanent and transitory components. The upper two panels of Figure \ref{fig:ts} display time-series plots for the bivariate state specification of the SDF $m(X_t,X_{t+1};\hat \alpha)$ and its permanent component, constructed as:
\begin{align*}
 \frac{\hat M_{t+1}^P}{\hat M_t^P} & = \hat \rho^{-1} m(X_t,X_{t+1};\hat \alpha) \frac{\hat \phi(X_{t+1})}{\hat \phi(X_t)} \,.
\end{align*}
As can be seen, the SDF and its permanent component evolve closely over time and exhibit strong counter-cyclicality. The transitory component (not plotted) is small, consistent with the literature on bounds which finds that the transitory component should be substantially smaller than the permanent component. The correlation of the permanent component series $\hat M_{t+1}^P/\hat M_t^P$ and GDP growth is approximately $-0.39$ whereas the correlation of the transitory component series $\hat M_{t+1}^T/\hat M_t^T$ and GDP growth is approximately $0.05$. The lower panels of Figure \ref{fig:ts} display time-series plots of the SDF and permanent component obtained under power utility using the same $(\hat \beta,\hat \gamma)$ as in the recursive preference specification. This panel shows that the permanent component, which is similar to that obtained under recursive preferences, is much more volatile than the SDF series. The large difference between the SDF and permanent component series under power utility is due to a very volatile transitory component, which implies a counterfactually large spread in average yields between short- and long-term bonds  (\cite{BackusChernovZin}).

\begin{figure}[t]
\centering
\includegraphics[width=0.9\linewidth]{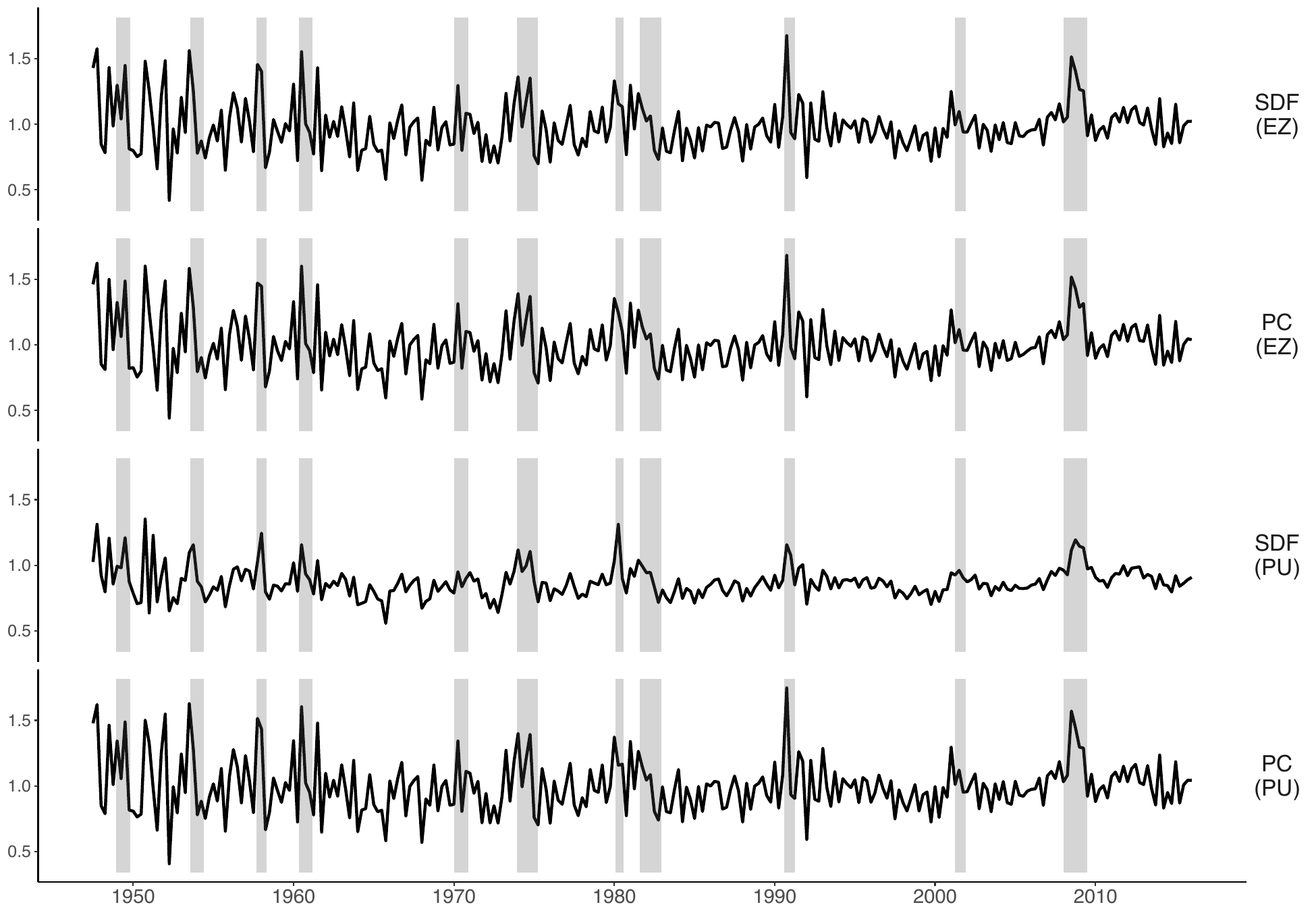}
\parbox{5.0in}{\caption{\label{fig:ts} \small Recovered time series of the SDF and and its permanent component (PC) when $X_t = (g_t,d_t)$. Upper panels are with \cite{EpsteinZin1989} recursive preferences with unit EIS (EZ), lower panels are with power utility (PU). Both use the estimated preference parameters $(\hat \beta,\hat \gamma) = (0.985,24.471)$. Shaded areas denote NBER recession periods.}}
\end{figure}

To understand further the long-run pricing implications of the estimates of $\rho$, $\phi$ and $\phi^*$, Figures \ref{fig:emp:phi_1}--\ref{fig:emp:phistar_2} plot the estimated $\phi$ and $\phi^*$ under recursive preferences for both the bivariate and univariate state specifications. It is evident from the vertical scales in Figures \ref{fig:emp:phi_1} and \ref{fig:emp:phi_2} that both estimates of $\phi$ are relatively flat, which explains the small transitory component obtained with recursive preferences. However, the estimated $\phi^*$ has a pronounced downward slope in $g$. The estimated $\phi^*$ in the bivariate specification is also downward-sloping in $d$ at low levels of consumption growth. Recall that Proposition \ref{p:id} shows that $\phi\phi^*$ is the Radon-Nikodym derivative of  $\wt Q$ with respect to $Q$. Figures \ref{fig:emp:rn_1}--\ref{fig:emp:rn_2} plot the estimated change of measure for the two specifications of the state process. As the estimate of $\phi$ is relatively flat, the estimated change of measure is characterized largely by $\hat \phi^*$. In the bivariate case, the distribution $\wt Q$ assigns relatively more mass to regions of the state space in which there is low dividend and consumption growth than the stationary distribution $Q$, and relatively less mass to regions with high consumption growth.

\begin{figure}[p]
\centering
\begin{subfigure}{.5\textwidth}
  \centering
  \includegraphics[width=\linewidth]{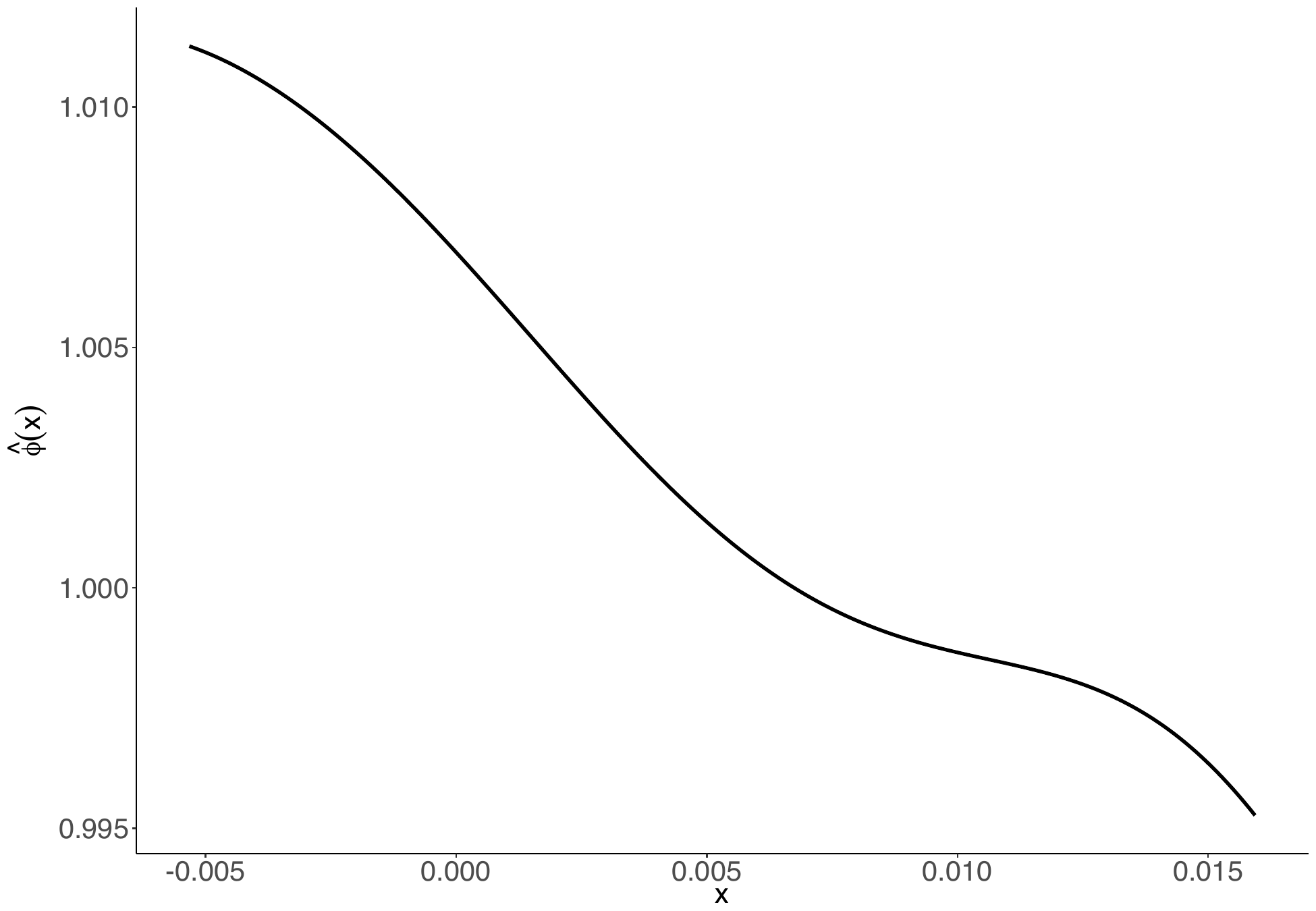}
  \caption{\footnotesize Plot of $\hat \phi(x)$ for $X_t = g_t$}
  \label{fig:emp:phi_1}
\end{subfigure}%
\begin{subfigure}{.5\textwidth}
  \centering
  \includegraphics[width=\linewidth]{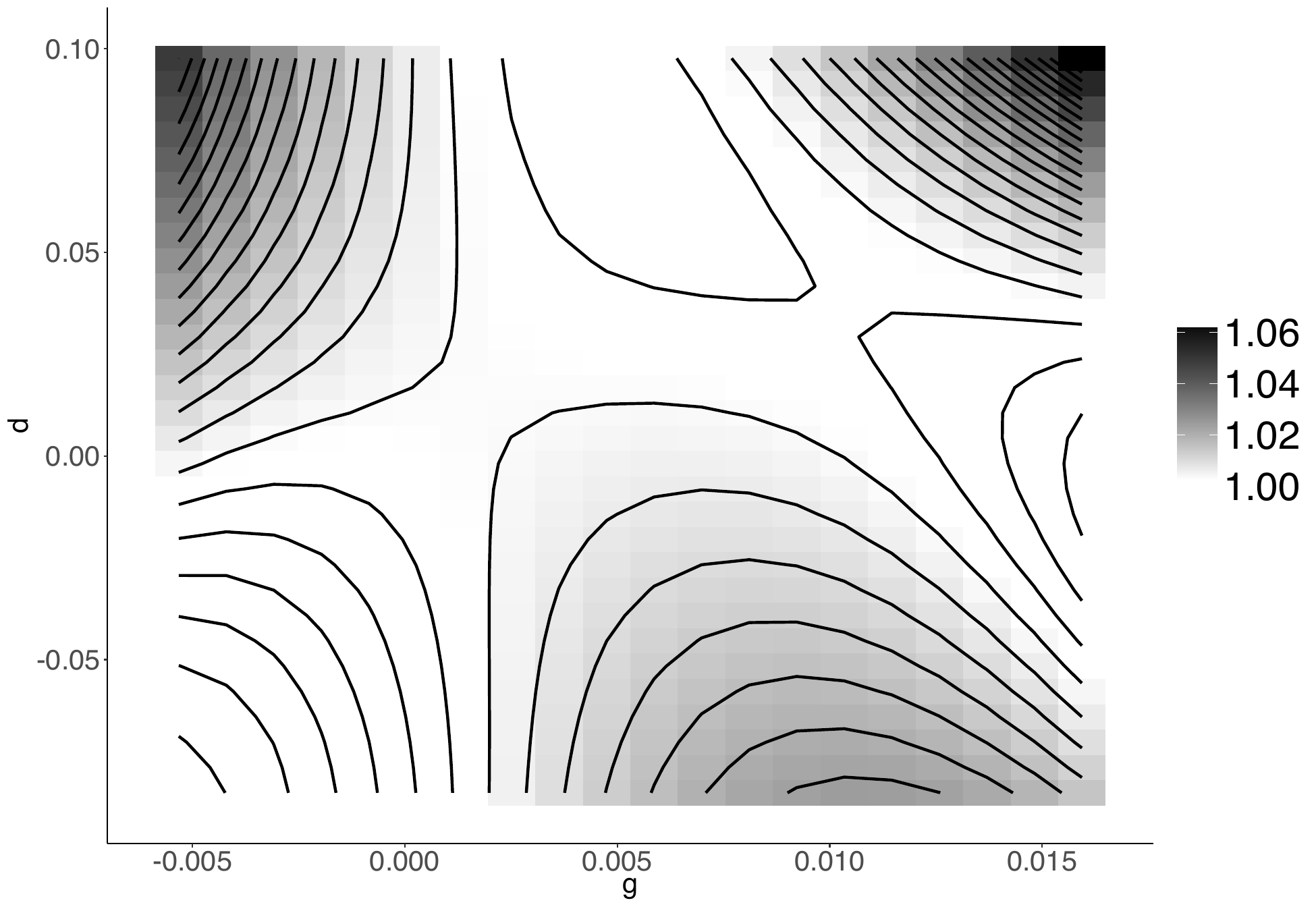}
  \caption{\footnotesize Contour plot of $\hat \phi(x)$ for $X_t = (g_t,d_t)$}
  \label{fig:emp:phi_2}
\end{subfigure} \\[10pt]
\begin{subfigure}{.5\textwidth}
  \centering
  \includegraphics[width=\linewidth]{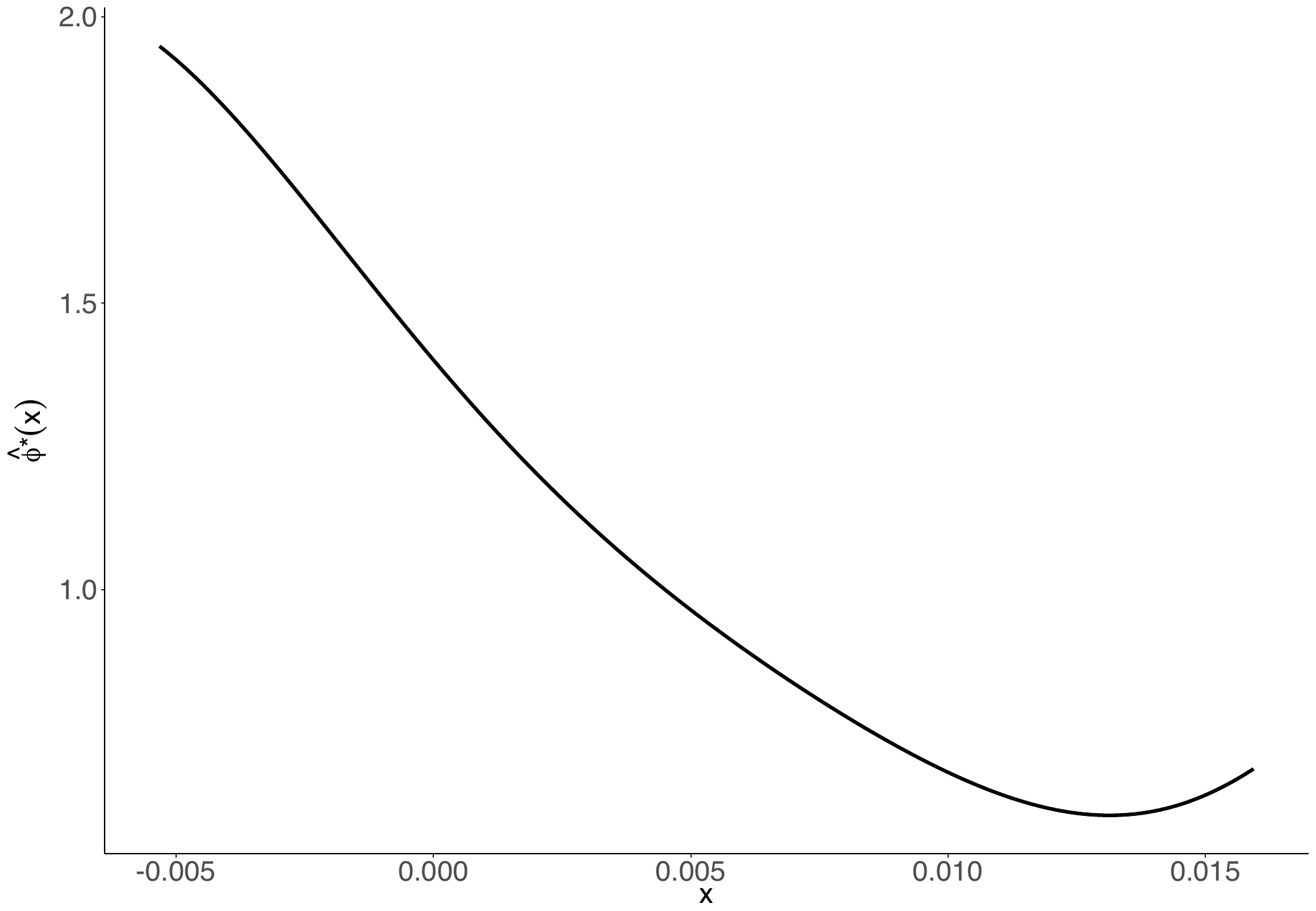}
  \caption{\footnotesize Plot of $\hat \phi^*(x)$ for $X_t = g_t$}
  \label{fig:emp:phistar_1}
\end{subfigure}%
\begin{subfigure}{.5\textwidth}
  \centering
  \includegraphics[width=\linewidth]{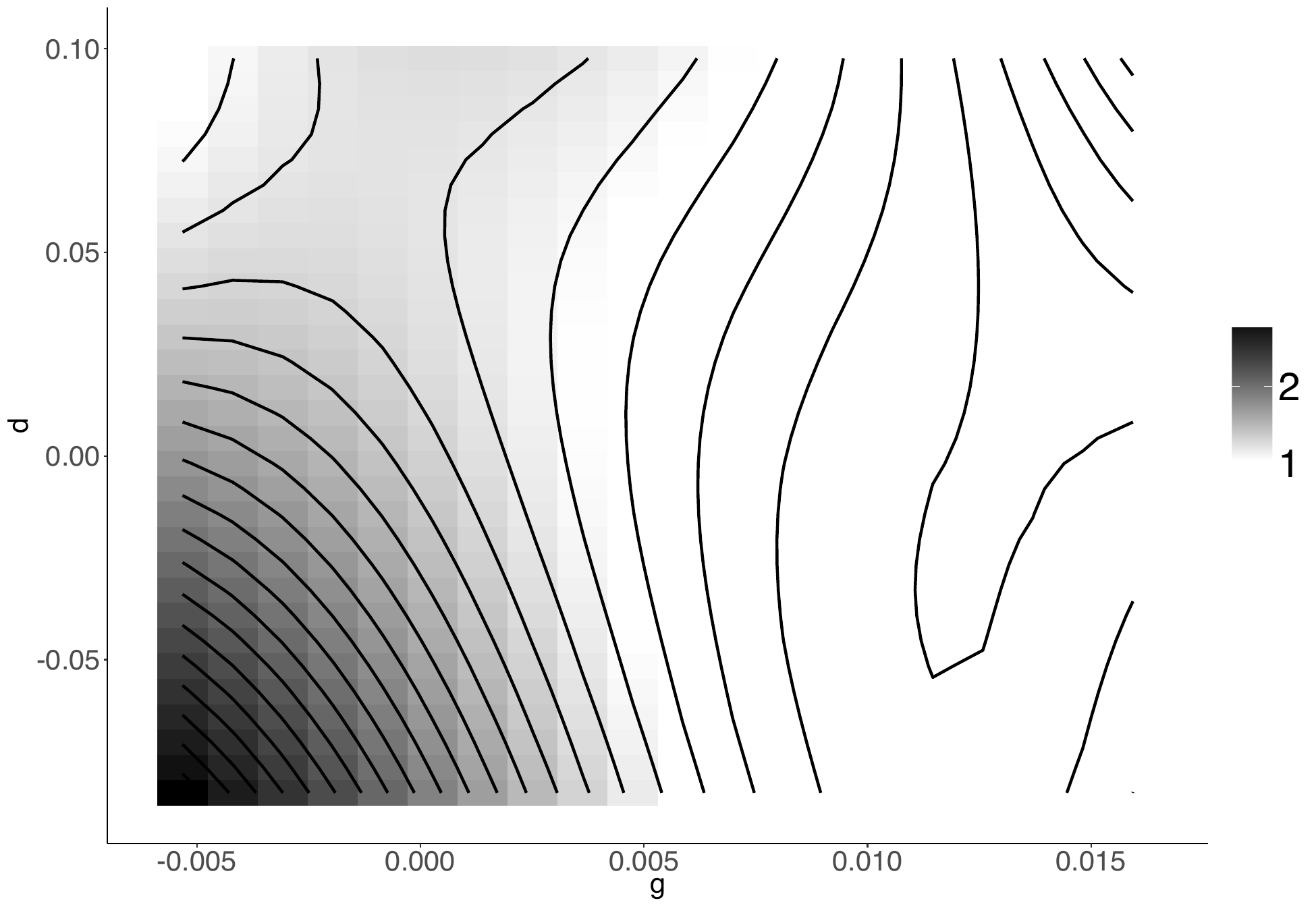}
  \caption{\footnotesize Contour plot of $\hat \phi^*(x)$ for $X_t = (g_t,d_t)$}
  \label{fig:emp:phistar_2}
\end{subfigure} \\[10pt]
\begin{subfigure}{.5\textwidth}
  \centering
  \includegraphics[width=\linewidth]{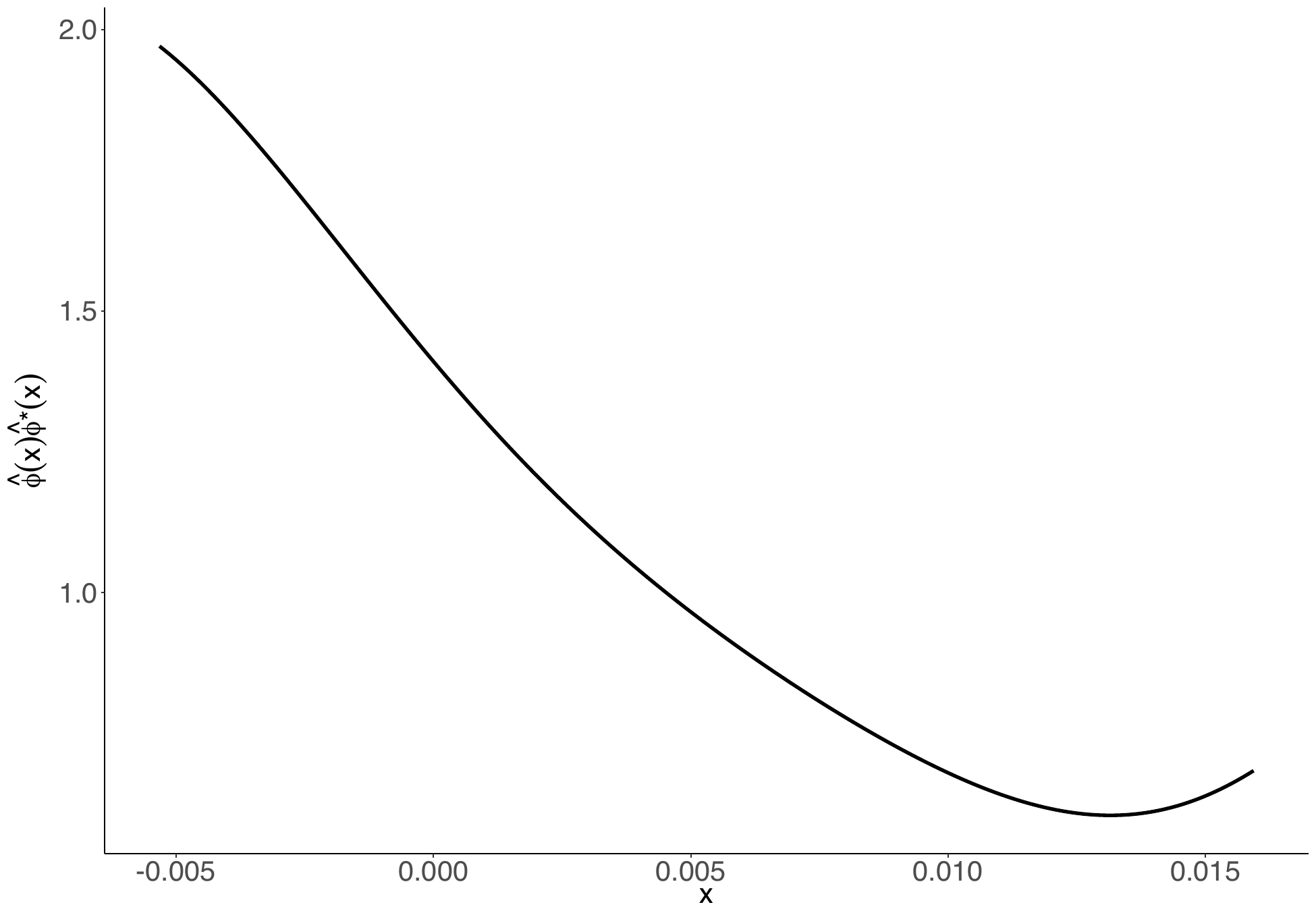}
  \caption{\footnotesize Plot of $\hat \phi(x)\hat \phi^*(x)$ for $X_t = g_t$}
  \label{fig:emp:rn_1}
\end{subfigure}%
\begin{subfigure}{.5\textwidth}
  \centering
  \includegraphics[width=\linewidth]{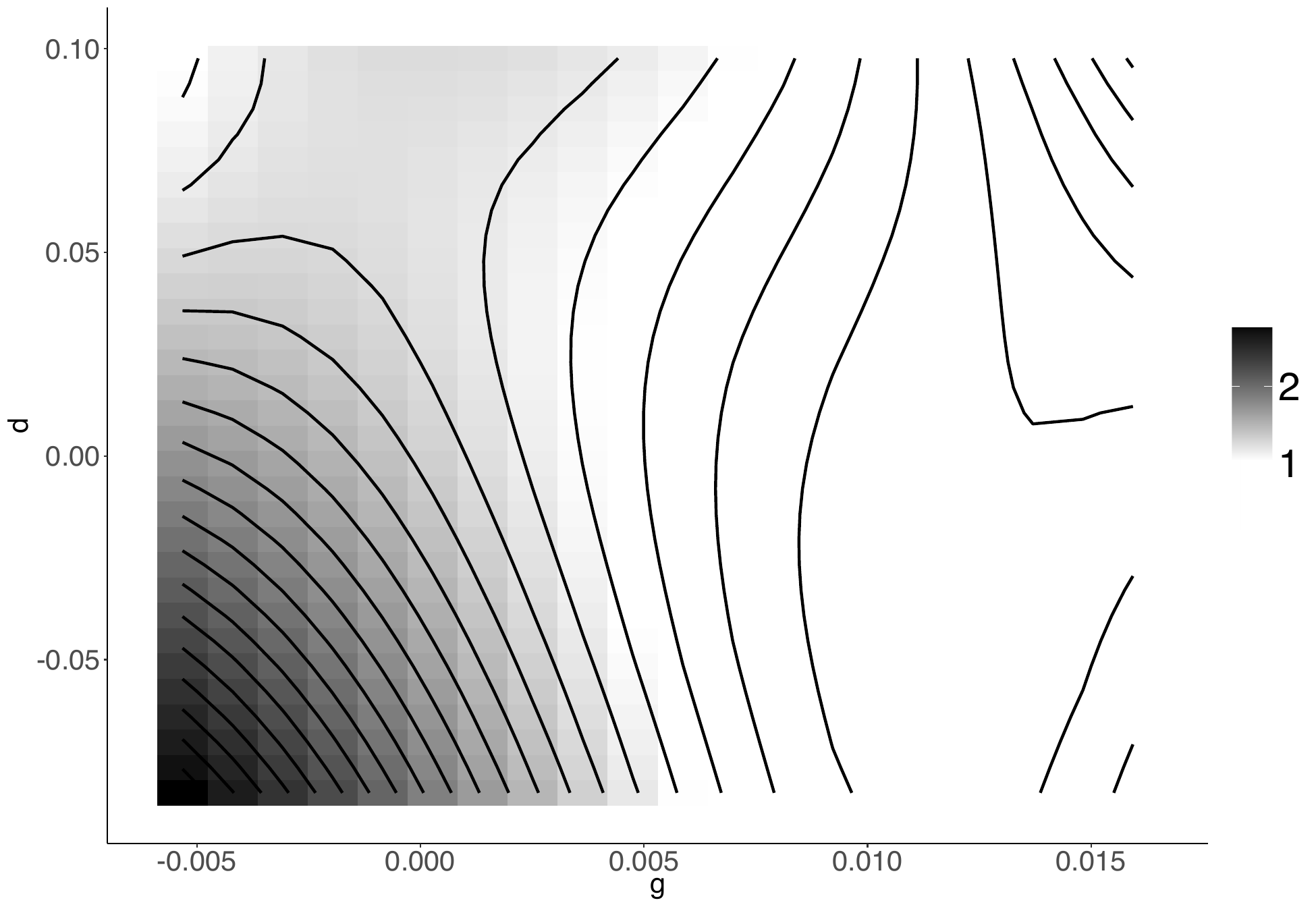}
  \caption{\footnotesize Contour plot of $\hat \phi(x) \hat \phi^*(x)$ for $X_t = (g_t,d_t)$}
  \label{fig:emp:rn_2}
\end{subfigure}
\begin{center}
\parbox{5.0in}{\caption{\label{fig:emp} \small Plots of $\hat \phi$ (upper panels), $\hat \phi^*$ (middle panels) and the estimated change of measure  $\hat \phi(x) \hat \phi^*(x)$ between the stationary distribution $Q$ and the distribution $\wt Q$ corresponding to $\wt{\mb E}$ (lower panels) under recursive preferences using the estimated preference parameters in the left panel of Table \ref{tab:emp}.}}
\end{center}
\end{figure}

Finally, we investigate the role of nonlinearities and non-Gaussianity in explaining certain features of the long-end of the term structure. Figure \ref{fig:emp:yield} presents nonparametric estimates of (a) the (quarterly) long-run yield and (b) the correlation between the logarithm of the permanent and transitory components, namely $\hat m_{t+1}^P = \log (\hat M^P_{t+1}/\hat M^P_t)$ and $\hat m^T_{t+1} = \log (\hat M^T_{t+1}/\hat M^T_t)$, recovered from the data on $X_t = (g_t,d_t)$ with $\beta = 0.994$ and $\gamma$ increased from $\gamma = 1$ to $\gamma = 35$. The nonparametric estimates are presented alongside estimates for two parametric specifications of the state process. The first assumes $X_t = (g_t,d_t)$ is a Gaussian VAR(1), i.e. $X_t - \mu = A (X_t - \mu) + e_{t+1}$ where the $e_{t+1}$ are i.i.d. $N(0,\Sigma)$. The second is a Gaussian AR(1) for log consumption growth with stochastic volatility:
\[
 g_{t+1} - \mu  = \kappa ( g_{t} - \mu ) + \sqrt{v_t} e_{t+1}\,,  \quad e_{t+1} \sim \mbox{ i.i.d. N$(0,1)$}
\]
where  $\{v_t\}$ is a first-order autoregressive gamma process (a discrete-time version of the Feller square-root process; see \cite{GourierouxJasiak2006}) so the state vector is $X_t = (g_t,v_t)$. We refer to the second specification as SV-AR(1). The long-run yield and the correlation  between $m^P_{t+1} = \log(M_{t+1}^P/M_t^P)$ and $m^T_{t+1} = \log(M_{t+1}^T/M_t^T)$ were obtained analytically as functions of $\beta$, $\gamma$, and the estimates of the VAR(1) and SV-AR(1) parameters.\footnote{The VAR(1) parameters are estimated by OLS. The SV-AR(1)  parameters are estimated via indirect inference using an AR(1) with GARCH(1,1) errors as an auxiliary model. Derivations and further details on estimation are in the supplementary material.}

\begin{figure}[t]
\centering
\begin{subfigure}{.5\textwidth}
  \centering
  \includegraphics[width=\linewidth]{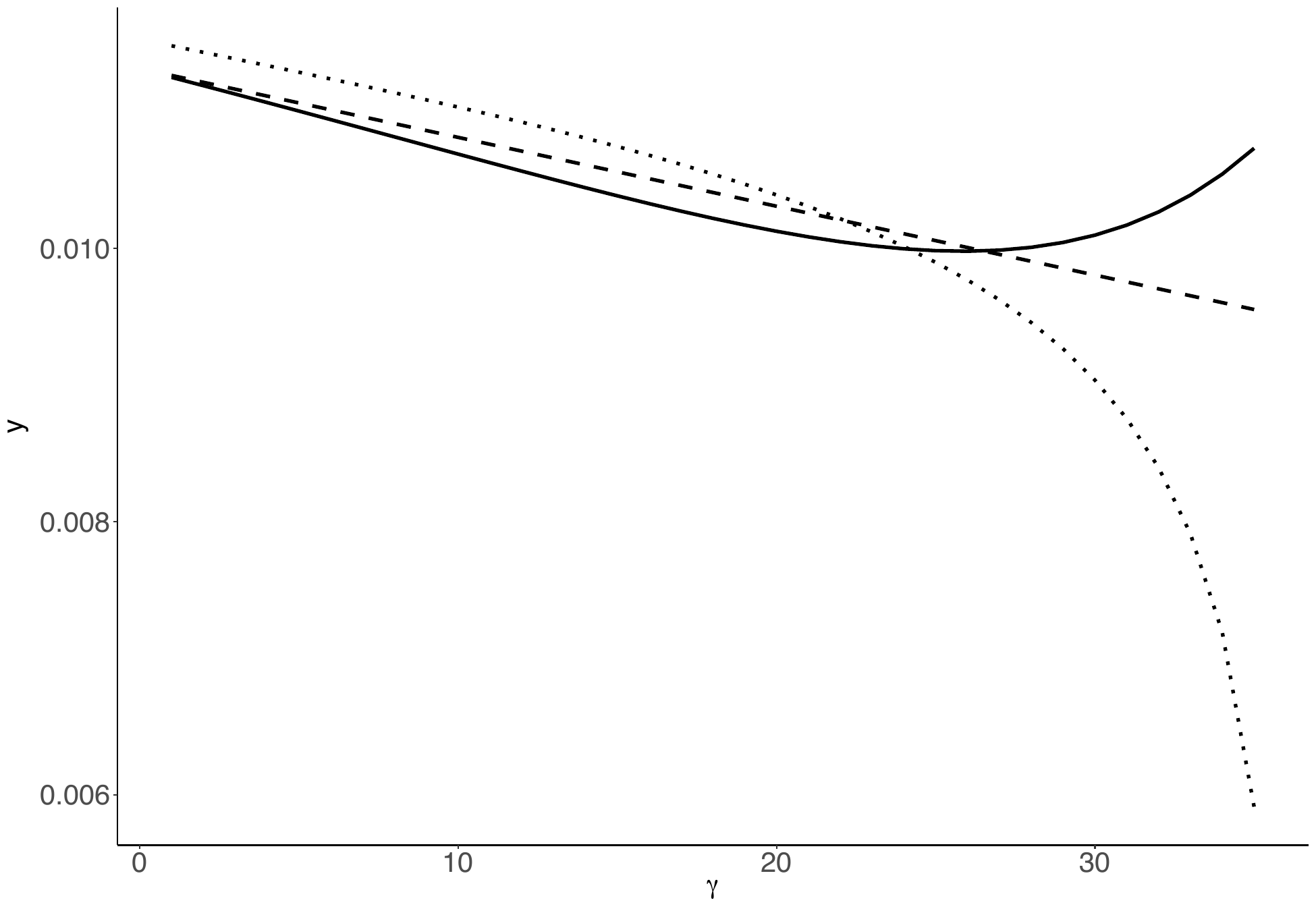}
  \caption{\footnotesize Long-run yield}
  \label{fig:emp:yld}
\end{subfigure}%
\begin{subfigure}{.5\textwidth}
  \centering
  \includegraphics[width=\linewidth]{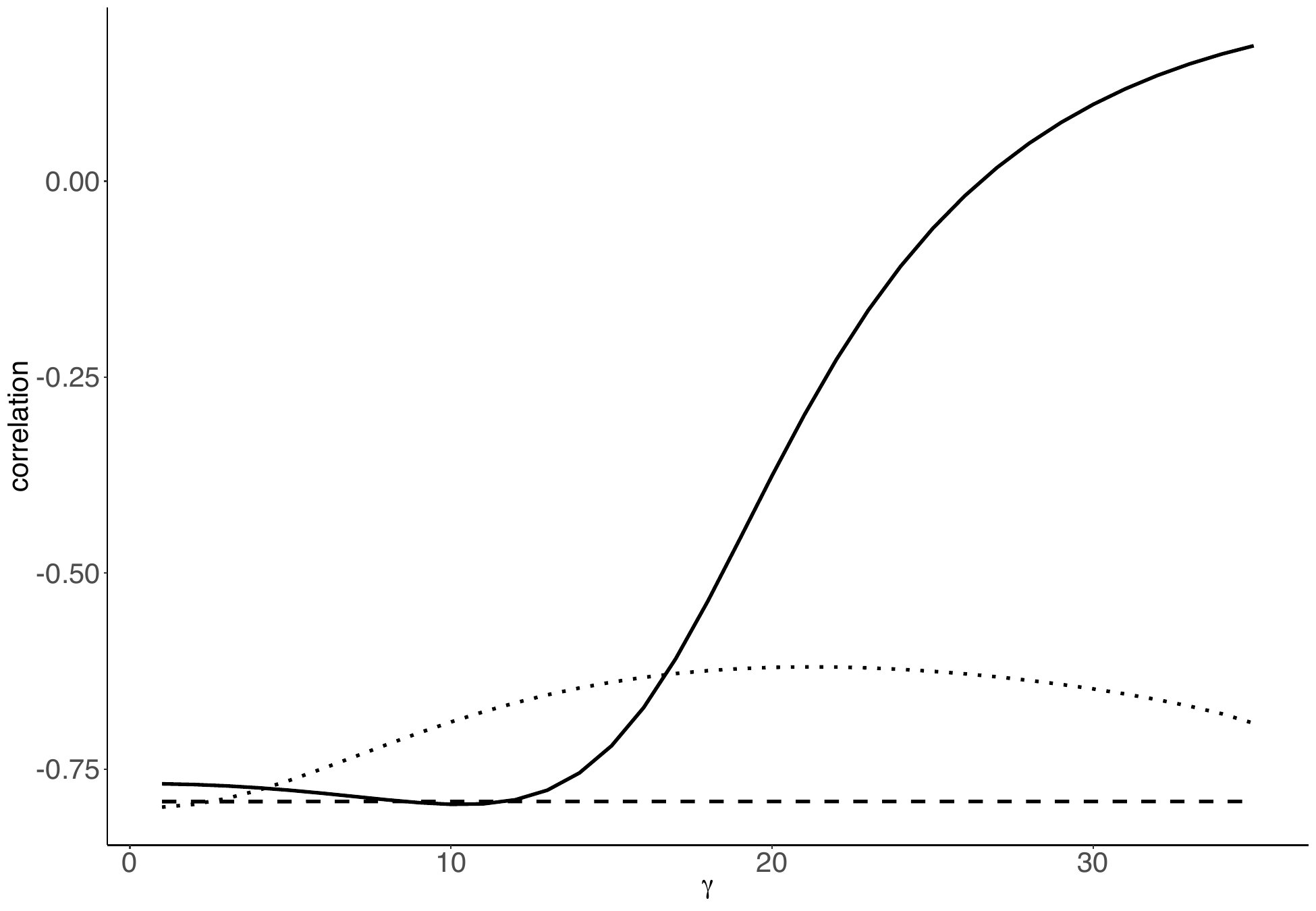}
  \caption{\footnotesize Correlation between $\hat m^P_{t+1}$ and $\hat m^T_{t+1}$}
  \label{fig:emp:cor}
\end{subfigure}
\begin{center}
\parbox{5.0in}{\caption{\label{fig:emp:yield} \small Solid lines: nonparametric estimates of the quarterly long-run yield and correlation between $\hat m_{t+1}^P$ and $\hat m^T_{t+1}$ under recursive preferences with $\beta = 0.994$ for different $\gamma$ with $X_t = (g_t,d_t)$. Also displayed: parametric estimates obtained from fitting a Gaussian VAR(1) to $X_t = (g_t,d_t)$ (dashed lines) and fitting a SV-AR(1) to $g_t$ (dotted lines).}}
\end{center}
\end{figure}

Figure \ref{fig:emp:yld} shows that the nonparametric estimates of the long-run yield are non-monotontic, whereas the parametric estimates are monotonically decreasing. This non-monotonicity is not apparent in the nonparametric estimates using $X_t = g_t$. It is also clear that the nonparametric estimates of the long-run yield are much larger for high $\gamma$ than the SV-AR(1) model.

Figure \ref{fig:emp:cor} displays the sample correlation of the nonparametric estimates $\hat m^P_{t+1}$ and $\hat m^T_{t+1}$ of the log permanent and transitory component series for different values of $\gamma$. This is compared with the correlation of the log permanent and transitory components $m_{t+1}^P$ and $m_{t+1}^T$ for the two parametric state process specifications. The estimated correlation of the nonparametric estimates is negative for low to moderate values of $\gamma$, but becomes positive for larger values of $\gamma$. Similar results are obtained using lower- and higher-dimensional bases. In contrast, the correlations for the parametric state process specifications are around the same level as the nonparametric estimates for low values of $\gamma$ but remain negative for larger values of $\gamma$. A recent literature has emphasized the role of positive dependence between the permanent and transitory components in explaining excess returns of long-term bonds \citep{BakshiChabiYo,BC-YG,BC-YG:rec}. Positive dependence also features in models in which the term structure of risk prices is downward sloping (see the example presented in section 7.2 in \cite{BorovickaHansen}). However, positive dependence is known to be difficult to generate via conventional preference specifications in workhorse models with exponentially-affine dynamics. Although the correlation is estimated imprecisely for large values of $\gamma$, this finding at least suggests that nonlinearities in state dynamics may have a role to play in explaining salient features of the long end of the yield curve.

\section{Conclusion} \label{s:conc}

This paper introduces an empirical framework to analyze the permanent and transitory components of SDF processes in the long-run factorization of \cite{AJ}, \cite{HS2009}, and \cite{Hansen2012}. We show how to estimate nonparametrically the solution to the Perron-Frobenius eigenfunction problem of \cite{HS2009} from time-series data on state variables and a SDF process. Our empirical framework allows researchers to (i) recover the time series of the estimated permanent and transitory components and investigate their properties and (ii) estimate the yield and the change of measure which characterize pricing over long investment horizons. This represents a useful contribution relative to existing empirical works which have established bounds on various moments of the permanent and transitory components as functions of asset returns, but have not extracted the components directly from data. The methodology is nonparametric in that it does not impose tight parametric restrictions on the dynamics of the state variables or the joint distribution of the state variables and the SDF process.

The main theoretical contributions of the paper are as follows. First, we establish consistency and convergence rates of the Perron-Frobenius eigenfunction estimators. Second, we establish asymptotic normality and some efficiency properties of the eigenvalue estimator and estimators of related functionals. Third, we introduce nonparametric estimators of the continuation value function in a class of models with recursive preferences by reinterpreting the value function recursion as a nonlinear Perron-Frobenius problem and we establish consistency and convergence rates of the value function estimators.

The econometric methodology may be extended and applied in several different ways. First, the methodology can be applied to study more general multiplicative functional processes such as the valuation and stochastic growth processes in \cite{HansenHeatonLi}, \cite{HS2009}, and \cite{Hansen2012}. Second, the methodology can be applied to models with latent state variables. The main consistency and convergence rate results (Theorems \ref{t:rate} and \ref{t:fpest}) are sufficiently general that they apply equally to such cases. Finally, our analysis was conducted within the context of structural models in which the SDF process was linked tightly to preferences. A further extension would be to apply the methodology to SDF processes which are extracted flexibly from panels of asset returns data.

\appendix

\section{Additional results on estimation}\label{ax:est}

\subsection{Bias and variance calculations for Theorem \ref{t:rate}} \label{ax:est:cgce}

The results in this subsection draw heavily on arguments from \cite{Gobetetal}.
The first result shows that the approximate solutions $\rho_k$, $\phi_k^{\phantom *}$ and $\phi_k^*$ from the eigenvector problem (\ref{e:gev}) are well defined and unique (i.e. up to sign and scale normalization) for all $k$ sufficiently large. 

\begin{lemma} \label{lem:exist}
Let Assumptions \ref{a:id} and \ref{a:bias} hold. Then: there exists $K \in \mb N$ such that for all $k \geq K$, the maximum eigenvalue $\rho_k$ of the eigenvector problem (\ref{e:gev}) is real and simple, and hence $(\mf M, \mf G)$ has unique right- and left-eigenvectors $c_k^{\phantom *}$ and $c_k^*$ corresponding to $\rho_k$.
\end{lemma}

\begin{lemma} \label{lem:bias}
Let Assumptions \ref{a:id} and \ref{a:bias} hold. Then:
\begin{enumerate} 
\item[(a)] $|\rho_k - \rho| = O(\delta_k^{\phantom *})$ 
\item[(b)] $\|\phi_k - \phi\| = O(\delta_k^{\phantom *})$
\item[(c)] $\|\phi_k^* - \phi^*\| = O(\delta_k^*)$
\end{enumerate}
where $\delta_k^{\phantom *}$ and $\delta_k^*$ are defined in display (\ref{e:deltas}). The rates should be understood to hold under the scale normalizations $\|\phi\| = 1$, $\| \phi_k\| = 1$, $\| \phi^*_{\phantom k} \| = 1$, and $\| \phi^*_k \| = 1$ and the sign normalizations $\langle \phi_k,\phi \rangle \geq 0$ and $\langle \phi_k^*,\phi^* \rangle \geq 0$.
\end{lemma}

The following result shows that the solutions $\hat \rho$, $\hat c^{\phantom *}$ and $\hat c^*$ to the sample eigenvector problem (\ref{e:est}) are well defined and unique with probability approaching one (wpa1).

\begin{lemma} \label{lem:exist:hat}
Let Assumptions \ref{a:id}--\ref{a:var} hold. Then: wpa1, the maximum eigenvalue $\hat \rho$ of the generalized eigenvector problem (\ref{e:est}) is real and simple, and hence $(\wh{\mf M},\wh{ \mf G})$ has unique right- and left-eigenvectors $\hat c$ and $\hat c^*$ corresponding to $\hat \rho$.
\end{lemma}

\begin{lemma}\label{lem:var}
Let Assumptions \ref{a:id}--\ref{a:var} hold. Then:
\begin{enumerate} 
\item[(a)] $|\hat \rho - \rho_k| = O_p(\eta_{n,k}^{\phantom *})$ 
\item[(b)] $\|\hat \phi - \phi_k\| = O_p(\eta_{n,k}^{\phantom *})$
\item[(c)] $\|\hat \phi^* - \phi_k^*\| = O_p(\eta_{n,k}^{*})$
\end{enumerate}
where $\eta_{n,k}^{\phantom *}$ and $\eta_{n,k}^*$ are defined in display (\ref{e:etas}). The rates should be understood to hold under the scale normalizations $\|\hat \phi\| = 1$, $\| \phi_k\| = 1$, $\| \hat \phi^*_{\phantom k} \| = 1$ and $\| \phi^*_k \| = 1$ and the sign normalizations $\langle \hat \phi,\phi_k\rangle \geq 0$ and $\langle \hat \phi^* ,\phi^*_k \rangle \geq 0$.
\end{lemma}

\subsection{Bias and variance calculations for Theorem \ref{t:fpest}}\label{ax:est:fp}

The following two Lemmas apply known results from the literature on the solution of nonlinear equations by projection methods (see, e.g., Chapter 19 of \cite{Kras}). 
The first result shows that $h_k$ is well defined for all $k$ sufficiently large.

\begin{lemma}\label{lem:fp:exist}
Let Assumptions \ref{a:fp:exist} and \ref{a:fp:bias}(b) hold. Then: there exists $\varepsilon > 0$ and $K \in \mb N$ such that for all $k \geq K$ the projected fixed-point problem (\ref{e:pfpe}) has at least one solution $h_k$ in the ball $N_k = \{ \psi \in B_k : \|\psi - h\| < \varepsilon\}$. 
\end{lemma}

\begin{remark} \label{rmk:nbhd}
Although the ball $N_k$ may contain multiple solutions $h_k$ of the projected fixed-point problem (\ref{e:pfpe}), under the conditions of Lemma \ref{lem:fp:exist} we have that $\sup_{h_k \in H_k} \|h_k - h\| = o(1)$ where $H_k$ denotes the set of all solutions to (\ref{e:pfpe}) in $N_k$.
\end{remark}

\begin{remark} \label{rmk:fpunique}
If Assumption \ref{a:fp:exist}(c) is strengthened to require that $\mb T$ is continuously Fr\'echet differentiable at $h$ with $r(\mb D_h) < 1$ then there exists $K \in \mb N$ and $\varepsilon > 0$ such that for all $k \geq K$ the projected fixed-point problem (\ref{e:pfpe}) has a unique  solution $h_k$ in the ball $N_k$.
\end{remark}

In view of Remark \ref{rmk:nbhd}, in what follows we let $h_k$ be any one of the solutions to (\ref{e:pfpe}) in $N_k$ (or the unique solution under the additional assumption of continuous Fr\'echet differentiability of $\mb T$ at $h$). Let $\chi_k = h_k /\|h_k\|$ and $\lambda_k = \| \Pi_k \mb T \chi_k\|$.

\begin{lemma}\label{lem:bias:fp}
Let Assumptions \ref{a:fp:exist} and \ref{a:fp:bias} hold. Then:
\begin{enumerate}
\item[(a)] $|\lambda_k - \lambda| = O( \tau_k )$
\item[(b)] $\|\chi_k - \chi\| = O( \tau_k )$
\item[(c)] $\|h_k - h\| = O( \tau_k )$.
\end{enumerate}
\end{lemma}

We now show that, wpa1, the sample fixed-point problem has a solution $\hat v$ for which $\hat h(x) = \hat v'b^k(x)$ belongs to $N_k$. We then  derive convergence rates of the estimators formed using $\hat v$ (see display (\ref{e:fpest})). The following two results are new.

\begin{lemma}\label{lem:fphat:exist}
Let Assumptions \ref{a:fp:exist}--\ref{a:fp:var} hold. Then: wpa1, there exists a fixed point $\hat v$ of $\wh{\mf G}^{-1}\wh{\mf T}$ such that the function $\hat h(x) = b^k(x)' \hat v$ belongs to $N_k$. Moreover, $\|\hat h - h\| = o_p(1)$.
\end{lemma}

\begin{remark} \label{rmk:nbhd:2}
Although there may exist multiple fixed points $\hat v$ of $\wh{\mf G}^{-1}\wh{\mf T}$ for which $\hat h(x) = \hat v'b^k(x)$ belongs to $N_k$, under the conditions of Lemma \ref{lem:fphat:exist} we have that $\sup_{\hat h_k \in \hat H_k} \|\hat h - h\| = o_p(1)$ where $\hat H_k$ denotes the set of all such $b^k(x)'\hat v$ belonging to $N_k$.
\end{remark}

In view of Remark \ref{rmk:nbhd:2}, the following lemma applies to estimators $\hat \lambda$, $\hat \chi$, and $\hat h$ in (\ref{e:fpest}) formed from any fixed point $\hat v$ of $\wh{\mf G}^{-1}\wh{\mf T}$ for which $b^k(x)'\hat v \in N_k$.

\begin{lemma}\label{lem:var:fp}
Let Assumptions \ref{a:fp:exist}--\ref{a:fp:var} hold. Then: 
\begin{enumerate}
\item[(a)] $|\hat \lambda - \lambda_k| = O_p( \nu_{n,k} ) + o_p(\tau_k) $
\item[(b)] $\|\hat \chi - \chi_k \| = O_p( \nu_{n,k} ) + o_p(\tau_k)$
\item[(c)] $\|\hat h - h_k\| = O_p(  \nu_{n,k} ) + o_p(\tau_k)$.
\end{enumerate}
\end{lemma}

\section{Additional results on inference} \label{ax:inf}

\subsection{Asymptotic normality of long-run entropy estimators}

Here we consider the asymptotic distribution of the estimator $\hat L$ of the entropy of the permanent component of the SDF. In Case 1, the estimator of the long-run entropy is:
\[
 \hat L = \log \hat \rho - \frac{1}{n} \sum_{t=0}^{n-1} \log m(X_t,X_{t+1})\,.
\]
Recall that $\psi_{\rho,t} = \psi_\rho(X_t,X_{t+1})$ where the influence function $\psi_\rho$ is defined in (\ref{e:inf:def}). Define:
\[
 \psi_{lm}(x_t,x_{t+1}) = \log m(x_t,x_{t+1}) - \mb E[\log m(X_t,X_{t+1})]
\]
set $\psi_{lm,t} = \psi_{lm}(X_t,X_{t+1})$. Let $\hbar = (\rho^{-1} \, , -1)'$.

\begin{proposition}\label{p:asydist:L:1}
Let the assumptions of Theorem \ref{t:asydist:1} hold and $\frac{1}{\sqrt n} \sum_{t=0}^{n-1} (\psi_{\rho,t},\psi_{lm,t})' \to_d N(0,W)$ for some finite matrix $W$. Then: 
\[
 \sqrt n (\hat L - L) \to_d N(0,V_L)
\]
where $V_L = \hbar' W \hbar$.
\end{proposition}

In the preceding proposition, $V_L$ will be the long-run variance:
\[
 V_L = \sum_{t \in \mb Z} \mr{Cov}(\psi_L(X_0,X_1),\psi_L(X_t,X_{t+1}))
\]
where $\psi_L(X_t,X_{t+1}) = \rho^{-1} \psi_\rho(X_t,X_{t+1}) - \psi_{lm}(X_t,X_{t+1})$. Theorem \ref{t:eff:main} below shows that $V_L$ is the semiparametric efficiency bound for $L$.

In Case 2, the estimator of the long-run entropy is:
\[
 \hat L = \log \hat \rho - \frac{1}{n} \sum_{t=0}^{n-1} \log m(X_t,X_{t+1},\hat \alpha)\,.
\]
As with asymptotic normality of $\hat \rho$, the asymptotic distribution of $\hat L$ will depend on the manner in which $\hat \alpha$ was estimated. 
For brevity, we just consider the parametric case studied in Theorem \ref{t:asydist:2a}. Let $\psi_{lm}$ and $\psi_{lm,t}$ be as previously defined with $m(x_t,x_{t+1}) = m(x_t,x_{t+1},\alpha_0)$. Recall $\psi_{\alpha,t}$ from Assumption \ref{a:parametric} and  define: 
\[
 \hbar_{[\mr{2a}]} = \left( \rho^{-1} \; , \;  \mb E \left[\left( \frac{\phi^*(X_t) \phi(X_{t+1})}{\rho} - \frac{1}{m(X_t,X_{t+1},\alpha)} \right) \frac{\partial m(X_t,X_{t+1},\alpha)}{\partial \alpha'} \right]  \; , \; -1 \right)'\,.
\]

\begin{proposition}\label{p:asydist:L:2a}
Let the assumptions of Theorem \ref{t:asydist:2a} hold. Also let (a) there exist a neighborhood $N_1$ of $\alpha_0$ upon which the function $\log m(x_0,x_1,\alpha)$ is continuously differentiable in $\alpha$ for all $(x_0,x_1) \in \mc X^2$ with: 
\[
 \mb E \bigg[ \sup_{\alpha \in N_1} \left\| \frac{1}{m(x_0,x_1,\alpha)} \frac{\partial m(x_0,x_1,\alpha) }{\partial \alpha} \right\| \bigg] < \infty
\]
and (b) $\frac{1}{\sqrt n} \sum_{t=0}^{n-1} (\psi_{\rho,t}^{\phantom \prime},\psi_{\alpha,t}',\psi_{lm,t}^{\phantom \prime})'  \to_d N(0,W_{[\mr{2a}]})$ for some finite matrix $W_{[\mr{2a}]}$. Then: 
\[
 \sqrt n (\hat L - L) \to_d N(0,V_L^{[\mr{2a}]})
\]
where $V_L^{[\mr{2a}]} = \hbar_{[\mr{2a}]}'W_{[\mr{2a}]}^{\phantom \prime} \hbar_{[\mr{2a}]}^{\phantom \prime}$.
\end{proposition}

\subsection{Semiparametric efficiency bounds in Case 1}

Let $P_n(x,A) = \Pr( X_{t+n} \in A | X_t = x)$ denote the $n$-step transition probability of $X$ for any Borel set $A$. We say that  $\{X_t\}_{t \in \mb Z}$ is \emph{uniformly ergodic} if:
\[
 \lim_{n \to \infty} \sup_{x \in \mc X} \| P_n(x,\cdot) - Q\|_{TV} =0
\]
where $\|\cdot\|_{TV}$ denotes total variation norm and $Q$ denotes the stationary distribution of $X$. 

\begin{assumption}\label{a:eff}
$\{X_t\}_{t \in \mb Z}$ is uniformly ergodic.
\end{assumption}

Sufficient conditions for Assumption \ref{a:eff}, such as Doeblin's condition, are well known.  Assumption \ref{a:eff} also implies that $\{X_t\}_{t \in \mb Z}$ is exponentially phi-mixing \cite[pp. 367--368]{IL}, and therefore exponentially beta- and rho-mixing.

\begin{theorem}\label{t:eff:main}
(1) Let Assumption \ref{a:id}, \ref{a:asydist}(c), and \ref{a:eff} hold and let $h : \mb R \to \mb R$ be continuously differentiable at $\rho$ with $h'(\rho) \neq 0$. Then: the efficiency bound for $h(\rho)$ is $h'(\rho)^2V_\rho$. \\
(2) If, in addition, $\mb E[(\log m(X_t,X_{t+1}))^2] < \infty$, then: the efficiency bound for $L$ is $V_L$.
\end{theorem}

\subsection{Sieve perturbation expansion} \label{ax:est:inf}

The following result shows that $\hat \rho - \rho_k$ behaves as a linear functional of $\wh{\mf M} - \rho_k \wh{\mf G}$ and is used to derive the asymptotic distribution of $\hat \rho$ in Theorem \ref{t:asydist:1}. It follows from Assumption \ref{a:var} that we can choose sequences of positive constants $\eta_{n,k,1}$ and $\eta_{n,k,2}$ such that:
\[
 \|\wh {\mf G}^o - \mf I\| = O_p(\eta_{n,k,1}) \quad \mbox{and} \quad \|\wh{\mf M}^o - {\mf M}^o\| = O_p(\eta_{n,k,2})
\]
with $\eta_{n,k,1} = o(1)$ and $\eta_{n,k,2} = o(1)$ as $n,k \to \infty$. Let $c_k$ and $c_k^*$ be normalized so that $\|\mf G^{1/2} c_k\| = 1$ and $c_k^{* \prime } \mf G c_k^{\phantom *} = 1$ (equivalent to $\|\phi_k\| = 1$ and $\langle \phi^*_k,\phi_k^{\phantom *} \rangle = 1$).

\begin{lemma}\label{lem:expansion}
Let Assumptions \ref{a:id}--\ref{a:var} hold. Then: 
\[
 \hat \rho - \rho_k =  c_k^{*\prime}  ( \wh{\mf M} - \rho_k \wh{\mf G}  )  c_k + O_p( \eta_{n,k,1} \times (  \eta_{n,k,1} +  \eta_{n,k,2})  ) \,.
\]
In particular, if $\|\wh {\mf G}^o - \mf I\| = o_p(n^{-1/4})$ and $\|\wh{\mf M}^o - {\mf M}^o\| = o_p(n^{-1/4})$ then: 
\[
  \sqrt n (\hat \rho - \rho_k)  = \sqrt n  c_k^{*\prime}  ( \wh{\mf M} - \rho_k \wh{\mf G}  )  c_k + o_p(1)\,.
\]
\end{lemma}

\newpage

{ 
\small 
\singlespacing
\putbib
}
\end{bibunit}

\begin{bibunit}

\newpage
\clearpage
\pagenumbering{arabic}\renewcommand{\thepage}{\arabic{page}}
\setcounter{equation}{0}
\renewcommand{\theequation}{S.\arabic{equation}}

\begin{center}
{\Large Supplement to ``Nonparametric Stochastic Discount Factor Decomposition''}

\vskip 24pt
{\large Timothy M. Christensen}

\end{center}

\vskip 8pt

This supplementary material contains sufficient conditions for several assumptions in Sections \ref{s:est} and \ref{s:recursive} and proofs of all results in the main text.

\section{Some sufficient conditions}\label{ax:suff}

This appendix presents sufficient conditions for Assumptions \ref{a:var}, \ref{a:asydist}(b) and \ref{a:fp:var} and bounds for the terms $\eta_{n,k}$ and $\eta_{n,k}^*$ in display (\ref{e:etas}) and $\nu_{n,k}$ in display (\ref{e:nudef}). Proofs of results in this appendix are contained in the Online Appendix.

\subsection{Sufficient conditions for Assumptions \ref{a:var} and \ref{a:asydist}(b)} \label{ax:est:mat}

We assume that the state process $X = \{X_t : t \in T\}$ is either beta-mixing or rho-mixing. 
The beta-mixing coefficient between two $\sigma $-algebras ${\mathcal{A}}$ and ${\mathcal{B}}$ is:
\[
 2\beta ({\mathcal{A}},{\mathcal{B}})=\sup \sum_{(i,j)\in I\times J}|{\mathbb{P}}(A_{i}\cap B_{j})-{\mathbb{P}}(A_{i}){\mathbb{P}}(B_{j})|
\]
with the supremum taken over all $\mc A$-measurable finite partitions $\{A_{i}\}_{i\in I}$ and $\mc B$-measurable finite partitions $\{B_{j}\}_{j\in J}$. The beta-mixing coefficients of $X$ are defined as:
\[
\beta_q =\sup_{t}\beta (\sigma (\ldots ,X_{t-1},X_{t}),\sigma (X_{t+q},X_{t+q+1},\ldots ))\,.
\]
We say that $X$ is \emph{exponentially beta-mixing} if $\beta_q \leq Ce^{-c q}$ for some $C,c> 0$. The rho-mixing coefficients of $X$ are defined as:
\[
 \rho_q = \sup_{ \psi \in L^2 : \mb E[\psi] = 0, \|\psi\| = 1}  \mb E\big[ \mb E[\psi(X_{t+q})|X_t]^2\big]^{1/2}\,.
\] 
We say that $X$ is \emph{exponentially rho-mixing} if $\rho_q \leq e^{-c q}$ for some $c > 0$.

We use the sequence  $\xi_k = \sup_x \| \mf G^{-1/2} b^k(x)\|$ to bound convergence rates. When $X$ has bounded rectangular support and $Q$ has a density that is bounded away from $0$ and $\infty$, $\xi_k$ is known to be $O(\sqrt k)$ for (tensor-product) spline, cosine, and certain wavelet bases and $O(k)$ for (tensor-product) polynomial series \citep{Newey1997,ChenChristensen-reg}. It is also possible to derive alternative sufficient conditions in terms of higher moments of $\| \mf G^{-1/2} b^k(X_t)\|$ (instead of $\sup_x \| \mf G^{-1/2} b^k(x)\|$) by extending arguments in \cite{Hansen2015WP} to accommodate weakly-dependent data and asymmetric matrices.

\subsubsection{Sufficient conditions in Case 1}

The first result below uses an exponential inequality for weakly-dependent random matrices from \cite{ChenChristensen-reg}.  The second extends arguments from \cite{Gobetetal}.

\begin{lemma}\label{lem:beta:1}
Let the following hold:
\begin{enumerate}
\item[(a)] $X$ is exponentially beta-mixing
\item[(b)] $\mb E[m(X_t,X_{t+1})^r] < \infty$ for some $r>2$
\item[(c)] $\xi_k^{2+4/r} (\log n)^2/n = o(1)$.
\end{enumerate}
Then: (1) Assumption \ref{a:var} holds. \\
(2) We may take $\eta_{n,k}^{\phantom *} = \eta_{n,k}^* =  \xi_k^{1+2/r} (\log n)/\sqrt n$ in display (\ref{e:etas}). \\
(3) If, in addition, $\xi_k^{4+8/r}(\log n)^4/n = o(1)$ then Assumption \ref{a:asydist}(b) holds.
\end{lemma}

\begin{lemma}\label{lem:rho:1}
Let the following hold:
\begin{enumerate}
\item[(a)] $X$ is exponentially rho-mixing
\item[(b)] $\mb E[m(X_t,X_{t+1})^r] < \infty$ for some $r>2$
\item[(c)] $\xi_k^{2+4/r} k/n = o(1)$.
\end{enumerate}
Then: (1) Assumption \ref{a:var} holds. \\
(2) We may take $\eta_{n,k}^{\phantom *} = \eta_{n,k}^* =  \xi_k^{1+2/r} /\sqrt n$ in display (\ref{e:etas}). \\
(3) If, in addition, $\xi_k^{4+8/r}k^2/n = o(1)$ then Assumption \ref{a:asydist}(b) also holds.
\end{lemma}

\subsubsection{Sufficient conditions in Case 2 with parametric first-stage}

The following lemma presents one set of sufficient conditions for Assumption \ref{a:var} and \ref{a:asydist}(b) when $\alpha_0 \in \mc A \subseteq \mb R^{{d_\alpha}}$ is a finite-dimensional parameter.

\begin{lemma}\label{lem:beta:2}
Let the conditions of Lemma \ref{lem:beta:1} hold for $m(x_0,x_1) = m(x_0,x_1;\alpha_0)$, and let:
\begin{enumerate}
\item[(a)] $\|\hat \alpha - \alpha_0\| = O_p(n^{-1/2})$
\item[(b)] $m(x_0,x_1;\alpha)$ be continuously differentiable in $\alpha$ on a neighborhood $N$ of $\alpha_0$ for all $(x_0,x_1) \in \mc X^2$ and let there exist a function $\bar m : \mc X^2 \to \mb R$ with $\mb E[\bar m(X_t,X_{t+1})^2 ] < \infty$ such that:
\[
 \sup_{\alpha \in N} \left\| \frac{\partial m(x_0,x_1;\alpha)}{\partial \alpha} \right\| \leq \bar m(x_0,x_1) \quad \mbox{for all $(x_0,x_1) \in \mc X^2$.}
\]
\end{enumerate}
Then: (1) Assumption \ref{a:var} holds. \\
(2) We may take $\eta_{n,k}^{\phantom *} = \eta_{n,k}^* =  \xi_k^{1+2/r} (\log n)/\sqrt n$ in display (\ref{e:etas}). \\
(3) If, in addition, $\xi_k^{4+8/r}(\log n)^4/n = o(1)$ then Assumption \ref{a:asydist}(b) holds.
\end{lemma}

The conditions on $k$ and bounds for $\eta_{n,k}$ and $\eta_{n,k}^*$ are the same as Lemma \ref{lem:beta:1}. Therefore, here first-stage estimation of $\alpha$ does not reduce the convergence rates of $\wh{\mf G}$ and $\wh{\mf M}$ relative to Case 1.

\subsubsection{Sufficient conditions in Case 2 with semi/nonparametric first-stage}

We now present one set of sufficient conditions for Assumptions \ref{a:var} and \ref{a:asydist}(b) when $\alpha_0 \in \mc A \subseteq \mb A$ is an  infinite-dimensional parameter and the parameter space is $\mc A \subseteq \mb A$ (a Banach space) equipped with a norm $\|\cdot\|_{\mc A}$. This includes the case in which $\alpha$ is a function, i.e. $\alpha = h$ with $\mb A = \mb H$ a function space, and the case in which $\alpha$ consists of both finite-dimensional and function parts, i.e. $\alpha = (\theta,h)$ with $\mb A = \Theta \times \mb H$ where $\Theta \subseteq \mb R^{\dim (\theta)}$. 

For each $\alpha \in \mc A$ we define $\mb M^{(\alpha)}$ as the operator $\mb M^{(\alpha)}\psi(x) = \mb E[m(X_t,X_{t+1};\alpha) \psi(X_{t+1}) | X_t = x]$ with the understanding that $\mb M^{(\alpha_0)} = \mb M$. 
Let $\mc M = \{ m(x_0,x_1;\alpha) - m(x_0,x_1;\alpha_0) : \alpha \in \mc A\}$. We say $\mc M$ has an envelope function $E$ if there exists some measurable $E : \mc X^2 \to [1,\infty)$ such that $|m(x_0,x_1)| \leq E(x_0,x_1)$ for every $(x_0,x_1) \in \mc X$ and $m \in \mc M$. Let $\mc M^* = \{ m /E : m \in \mc M\}$. The functions in $\mc M^*$ are clearly bounded by $\pm 1$. Let $N_{[\,\,]}(u,\mc M^*,\|\cdot\|_{p})$ denote the entropy with bracketing of $\mc M^*$ with respect to the $L^p$ norm $\|\cdot\|_{p}$.  Finally, let $\ell^*(\alpha) = \|\mb M^{(\alpha)} - \mb M\|$ and observe that $\ell^*(\alpha_0) = 0$.

\begin{lemma}\label{lem:beta:3}
Let the conditions of Lemma \ref{lem:beta:1} hold for $m(x_0,x_1) = m(x_0,x_1;\alpha_0)$, and let:
\begin{enumerate}
\item[(a)] $\mc M$ have envelope function $E$ with $\|E\|_{4s} < \infty$ for some $s > 1$
\item[(b)] $\log N_{[\,\,]}(u,\mc M^*,\|\cdot\|_{\frac{4sv}{2s-v}}) \leq C_{[\,\,]} u^{-2\zeta}$ for some constants $C_{[\,\,]} > 0$, $\zeta \in (0,1)$ and $v \in (1,2s)$
\item[(c)] $\ell^*(\alpha)$ is pathwise differentiable at $\alpha_0$ with $| \ell^*(\alpha) - \ell^*(\alpha_0) - \dot \ell^*_{\alpha_0} [\alpha - \alpha_0]| = O(\|\alpha - \alpha_0\|^2_{\mc A})$, $\|\hat \alpha - \alpha_0\|_{\mc A} = o_p(n^{-1/4})$ and $\sqrt n \dot \ell_{\alpha_0}^*[\hat \alpha - \alpha_0] = O_p(1)$
\item[(d)] $\xi_k^{4-\frac{2s-v}{sv}} (k \log k) /n = o(1)$, $\xi_k^{\zeta \frac{2s-v}{2sv}} = O( \sqrt {k \log k})$ and $(\log n) = O(\xi_k^{1/3})$.
\end{enumerate}
Then: (1) Assumption \ref{a:var} holds. \\
(2) We may take $\eta_{n,k}^{\phantom *} = \eta_{n,k}^* =  \xi_k^{1+2/r} (\log n)/\sqrt n + \xi_k^{2-\frac{2s-v}{2sv}}\sqrt{ (k \log k) /n}$ in display (\ref{e:etas}). \\
(3) If, in addition, $[\xi_k^{4+8/r}(\log n)^4 + \xi_k^{8-\frac{4s-2v}{sv}} (k \log k)^2 ]/n = o(1)$ then Assumption \ref{a:asydist}(b) holds.
\end{lemma}

Note that the condition $\xi_k^{\zeta \frac{2s-v}{2sv}} = O( \sqrt {k \log k})$ is trivially satisfied when $\xi_k = O(\sqrt k)$. 

\subsection{Sufficient conditions for Assumption \ref{a:fp:var}} \label{ax:est:mat:fp}

The following is one set of sufficient conditions for Assumption \ref{a:fp:var} assuming beta-mixing. Recall that $\xi_k = \sup_x \| \mf G^{-1/2} b^k(x)\|$. 

\begin{lemma}\label{lem:beta:4}
Let the following hold:
\begin{enumerate}
\item[(a)] $X$ is exponentially beta-mixing
\item[(b)] $\mb E[(G_{t+1}^{1-\gamma})^{2s}] < \infty$ for some $s > 1$
\item[(c)] $[ \xi_k^2 (\log n)^2  + \xi_k^{2+2\beta}  k]/n= o(1)$ and $(\log n)^{\frac{2s-1}{s-1}}k/n = o(1)$.
\end{enumerate}
Then: (1) Assumption \ref{a:fp:var} holds. \\
(2) We may take $\nu_{n,k} = \xi_k^{1+\beta} \sqrt{k/n} + \xi_k(\log n)/\sqrt n$ in display (\ref{e:nudef}).
\end{lemma}

\section{Proofs of results in the main text} \label{ax:proofs}

{\bf Notation:} For $v \in \mb R^k$, define: 
\[
 \|v\|_{\mf G}^2 = v' \mf G_k v
\]
or equivalently $\|v\|_{\mf G} = \|\mf G_k^{1/2} v\|$. For any matrix $\mf A \in \mb R^{k \times k}$ we define:
\[
 \|\mf A\|_{\mf G} = \sup\{ \|\mf A v\|_{\mf G} : v \in \mb R^k, \|v\|_{\mf G} = 1\}\,.
\]
We also define the inner product weighted by $\mf G_k$, namely $\langle u,v \rangle_{\mf G} = u' \mf G_k v$. The inner product $\langle \cdot,\cdot \rangle_{\mf G}$ and its norm $\|\cdot\|_{\mf G}$ are germane for studying convergence of the matrix estimators, as $(\mb R^k,\langle \cdot,\cdot\rangle_{\mf G})$ is isometrically isomorphic to $(B_k,\langle \cdot,\cdot \rangle)$. The notation $a_n \lesssim b_n$ for two positive sequences $a_n $ and $b_n$ means that there exists a finite positive constant $C$ such that $a_n \leq C b_n$ for all $n$ sufficiently large; $a_n \asymp b_n$ means $a_n \lesssim b_n$ and $b_n \lesssim a_n$.

\subsection{Proofs of results in Sections \ref{s:setup}, \ref{s:est} and \ref{s:recursive}}

\begin{proof}[Proof of Proposition \ref{p:id}]
Theorem V.6.6 of \cite{Schaefer1974} implies, in view of Assumption \ref{a:id:0}, that $\rho := r(\mb M) > 0$ and that $\mb M$ has a unique positive eigenfunction $\phi \in L^2$ corresponding to $\rho$. Applying the result to $\mb M^*$ in place of $\mb M$ guarantees existence of $\phi^* \in L^2$. This proves part (a). Theorem V.6.6 of \cite{Schaefer1974} also implies that $\rho$ is isolated and the largest eigenvalue of $\mb M$. Theorem V.5.2(iii) of \cite{Schaefer1974}, in turn, implies that $\rho$ is simple, completing the proof of part (c). Theorem V.5.2(iv) of \cite{Schaefer1974} implies that $\phi$ is the unique positive solution to (\ref{e:pev}). The same result applied to $\mb M^*$ in place of $\mb M$ guarantees uniqueness of $\phi^*$, proving part (b).  
Part (d) follows from Proposition \ref{p:lr}.
\end{proof}

\begin{proof}[Proof of Theorem \ref{t:rate}]
Immediate from Lemmas \ref{lem:bias} and \ref{lem:var}.
\end{proof}

\begin{proof}[Proof of Corollary \ref{c:rate}]
We first verify Assumption \ref{a:bias}. By Theorem 12.8 of \cite{Schumaker2007} and (ii)--(iv), for each $\psi \in L^2$ there exists a $h_k(\mb M \psi) \in B_k$ such that: 
\[
 \|\mb M \psi - h_k(\mb M \psi)\| \lesssim k^{-\bar p/d} \|\mb M \psi\|_{W^{\bar p}} \lesssim k^{-\bar p/d} \|\psi\|\,.
\]
Therefore, 
\[
 \|\mb M \psi - \Pi_k \mb M \psi\| = \|\mb M \psi - h_k(\mb M \psi)  + \Pi_k (h_k(\mb M \psi) - \mb M \psi)\| \leq 2 \|\mb M \psi - h_k(\mb M \psi)\| \lesssim k^{-\bar p/d} \|\psi\|
\]
and so $\|\mb M - \Pi_k \mb M\| = O(k^{-\bar p/d}) = o(1)$ as required. 

Similar arguments yield $\delta_k= O(k^{-p/d})$ and $\delta_k^* = O(k^{-p/d})$.

By Lemma \ref{lem:rho:1}, conditions (iv)--(vii) are sufficient for Assumption \ref{a:var} and we may take $\eta_{n,k} = \eta_{n,k}^* = k^{(r+2)/(2r)}/\sqrt n$. Choosing $k \asymp n^{\frac{rd}{2rp+(2+r)d}}$ balances bias and variance terms and we obtain the convergence rates as stated.
\end{proof}

\begin{proof}[Proof of Remark \ref{rmk:delta}]
First observe that $\mb M \phi = \sum_{n=1}^\infty \mu_n \langle \phi , \varphi_n \rangle g_n$. Taking the inner product of both sides of $\mb M \phi = \rho\phi$ with $g_n$, we obtain  $\mu_n \langle \phi,\varphi_n \rangle = \rho \langle \phi,g_n\rangle$ for each $n \in \mb N$. By Parseval's identity, $\| \phi\|^2 = \sum_{n \in \mb N} \langle \phi,\varphi_n \rangle^2 \geq \rho^2 \sum_{n \in \mb N: \mu_n > 0} \mu_n^{-2} \langle \phi,g_n \rangle^2$. Similarly, $\|\phi^*\|^2 \geq \rho^2 \sum_{n \in \mb N:\mu_n > 0} \mu_n^{-2} \langle \phi^*,\varphi_n \rangle^2$. Note that $\langle \phi,g_n \rangle = 0$ and $\langle \phi^*,\varphi_n \rangle = 0$ if $\mu_n = 0$.

As $B_k$ spans the linear subspace in $L^2$ generated by $\{g_n\}_{n=1}^k$, we have $\phi_k := \sum_{n=1}^k \langle \phi,g_n \rangle g_n \in B_k$. Therefore, assuming $\mu_{k+1} > 0$ (else the result is trivially true): 
\[
 \| \phi - \phi_k\|^2 = \sum_{n \geq k+1} \langle \phi,g_n \rangle^2 = \mu_{k+1}^2 \sum_{n \geq k+1} \frac{ \langle \phi,g_n \rangle^2 }{\mu_{k+1}^2} \leq \mu_{k+1}^2 \sum_{n \geq k+1:\mu_n > 0} \frac{ \langle \phi,g_n \rangle^2 }{\mu_n^2} \leq \mu_{k+1}^2 \frac{\|\phi\|^2}{\rho^2}\,.
\]
It follows that: 
\[
 \delta_k = \| \phi - \Pi_k \phi\| = \| \phi - \phi_k + \Pi_k (\phi_k - \phi)\| \leq 2 \| \phi - \phi_k\| = O(\mu_{k+1})\,.
\]
A similar argument gives $\delta_k^* = O(\mu_{k+1})$.
\end{proof}

Before proving Theorem \ref{t:asydist:1} we first present a lemma that controls higher-order bias terms involving $\phi_k^{\phantom *}$ and $\phi_k^*$. Define:
\[
 \psi_{k,\rho}(x_0,x_1) = \phi_k^*(x_0) m(x_0,x_1) \phi_k^{\phantom{*}}(x_1) - \rho_k \phi^*_k(x_0) \phi_k^{\phantom{*}}(x_0)
\]
with $\phi_k^{\phantom *}$ and $\phi_k^*$ normalized so that $\|\phi_k\| = 1$ and $\langle \phi_k^{\phantom *}, \phi_k^* \rangle = 1$, and:
\[
 \Delta_{\psi,n,k} = \frac{1}{n} \sum_{t=0}^{n-1} \big(  \psi_{\rho,k}(X_t,X_{t+1}) - \psi_\rho(X_t,X_{t+1}) \big)
\]
where $\psi_\rho$ is from display (\ref{e:inf:def}).

To simplify notation, let $\phi_t = \phi(X_t)$, $\phi^*_t = \phi^*(X_t)$, $\phi_{k,t} = \phi_k(X_t)$ and $\phi^*_{k,t} = \phi_k^*(X_t)$.

\begin{lemma} \label{lem:inf:1}
Assumption \ref{a:id} and \ref{a:bias} hold. Then: $\Delta_{\psi,n,k} = O_p( \delta_k^{\phantom *} + \delta_k^*)$.
\end{lemma}

\begin{proof}[Proof of Lemma \ref{lem:inf:1}]
First write 
\begin{align*}
 \Delta_{\psi,n,k}  & =  \frac{1}{n} \sum_{t=0}^{n-1} (\phi^*_{k,t} - \phi^*_t) m(X_t,X_{t+1}) \phi_{k,t+1} + \frac{1}{n} \sum_{t=0}^{n-1} \phi^*_t m(X_t,X_{t+1}) (\phi_{k,t+1} - \phi_{t+1}) \\
 & \quad - (\rho_k - \rho) \frac{1}{n} \sum_{t=0}^{n-1} \phi^*_{k,t} \phi_{k,t} - \rho \frac{1}{n} \sum_{t=0}^{n-1} (\phi^*_{k,t} \phi_{k,t} - \phi^*_t \phi_t) \quad  =: \quad \wh T_1 + \wh T_2 + \wh T_3 + \wh T_4\,.
\end{align*}
By iterated expectations:
\begin{align*}
 \mb E[ | (\phi^*_{k,t} - \phi^*_t) m(X_t,X_{t+1}) \phi_{k,t+1}|] 
 & = \langle |\phi^*_k - \phi^*|, \mb M (|\phi_k|) \rangle 
 \leq \|\phi_k^* - \phi^*\|  \| \mb M \| \| \phi_k\| = O(\delta_k^*)
\end{align*}
using Cauchy-Schwarz, boundedness of $\mb M$ (Assumption \ref{a:id}) and Lemma \ref{lem:bias} (note that the normalization $\langle \phi_k^*,\phi_k^{\phantom *} \rangle = 1$ and $\langle \phi,\phi^* \rangle = 1$ instead of $\|\phi_k^*\| = 1$ and $\| \phi^*\| = 1$ do not affect the conclusions of Lemma \ref{lem:bias}). Markov's inequality then implies $\wh T_1 = O_p(\delta_k^*)$. Similarly,
\begin{align*}
 \mb E[|\phi^*_t m(X_t,X_{t+1}) (\phi_{k,t+1} - \phi_{t+1})|] 
 = \langle \phi^* , \mb M ( | \phi_{k} - \phi|) \rangle
  \leq \| \phi^*\| \|\mb M\| \| \phi_k - \phi\| = O(\delta_k) 
\end{align*}
and so $\wh T_2 = O_p(\delta_k)$.

Since $\rho_k - \rho = O(\delta_k)$ by Lemma \ref{lem:bias}(a) and $\frac{1}{n} \sum_{t=0}^{n-1} \phi^*_{k,t} \phi_{k,t}  = O_p(1)$ follows from Lemma \ref{lem:bias}(b)(c), we obtain $\wh T_3 = O_p(\delta_k)$. Finally,
\begin{align*}
 \mb E[|\phi^*_{k,t} \phi_{k,t} - \phi^*_t \phi_t|] 
 & \leq \|\phi_k^* - \phi^*\| \| \phi_k\| + \| \phi^*\| \| \phi_k - \phi\| = O(\delta_k^{\phantom *} + \delta_k^*)
\end{align*}
again by Cauchy-Schwarz and Lemma \ref{lem:bias}. Therefore, $\wh T_4 = O_p(\delta_k^{\phantom *} + \delta_k^*)$.
\end{proof}

\begin{proof}[Proof of Theorem \ref{t:asydist:1}]
First note that:
\begin{align}
 \sqrt n ( \hat \rho - \rho) & = \sqrt n ( \hat \rho - \rho_k) + \sqrt n ( \rho_k - \rho) \notag \\
 & = \sqrt n ( \hat \rho - \rho_k) + o(1) \notag \\
 & = \sqrt n c_k^{*\prime} ( \wh{\mf M} - \rho_k \wh{\mf G}) c_k + o_p(1) \label{a:asydist:pf:1}
\end{align}
where the second line is by Assumption \ref{a:asydist}(a) and the third line is by Lemma \ref{lem:expansion} and Assumption \ref{a:asydist}(b) (under the normalizations $\|\mf G c_k\| = 1$ and $c_k^{* \prime } \mf G c_k^{\phantom *} = 1$). By identity, we may write the first term on the right-hand side of display (\ref{a:asydist:pf:1}) as:
\begin{align}
 \sqrt n c_k^{*\prime} ( \wh{\mf M} - \rho_k \wh{\mf G}) c_k 
 & = \frac{1}{\sqrt n} \sum_{t=0}^{n-1} \psi_{\rho}(X_t,X_{t+1}) + \sqrt n \times \Delta_{\psi,n,k} \notag \\
 & = \frac{1}{\sqrt n} \sum_{t=0}^{n-1} \psi_{\rho}(X_t,X_{t+1}) + o_p(1) \label{a:asydist:pf:2}
\end{align}
where the second line is by Lemma \ref{lem:inf:1} and Assumption \ref{a:asydist}(a).
The result follows by substituting (\ref{a:asydist:pf:2}) into (\ref{a:asydist:pf:1}) and applying a CLT for stationary and ergodic martingale differences (e.g. \cite{Billingsley1961}), which is valid in view of Assumption \ref{a:asydist}(c).
\end{proof}

\begin{proof}[Proof of Theorem \ref{t:eff}]
This is a consequence of Theorem \ref{t:eff:main} in Appendix \ref{ax:inf}.
\end{proof}

\begin{proof}[Proof of Theorem \ref{t:asydist:2a}]
Let $m_t(\alpha) = m(X_t,X_{t+1};\alpha)$.  By Assumption \ref{a:asydist}(a), Lemma \ref{lem:expansion} and Assumption \ref{a:asydist}(b):
\begin{align}
 \sqrt n ( \hat \rho - \rho)
 & = \frac{1}{\sqrt n } \sum_{t=0}^{n-1} \Big( \phi^*_{k,t}\phi_{k,t+1} m(X_t,X_{t+1},\hat \alpha) - \rho_k \phi^*_{k,t}\phi_{k,t} \Big) + o_p(1) \notag \\
 & = \frac{1}{\sqrt n} \sum_{t=0}^{n-1} \psi_\rho(X_t,X_{t+1})  + \frac{1}{\sqrt n} \sum_{t=0}^{n-1}  \phi_{k,t}^*\phi_{k,t+1}^{\phantom *}\big(   m_t(\hat \alpha) - m_t( \alpha_0) \big) + o_p(1) \label{e:2a:pf1} 
\end{align}
where the second equality is by Lemma \ref{lem:inf:1}.

We decompose the second term on the right-hand side of (\ref{e:2a:pf1}) as:
\begin{align}
 \frac{1}{\sqrt n} \sum_{t=0}^{n-1} \phi_{k,t}^*\phi_{k,t+1}^{\phantom *} \big( m_t(\hat \alpha) - m_t( \alpha_0) \big) \notag
 & = \frac{1}{\sqrt n} \sum_{t=0}^{n-1} \phi^*_t \phi_{t+1}^{\phantom *}\frac{\partial m_t(\alpha_0)}{\partial \alpha'} (\hat \alpha - \alpha_0) \\
 & \quad + \frac{1}{\sqrt n} \sum_{t=0}^{n-1} \phi_t^*\phi_{t+1}^{\phantom *}\Big( m_t(\hat \alpha) - m_t(\alpha_0) - \frac{\partial m_t(\alpha_0)}{\partial \alpha'} (\hat \alpha - \alpha_0) \Big) \notag \\
 & \quad + \frac{1}{\sqrt n} \sum_{t=0}^{n-1} (  \phi_{k,t}^*\phi_{k,t+1}^{\phantom *} -  \phi_t^*\phi_{t+1}^{\phantom *} ) ( m_t(\hat \alpha) - m_t(\alpha_0) ) \notag \\
 & =:  \frac{1}{\sqrt n} \sum_{t=0}^{n-1} \phi^*_t \phi_{t+1}^{\phantom *}\frac{\partial m_t(\alpha_0)}{\partial \alpha'} (\hat \alpha - \alpha_0) + \wh T_1 + \wh T_2 \label{e:2a:pf2} \,.
\end{align}
For term $\wh T_1$, whenever $\hat \alpha \in N$ (which it is wpa1) we may take a mean value expansion to obtain:
\[
 \wh T_1 = \frac{1}{n} \sum_{t=0}^{n-1} \phi_t^*\phi_{t+1}^{\phantom *}\Big( \frac{\partial m_t(\tilde  \alpha)}{\partial \alpha'}  - \frac{\partial m_t(\alpha_0)}{\partial \alpha'}\Big)  \times \sqrt n (\hat \alpha - \alpha_0) 
\]
where $\tilde \alpha$ is in the segment between $\hat \alpha$ and $\alpha_0$. It follows by routine arguments (e.g. Lemma 4.3 of \cite{NeweyMcFadden}, replacing the law of large numbers by the ergodic theorem) that:
\begin{equation} \label{e:2a:pf2:1}
 \frac{1}{ n} \sum_{t=0}^{n-1} \phi_t^*\phi_{t+1}^{\phantom *}\Big( \frac{\partial m_t(\tilde  \alpha)}{\partial \alpha}  - \frac{\partial m_t(\alpha_0)}{\partial \alpha}\Big)  = o_p(1)
\end{equation}
holds under Assumption \ref{a:parametric}(c)(d). Moreover, $\sqrt n (\hat \alpha - \alpha_0) = O_p(1)$ by Assumption \ref{a:parametric}(a)(b). Therefore, $\wh T_1 = o_p(1)$.

For term $\wh T_2$, observe that by Assumption \ref{a:parametric}(c), whenever $\hat \alpha \in N$ (which it is wpa1) we have:
\[
 |m_t(\hat \alpha ) - m_t(\alpha_0)| \leq \bar m(X_t,X_{t+1}) \times \| \hat \alpha - \alpha_0\|
\]
where $\max_{0 \leq t \leq n-1} |\bar m(X_t,X_{t+1}) | = o_p(n^{1/s})$ because $E[\bar m(X_t,X_{t+1})^s] < \infty$. Therefore, wpa1 we have:
\begin{align*}
 \wh T_2 & \leq \sqrt n \times  \frac{1}{n} \sum_{t=0}^{n-1} |  \phi_{k,t}^*\phi_{k,t+1}^{\phantom *} -  \phi_t^*\phi_{t+1}^{\phantom *} | \times \max_{0 \leq t \leq n-1} |\bar m(X_t,X_{t+1}) | \times \| \hat \alpha - \alpha_0 \| \\
 & =  \frac{1}{n} \sum_{t=0}^{n-1} |  \phi_{k,t}^*\phi_{k,t+1}^{\phantom *} -  \phi_t^*\phi_{t+1}^{\phantom *} |  \times o_p(n^{1/s})  \quad = \quad O_p( \delta_k^{\phantom *} + \delta_k^*) \times  o_p(n^{1/s}) 
\end{align*}
by similar arguments to the proof of Lemma \ref{lem:inf:1}. Finally, observe that $n^{1/s} ( \delta_k^{\phantom *} + \delta_k^* ) = o(1)$ by Assumption \ref{a:asydist}(a) and the condition $s \geq 2$. Therefore, $\wh T_2 = o_p(1)$.

Since $\wh T_1$ and $\wh T_2$ in display (\ref{e:2a:pf2}) are both $o_p(1)$, we have:
\begin{align*}
  \frac{1}{\sqrt n} \sum_{t=0}^{n-1} \phi_{k,t}^*\phi_{k,t+1}^{\phantom *} \big( m_t(\hat \alpha) - m_t( \alpha_0) \big) 
  & = \bigg( \frac{1}{ n} \sum_{t=0}^{n-1} \phi^*_t \phi_{t+1}^{\phantom *}\frac{\partial m_t(\alpha_0)}{\partial \alpha'} \bigg) \sqrt n (\hat \alpha - \alpha_0) + o_p(1) \\
  & = \mb E \left[  \phi^*(X_t) \phi(X_{t+1})\frac{\partial m(X_t,X_{t+1};\alpha_0)}{\partial \alpha'} \right] \sqrt n (\hat \alpha - \alpha_0) + o_p(1) \,.
\end{align*}
Substituting into (\ref{e:2a:pf1}) and using Assumption \ref{a:parametric}(a):
\[
 \sqrt n (\hat \rho - \rho) = \frac{1}{\sqrt n } \sum_{t=0}^{n-1} h_{[\mr{2a}]}'  \bigg( \begin{array}{c} \psi_{\rho,t} \\ \psi_{\alpha,t} \end{array} \bigg)  + o_p(1)
\]
and the result follows by Assumption \ref{a:parametric}(b).
\end{proof}

\begin{proof}[Proof of Theorem \ref{t:asydist:2b}]
We follow similar arguments to the proof of Theorem \ref{t:asydist:2a}. Here, we can decompose the second term on the right-hand side of display (\ref{e:2a:pf1}) as:
\begin{align*}
 \frac{1}{\sqrt n} \sum_{t=0}^{n-1} \phi_{k,t}^*\phi_{k,t+1}^{\phantom *} \big( m_t(\hat \alpha) - m_t( \alpha_0) \big) & = \sqrt n (\ell(\hat \alpha) - \ell(\alpha_0) )+ \wh T_1 + \wh T_2 \\
 & = \frac{1}{\sqrt n } \sum_{t=0}^{n-1} \psi_{\ell,t} + o_p(1) + \wh T_1 + \wh T_2
\end{align*}
where the second line is by Assumption \ref{a:nonpara}(b)(c), with:
\begin{align*}
 \wh T_1 & = \frac{1}{\sqrt n} \sum_{t=0}^{n-1} \Big( \phi_t^* \phi_{t+1}^{\phantom *} (m_t(\hat \alpha) - m_t(\alpha_0))  - (\ell(\hat \alpha) - \ell(\alpha_0)) \Big) \\
 \wh T_2 & = \frac{1}{\sqrt n} \sum_{t=0}^{n-1} (\phi_{k,t}^*\phi_{k,t+1}^{\phantom *} - \phi_t^* \phi_{t+1}^{\phantom *}) (m_t(\hat \alpha) - m_t(\alpha_0)) \,.
\end{align*}
 The result will follow by Assumption \ref{a:nonpara}(c)(d) provided $\wh T_1$ and $\wh T_2$ are both $o_p(1)$.

For term $\wh T_1$, notice that $\wh T_1 = \mc Z_n(g_{\hat \alpha})$ where $\mc Z_n$ denotes the centered empirical process on $\mc G$. We have $\mb K(g_{\hat \alpha},g_{\hat \alpha}) = o_p(1)$ by Assumption \ref{a:nonpara}(c). Appropriately modifying the arguments of Lemma 19.24 in \cite{vdV} (i.e. replacing the $L^2$ norm by the norm induced by $\mb K$, which is the appropriate semimetric for the weakly dependent case) gives $\mc Z_n(g_{\hat \alpha}) \to_p 0$.

For term $\wh T_2$, observe that:
\[
 \mb E[ | (\phi_{k,t}^*\phi_{k,t+1}^{\phantom *} - \phi_t^* \phi_{t+1}^{\phantom *}) (m_t(\hat \alpha) - m_t(\alpha_0)) |] \lesssim \mb E \big[ |(\phi_{k,t}^*\phi_{k,t+1}^{\phantom *} - \phi_t^* \phi_{t+1}^{\phantom *}) |^{s/(s-1)} \big]^{(s-1)/s}
\]
by Assumption \ref{a:nonpara}(e) and H\"older's inequality. We complete the proof assuming  $\|\phi_k\|_{2s/(s-2)} = O(1)$ and $\|\phi^*\|_{2s/(s-2)} < \infty$; the proof under the alternative condition in Assumption \ref{a:nonpara}(e) is analogous. By the Minkowski and H\"older's inequalities and Assumption \ref{a:nonpara}(e) we have: 
\begin{align*}
 & \mb E \big[ |(\phi_{k,t}^*\phi_{k,t+1}^{\phantom *} - \phi_t^* \phi_{t+1}^{\phantom *}) |^{s/(s-1)} \big]^{(s-1)/s} \\
 & \leq \mb E \big[ |(\phi_{k,t}^* - \phi_t^* ) \phi_{k,t+1}^{\phantom *} |^{s/(s-1)} \big]^{(s-1)/s} + \mb E \big[ |\phi_t^* (\phi_{k,t+1}^{\phantom *} - \phi_{t+1}^{\phantom *}) |^{s/(s-1)} \big]^{(s-1)/s} \\
 & \leq \| \phi_k^* - \phi^*\| \|\phi_k\|_{2s/(s-2)} + \| \phi_k - \phi\| \| \phi^*\|_{2s/(s-2)} \\
 & = O(1) \times O(\delta_k^* + \delta_k^{\phantom *}) \,.
\end{align*}
It follows by Assumption \ref{a:asydist}(a) and Markov's inequality that $\wh T_2 = o_p(1)$.
\end{proof}

The following Lemma is based on Lemma 6.10 in \cite{AGN}.

\begin{lemma}\label{lem:fpiter}
Let the conditions of Proposition \ref{p:nl} hold. Then: there exists finite positive constants $C,c$ and a neighborhood $N$\! of $h$ such that:
\[
 \| \mb T^n \psi - h\| \leq C e^{-cn} 
\]
for all $\psi \in N$.
\end{lemma}

\begin{proof}[Proof of Lemma \ref{lem:fpiter}]
Fix some constant $\bar a$ such that $r(\mb D_h) < \bar a < 1$. By the Gelfand formula, there exists $m \in \mb N$ such that $\|\mb D_h^m\| < \bar a^m$. Fr\'echet differentiability of $\mb T$ at $h$ together with the chain rule for Fr\'echet derivatives implies that:
\[
 \|\mb T^m \psi - \mb T^m h - \mb D_h^m (\psi-h)\| = o(\|\psi - h\|) \quad \mbox{as $\|\psi - h\| \to 0$}
\]
hence:
\[
 \| \mb T^m \psi - h\| \leq \| \mb D_h^m \| \| \psi - h\| + o(\| \psi - h\|) < (\bar a^m+o(1)) \times \| \psi - h\|\,.
\]
We may choose $\epsilon > 0$ and $a \in (\bar a,1)$ such that $\| \mb T^m \psi - h\| \leq a^m \| \psi - h\|$ for all $\psi \in B_\epsilon(h) : = \{ \psi \in L^2 : \| \psi - h\| < \epsilon\}$. ($B_\epsilon(h)$ is the neighborhood in the statement of the lemma.) Then for any $\psi \in B_\epsilon(h)$ and any $k \in \mb N$ we have: 
\begin{equation} \label{lem:fpiter:1}
 \| \mb T^{km} \psi - h\| \leq a^{km} \| \psi - h\|\,.
\end{equation}
It is straightforward to show via induction that boundedness of $\mb G$ and homogeneity of degree $\beta$ of $\mb T$ together imply: 
\begin{equation} \label{lem:fpiter:2}
 \| \mb T^n \psi_1 - \mb T^n \psi_2 \| \leq ( 1 + \| \mb G\|)^{\frac{1}{1-\beta}} \| \psi_1 - \psi_2\|^{\beta^n}
\end{equation}
for any $\psi_1,\psi_2 \in L^2$. 

Take any $n \geq m$ and let $k = \lfloor n/m \rfloor $. By (\ref{lem:fpiter:1}) and (\ref{lem:fpiter:2}) we have:
\begin{align*}
 \| \mb T^n \psi - h\| & = \| \mb T^{(n-km)} \mb T^{km} \psi - \mb T^{(n-km)} h\| \\
 & \leq (1+\|\mb G\|)^{\frac{1}{1-\beta}} \|\mb T^{km} \psi - h\|^{\beta^{(n-km)}} \\
 & \leq (1+\|\mb G\|)^{\frac{1}{1-\beta}}\epsilon^{\beta^{(n-km)}}(a^{ km})  ^{\beta^{(n-km)}}
\end{align*}
for any $\psi \in B_\epsilon(h)$. The result follows for suitable choice of $C$ and $c$.
\end{proof}

\begin{proof}[Proof of Proposition \ref{p:nl}]
Take $C$ and $c$ from Lemma \ref{lem:fpiter} and $B_\epsilon(h)$ from the proof of Lemma \ref{lem:fpiter}. Let $N = \{ \psi \in L^2 : \| \psi - \chi\| < \epsilon/\|h\|\}$ and note that $\{ \|h\| \psi  : \psi \in N\} = B_{\epsilon}(h)$. 

Take any $\psi \in \{ a f : f \in N, a \in \mb R \setminus \{0\}\}$. For any such $\psi$ we can write $\psi = (a /\|h\|) f^*$ where $f^* = \|h\|f  \in B_{\epsilon }(h)$. By homogeneity of $\mb T$:
\[
 \chi_{n+1}(\psi) = \frac{\mb T^n (\chi_1( \psi))}{\|\mb T^n (\chi_1( \psi))\|}  = \frac{\mb T^n (\chi_1(f^*))}{\|\mb T^n (\chi_1(f^*))\|}  = \chi_{n+1}(f^*)
\]
for each $n \geq 1$ (note positivity of $\mb G$ ensures that $\| \mb T^n f^* \| > 0$ for each $n$ and each $f^* \in N$). It follows from Lemma \ref{lem:fpiter} that:
\[
 \| \chi_{n+1}(\psi) - \chi\| = \| \chi_{n+1}(f^*) - \chi\| =  \left\| \frac{\mb T^n (f^*)}{\| \mb T^n (f^*)\|} - \frac{h}{\|h\|} \right\| \leq \frac{2}{\|h\|} \| \mb T^n (f^*) - h\| \leq \frac{2}{\|h\|} C e^{-cn}
\]
as required.
\end{proof}

\begin{proof}[Proof of Corollary \ref{c:local-id}]
The result for $\chi$ is stated in the text. For $h$, let $C$, $c$, and $B_\epsilon(h)$ be as in Lemma \ref{lem:fpiter} and its proof. 
Suppose $h'$ is a fixed point of $\mb T$ belonging to $B_\epsilon(h)$. Then by Lemma \ref{lem:fpiter}:
\[
 \| h' - h\| = \| \mb T^n h' - h\| \leq C e^{-cn} \to 0
\]
hence $h' = h$. 
\end{proof}

\begin{proof}[Proof of Theorem \ref{t:fpest}]
Immediate from Lemmas \ref{lem:bias:fp} and \ref{lem:var:fp}.
\end{proof}

\subsection{Proofs for Appendix \ref{ax:est:cgce}}

\begin{proof}[Proof of Lemma \ref{lem:exist}]
We first prove that there exists $K \in \mb N$ such that the maximum eigenvalue $\rho_k$ of the operator $\Pi_k \mb M : L^2 \to L^2$ is real and simple whenever $k \geq K$.

Under Assumption \ref{a:id}, $\rho$ is a simple isolated eigenvalue of $\mb M$. Therefore, there exists an $\epsilon > 0$ such that $|\lambda - \rho| > 2\epsilon$ for all $\lambda \in \sigma(\mb M)$. 
Let $\Gamma$ denote a positively oriented circle in $\mb C$ centered at $\rho$ with radius $\epsilon$. Let $\mc R(\mb M,z) = (\mb M - zI)^{-1}$ denote the resolvent of $\mb M$ evaluated at $z \in \mb C \setminus \sigma(\mb M)$, where $I$ is the identity operator. Note that:
\begin{equation} \label{e:crdef}
 C_{\mc R} : = \sup_{z \in \Gamma}\|\mc R(\mb M,z)\| < \infty
\end{equation}
because $\mc R(\mb M,z)$ is a holomorphic function on $\Gamma$ and $\Gamma$ is compact. 

By Assumption \ref{a:bias}, there exists $K \in \mb N$ such that:
\begin{equation} \label{e:cineq}
 C_{\mc R} \times \|\Pi_k \mb M - \mb M\|  < 1
\end{equation}
holds for all $k \geq K$. It follows by Theorem IV.3.18 on p. 214 of \cite{Kato} that whenever $k \geq K$: (i) the operator $\Pi_k \mb M$ has precisely one eigenvalue $\rho_k$ inside $\Gamma$ and $\rho_k$ is simple; (ii) $\Gamma \subset (\mb C \setminus \sigma(\Pi_k \mb M))$; and (iii) $\sigma(\Pi_k \mb M)\setminus \{\rho\}$ lies on the exterior of $\Gamma$. Note that $\rho_k$ must be real whenever $k \geq K$ because complex eigenvalues come in conjugate pairs. Thus, if $\rho_k$ were complex-valued then its conjugate would also be inside $\Gamma$, which would contradict the fact that $\rho_k$ is the unique eigenvalue of $\Pi_k \mb M$ on the interior of $\Gamma$.

Any nonzero eigenvalue of $\Pi_k \mb M$ is also en eigenvalue of $(\mf M,\mf G)$ with the same multiplicity. Therefore, the largest eigenvalue $\rho_k$ of  $(\mf M,\mf G)$ is positive and simple whenever $k \geq K$. 
\end{proof}

Let $\Pi_k \mb M|_{B_k}: B_k \to B_k$ denote the restriction of $\Pi_k \mb M$ to $B_k$.  Recall that $\phi^*_k(x) = b^k(x)'c_k^*$ where $c_k^*$ solves the left-eigenvector problem in (\ref{e:gev}). Here, $\phi_k^*$ is the eigenfunction of the adjoint $(\Pi_k \mb M|_{B_k})^* : B_k \to B_k$ corresponding to $\rho_k$. That is, $\langle (\Pi_k \mb M|_{B_k})^* \phi_k^* , \psi \rangle = \rho_k \langle \phi_k^*, \psi \rangle$ for all $\psi \in B_k$.

Another adjoint is also relevant for the next proof, namely $(\Pi_k \mb M)^* : L^2 \to L^2$ which is the adjoint of $\Pi_k \mb M$ in the space $L^2$. It follows from Lemma \ref{lem:exist} that $(\Pi_k \mb M)^*$ has an eigenfunction, say $\phi_k^+$ corresponding to $\rho_k$ whenever $k \geq K$. That is, $\langle (\Pi_k \mb M)^* \phi_k^+ , \psi \rangle = \rho_k \langle \phi_k^+, \psi \rangle$ for all $\psi \in L^2$. Notice that $\phi_k^+$ does not necessarily belong to $B_k$, so we may have that $\phi_k^* \neq \phi_k^+$.

\begin{proof}[Proof of Lemma \ref{lem:bias}]
Step 1: Proof of part (b). By Proposition 4.2 of \cite{Gobetetal} (taking $T = \mb M$, $T_\varepsilon = \Pi_k \mb M$  and $\Gamma=$ the boundary of $B(\kappa,\rho)$ in their notation), the inequality:
\[
 \|\phi - \phi_k\| \leq \mr {const} \times \|(\Pi_k \mb M - \mb M)\phi\|
\]
holds for all $k$ sufficiently large, where the constant depends only on $C_{\mc R}$. The result follows by noticing that 
\begin{equation} \label{e:pikbd}
 \|(\Pi_k \mb M - \mb M)\phi\| = \rho \times \| \Pi_k \phi - \phi\| = O(\delta_k) \,.
\end{equation}

Step 2: Proof of part (a). By Corollary 4.3 of \cite{Gobetetal}, the inequality:
\[
 |\rho - \rho_k| \leq \mr {const} \times \|(\Pi_k \mb M - \mb M)\phi\|
\]
holds for all $k$ sufficiently large, where the constant depends only on $C_{\mc R}$ and $\|\mb M\|$. The result follows by (\ref{e:pikbd}).

Step 3: Proof that $\|\phi_k^+ - \phi^*\| = O(\delta_k^*)$ under the normalizations $\| \phi^*_{\phantom k} \| = 1$ and $\| \phi^+_k \| = 1$. 
Let $P_k^*$ denote the spectral projection on the eigenspace of $(\Pi_k \mb M)^*$ corresponding to $\rho_k$. By the proof of Proposition 4.2 of \cite{Gobetetal} (taking $T = \mb M^*$, $T_\varepsilon = (\Pi_k \mb M)^*$ and $\Gamma=$ the boundary of $B(\kappa,\rho)$ in their notation and noting that $\|\mc R(\mb M^*,z)\| = \|\mc R(\mb M,\bar z)\|$ holds for all $z \in \Gamma$), the inequality:
\[
 \|\phi^* - P_k^* \phi^* \| \leq \mr{const} \times  \|((\Pi_k \mb M)^* - \mb M^*) \phi^*\|
\]
for all $k$ sufficiently large, where the constant depends only on $C_{\mc R}$. Moreover,
\[
 \|((\Pi_k \mb M)^* - \mb M^*) \phi^*\| = \| (\mb M^* \Pi_k  - \mb M^*) \phi^*\| = \| \mb M^* ( \Pi_k \phi^* - \phi^*)\| \leq \|\mb M\| \|\Pi_k \phi^* - \phi^*\| = O(\delta_k^*)
\]
by definition of $\delta_k^*$ (cf. display (\ref{e:deltas})) and boundedness of $\mb M$. Therefore, 
\begin{equation} \label{e:projbdstar}
 \|\phi^* - P_k^* \phi^* \| = O(\delta_k^*).
\end{equation}
Define $(\phi_k^+ \otimes \phi_k^{\phantom +}\!) \psi(x) = \langle \phi_k^{\phantom +},\psi \rangle \times \phi_k^+(x)$ for any $\psi \in L^2$. We use the fact that:
\[
 P_k^* = \frac{1}{\langle \phi_k^{\phantom +},\phi_k^+ \rangle}(\phi_k^+ \otimes \phi_k^{\phantom +}\!) 
\]
under the normalizations $\|\phi_k\| = 1$ and $\|\phi_k^+\| = 1$ \cite[p. 113]{Chatelin}. Then, under the sign normalization $\langle \phi^*, \phi_k^+ \rangle \geq 0$, we have:
\[
 \left\| \phi^*- \phi_k^+ \right\|^2 
 \leq 2 \| \phi^* - ( \phi^+_k \otimes \phi^+_k) \phi^* \|^2
\]
(see the proof of Proposition 4.2 of \cite{Gobetetal}). Moreover,
\[
 \| \phi^* - ( \phi^+_k \otimes \phi^+_k) \phi^* \|^2 \leq \left\| \phi^* - \bigg( \phi_k^+ \otimes \frac{\phi_k}{\langle \phi_k^{\phantom +},\phi_k^+ \rangle} \bigg) \phi^* \right\|^2 
 \equiv \|\phi^* - P_k^* \phi^* \|^2 
\]
 It follows by (\ref{e:projbdstar}) that $\|\phi^* - \phi_k^+\| = O(\delta_k^*)$.

\medskip

Step 4: Proof that $\|\phi_k^* - \phi^*\| = O(\delta_k^*)$. To relate $\phi_k^+$ to $\phi_k^*$, observe that by definition of $(\Pi_k \mb M)^*$ and $(\Pi_k \mb M|_{B_k})^*$ we must have:
\begin{align*}
 & & \mb E[\phi_k^+(X) \Pi_k \mb M \psi(X)] & =  \rho_k \mb E[\phi_k^+(X) \psi(X)]  & & \mbox{for all $\psi \in L^2$} & &  \\
 & & \mb E[\phi_k^*(X) \Pi_k \mb M \psi_k(X)] & =  \rho_k \mb E[\phi_k^*(X) \psi_k(X)] & & \mbox{for all $\psi_k \in B_k$.} & & 
\end{align*}
It follows from taking $\psi = \psi_k$ in the first line of the above display that $\Pi_k^{\phantom *} \phi^+_k = \phi^*_k$. Now by the triangle inequality and the fact that $\Pi_k$ is a weak contraction, we have:
\begin{align*}
 \|\phi^* - \phi_k^*\|  = \|\phi^*  - \Pi_k \phi^+_k\| 
 & \leq \|\phi^* - \Pi_k \phi^* \| + \|\Pi_k \phi^* - \Pi_k \phi^+_k\| \\
 & \leq \|\phi^* - \Pi_k \phi^* \| + \| \phi^* - \phi^+_k\| 
 = O(\delta_k^*) + O(\delta_k^*) 
\end{align*}
where the final equality is by  definition of $\delta_k^*$ (see display (\ref{e:deltas})) and Step 3.
\end{proof}

 The following lemma collects some useful bounds on the orthogonalized estimators.

\begin{lemma}\label{lem:matcgce}
(a) If $\wh{\mf G}$ is invertible then:
\[
 (\wh{\mf G}^o)^{-1} \wh{\mf M}^o - \mf M^o = \wh{\mf M}^o - \wh{\mf G}^o \mf M^o +  (\wh{\mf G}^o)^{-1} \Big((\wh{\mf G}^o - \mf I)^2{\mf M}^o  +  (\mf I - \wh{\mf G}^o)(\wh{\mf M}^o - \mf M^o ) \Big)\,.
\]
(b) In particular, if $\|\wh {\mf G}^o - \mf I\| \leq \frac{1}{2}$ we obtain:
\[
 \|(\wh{\mf G}^o)^{-1} \wh{\mf M}^o - \mf M^o \| \leq \|\wh{\mf M}^o - \mf M^o\| + 2 \|\wh{\mf G}^o - \mf I\|\times (\|\mf M^o\|  + \|\wh{\mf M}^o - \mf M^o\| )\,.
\]
\end{lemma}

\begin{proof}[Proof of Lemma \ref{lem:matcgce}]
If $\wh{\mf G}$ is invertible we have:
\begin{align*}
 (\wh{\mf G}^o)^{-1} \wh{\mf M}^o - \mf M^o  & = ( \mf I - (\wh{\mf G}^o)^{-1}(\wh{\mf G}^o - \mf I)) \wh{\mf M}^o - \mf M^o \notag \\
 & = \wh{\mf M}^o - \mf M^o - (\wh{\mf G}^o)^{-1}(\wh{\mf G}^o - \mf I) \mf M^o - (\wh{\mf G}^o)^{-1}(\wh{\mf G}^o - \mf I) (\wh{\mf M}^o - \mf M^o)\,. 
\end{align*}
Part (b) follows by the triangle inequality, noting that $\|(\wh{\mf G}^o)^{-1}\| \leq 2$ whenever $\|\wh{\mf G}^o - \mf I\| \leq \frac{1}{2}$. Substituting $(\wh{\mf G}^o)^{-1} = ( \mf I - (\wh{\mf G}^o)^{-1}(\wh{\mf G}^o - \mf I)) $ into the preceding display yields:
\begin{align*}
 (\wh{\mf G}^o)^{-1} \wh{\mf M}^o - \mf M^o 
 & = \wh{\mf M}^o - \mf M^o - ( \mf I - (\wh{\mf G}^o)^{-1}(\wh{\mf G}^o - \mf I))(\wh{\mf G}^o - \mf I) \mf M^o - (\wh{\mf G}^o)^{-1}(\wh{\mf G}^o - \mf I) (\wh{\mf M}^o - \mf M^o) \\
 & = \wh{\mf M}^o - \wh{\mf G}^o \mf M^o + (\wh{\mf G}^o)^{-1}(\wh{\mf G}^o - \mf I)^2 \mf M^o - (\wh{\mf G}^o)^{-1}(\wh{\mf G}^o - \mf I) (\wh{\mf M}^o - \mf M^o) \,. 
\end{align*}
as required.
\end{proof}

\begin{proof}[Proof of Lemma \ref{lem:exist:hat}]
Step 1: We show that:
\[
 \|\mc R(\Pi_k \mb M|_{B_k},z)\| \leq \|\mc R(\Pi_k \mb M,z)\|
\]
holds for all $z \in \mb C \setminus (\sigma(\Pi_k \mb M) \cup \sigma(\Pi_k \mb M|_{B_k}))$.
Fix any such $z$. For any $\psi_k \in B_k$ we have $\mc R(\Pi_k \mb M|_{B_k},z) \psi_k = \zeta_k$ where $\zeta_k = \zeta_k(\psi_k) \in B_k$ is given by $\psi_k = (\Pi_k \mb M - z I)\zeta_k$. For any $\psi \in L^2$ we have $\mc R(\Pi_k \mb M,z) \psi = \zeta$ where $\zeta = \zeta(\psi) \in L^2$ is given by $ \psi = (\Pi_k \mb M - z I) \zeta$. In particular, taking $\psi_k \in B_k$ we must have $\zeta_k(\psi_k) = \zeta(\psi_k)$. Therefore, $\mc R(\Pi_k \mb M|_{B_k},z) \psi_k = \mc R(\Pi_k \mb M,z) \psi_k$ holds for all $\psi_k \in B_k$. We now have:
\begin{align*}
 \|\mc R(\Pi_k \mb M|_{B_k},z)\| & =\sup\{  \|\mc R(\Pi_k \mb M|_{B_k},z) \psi_k\| :  \psi_k \in B_k , \|\psi_k\| = 1\}\\
 & =\sup\{  \|\mc R(\Pi_k \mb M,z) \psi_k\| :  \psi_k \in B_k , \|\psi_k\| = 1\}\\
 & \leq\sup\{  \|\mc R(\Pi_k \mb M,z) \psi\| :  \psi \in L^2 , \|\psi\| = 1\} =  \|\mc R(\Pi_k \mb M,z) \| \,.
\end{align*}

Step 2: We show that $(\wh{\mf M},\wh{\mf G})$ has a unique eigenvalue $\hat \rho$ inside $\Gamma$ wpa1, where $\Gamma$ is from the proof of Lemma \ref{lem:exist}. 

As the nonzero eigenvalues of $\Pi_k \mb M$, $\Pi_k \mb M|_{B_k}$, and $\mf G^{-1} \mf M$ are the same, it follows from the proof of Lemma \ref{lem:exist} that for all $k \geq K$ the curve $\Gamma$ encloses precisely one eigenvalue of $\mf G^{-1} \mf M$, namely $\rho_k$, and that $\rho_k$ is a simple eigenvalue of $\mf G^{-1} \mf M$.

Recall that $\mf G^{-1} \mf M^{\phantom 1}$ is isomorphic to $\Pi_k \mb M |_{B_k}$ on $(\mb R^k,\langle \cdot,\cdot\rangle_{\mf G})$. Let $\mc R(\mf G^{-1} \mf M,z)$ denote the resolvent of $\mf G^{-1} \mf M$ on\ $(\mb R^k,\langle \cdot,\cdot\rangle_{\mf G})$. By step 1, we then have:
\begin{equation} \label{e:resbd:k:1}
 \sup_{z \in \Gamma} \|\mc R(\mf G^{-1} \mf M,z)\|_{\mf G} 
 = \sup_{z \in \Gamma} \|\mc R(\Pi_k \mb M|_{B_k},z)\| 
 \leq \sup_{z \in \Gamma} \|\mc R(\Pi_k \mb M,z)\| \,.
\end{equation}
The second resolvent identity gives $\mc R(\Pi_k \mb M,z) = \mc R( \mb M,z) + \mc R(\Pi_k \mb M,z) (\mb M - \Pi_k \mb M) \mc R( \mb M,z)$. It follows that whenever (\ref{e:cineq}) holds (which it does for all $k \geq K$):
\begin{equation} \label{e:resbd:k:2}
 \sup_{z \in \Gamma} \|\mc R(\Pi_k \mb M,z)\| \leq \frac{C_{\mc R}}{1- C_{\mc R} \| \Pi_k \mb M - \mb M\|} = C_{\mc R} (1+o(1))
\end{equation}
by Assumption \ref{a:bias}. Combining (\ref{e:resbd:k:1}) and (\ref{e:resbd:k:2}), we obtain:
\begin{equation} \label{e:resbd:k} 
 \sup_{z \in \Gamma} \|\mc R(\mf G^{-1} \mf M,z)\|_{\mf G} = O(1)\,.
\end{equation}
By Lemma \ref{lem:matcgce}(b), Assumption \ref{a:var}, and boundedness of $\mb M$:
\[
 \| \wh {\mf G}^{-1} \wh{\mf M} - \mf G^{-1} \mf M \|_{\mf G} = \| (\wh {\mf G}^o)^{-1} \wh{\mf M}^o -   \mf M^o \|  = o_p(1) \,.
\]
It follows by (\ref{e:resbd:k}) that the inequality:
\begin{equation}\label{e:kato}
 \|\wh{\mf G}^{-1}\wh{\mf M} - {\mf G}^{-1}{\mf M}\|_{\mf G} \times \sup_{z \in \Gamma} \|\mc R(\mf G^{-1} \mf M,z)\|_{\mf G} < 1 
\end{equation}
holds wpa1. 

By Theorem IV.3.18 on p. 214 of \cite{Kato}, whenever (\ref{e:kato}) holds: $\wh{\mf G}^{-1}\wh{\mf M}$ has precisely one eigenvalue, say $\hat \rho$, inside $\Gamma$; $\hat \rho$ is simple, and; the remaining eigenvalues of $\wh{\mf G}^{-1}\wh{\mf M}$ are on the exterior of $\Gamma$. Note that $\hat \rho$ must necessarily be real whenever (\ref{e:kato}) holds (because complex eigenvalues come in conjugate pairs) hence the corresponding left- and right-eigenvectors $\hat c^*$ and $\hat c$ are also real and unique (up to scale).
\end{proof}

\begin{proof}[Proof of Lemma \ref{lem:var}]
Take $k \geq K$ from Lemma \ref{lem:exist} and work on the sequence of events upon which
\begin{equation}\label{e:kato:half}
 \|\wh{\mf G}^{-1}\wh{\mf M} - {\mf G}^{-1}{\mf M}\|_{\mf G} \times \sup_{z \in \Gamma} \|\mc R(\mf G^{-1} \mf M,z)\|_{\mf G} < \frac{1}{2} 
\end{equation}
holds. By the proof of Lemma \ref{lem:exist:hat}, this inequality holds wpa1 and $\hat \rho$, $\hat c$ and $\hat c^*$ to (\ref{e:est}) are unique on this sequence of events.

Step 1: Proof of part (b). Under the normalizations $\|\hat c\|_{\mf G} = 1$ and $ \|\hat c^*\|_{\mf G} = 1$, whenever (\ref{e:kato:half}) holds (which it does wpa1), we have
\[
 \|\hat \phi - \phi_k\|^2 = \|\hat c - c_k\|^2_{\mf G} \leq \sqrt 8  \sup_{z \in \Gamma} \|\mc R(\mf G^{-1} \mf M,z)\|_{\mf G} \times \| (\wh{\mf G}^{-1}\wh{\mf M} - {\mf G}^{-1}{\mf M})c_k \|_{\mf G}
\]
by Proposition 4.2 of \cite{Gobetetal} (setting $\wh{\mf G}^{-1}\wh{\mf M} = T_\varepsilon$, ${\mf G}^{-1}{\mf M} = T$ and $\Gamma = $ the boundary of $B(\kappa,\rho)$ in their notation). The result now follows by (\ref{e:resbd:k}) and the fact that 
\begin{equation} \label{e:mhat:cgce}
 \| (\wh{\mf G}^{-1}\wh{\mf M} - {\mf G}^{-1}{\mf M})c_k \|_{\mf G} = \|((\wh{\mf G}^o)^{-1}\wh{\mf M}^o - {\mf M}^o)\tilde c_k\| = O_p (\eta_{n,k})
\end{equation}
(cf. display (\ref{e:etas})).

Step 2: Proof of part (a). In view of (\ref{e:kato:half}), (\ref{e:resbd:k}) and the fact that $\|{\mf G}^{-1} {\mf M} \|_{\mf G} = \|\Pi_k \mb M|_{B_k}\| \leq \|\mb M\| < \infty$, by Corollary 4.3 of \cite{Gobetetal}, we have:
\[
  |\hat \rho - \rho_k| \leq O(1) \times \| (\wh{\mf G}^{-1}\wh{\mf M} - {\mf G}^{-1}{\mf M})c_k \|_{\mf G} \,.
\] 
The result follows by (\ref{e:mhat:cgce}).

Step 3: Proof of part (c). Identical arguments to the proof of part (b) yield: 
\[
  \|\hat \phi^* - \phi_k^* \| = \|\hat c^* - c_k^*\|_{\mf G} \leq \sqrt 8  \sup_{z \in \Gamma} \|\mc R(\mf G^{-1} \mf M',z)\|_{\mf G} \times \|(\wh{\mf G}^{-1}\wh{\mf M}' - {\mf G}^{-1}{\mf M}')c_k^*\|_{\mf G}
\]
under the normalization $\|\hat c^*\|_{\mf G} = \| c_k^*\|_{\mf G} = 1$. The result now follows by (\ref{e:resbd:k}), noting that $ \sup_{z \in \Gamma} \|\mc R(\mf G^{-1} \mf M',z)\|_{\mf G} =  \sup_{z \in \Gamma} \|\mc R(\mf G^{-1} \mf M,z)\|_{\mf G}$, and the fact that: 
\[
 \| (\wh{\mf G}^{-1}\wh{\mf M}' - {\mf G}^{-1}{\mf M}')c_k^* \|_{\mf G} = \|((\wh{\mf G}^o)^{-1}\wh{\mf M}^{o\prime} - {\mf M}^{o\prime})\tilde c_k^*\| = O_p (\eta_{n,k}^*)
\]
(cf. display (\ref{e:etas})).
\end{proof}

\subsection{Proofs for Appendix \ref{ax:est:fp}}

Some of the proofs in this subsection make use of properties of fixed point indices. We refer the reader to Section 19.5 of \cite{Kras} for details.

\begin{proof}[Proof of Lemma \ref{lem:fp:exist}]
By Assumption \ref{a:fp:exist} and Corollary \ref{c:local-id}, we may choose $\varepsilon > 0$ such that $\ol N = \{ \psi \in L^2 : \| \psi - h\| \leq \varepsilon\}$ contains only one fixed point of $\mb T$, namely $h$. We verify the conditions of Theorem 19.4 in \cite{Kras} where, in our notation, $\Omega = \ol N$, $E_n = B_k$, $P_n = \Pi_k$, $T = \mb T$, and $T_n = \Pi_k \mb T|_{B_k}$ (i.e. the restriction of $\Pi_k \mb T$ to $B_k$). The compactness condition is satisfied by Assumption \ref{a:fp:exist}(b) (recall that compactness of $\mb G$ implies compactness of $\mb T$). The fixed point $h$  has nonzero index by Assumption \ref{a:fp:exist}(c); see result (5) on p. 300 of \cite{Kras}. Finally, condition (19.28) in \cite{Kras} holds by Assumption \ref{a:fp:bias}(b) and their condition (19.29) is trivially satisfied.
\end{proof}

\begin{proof}[Proof of Remark \ref{rmk:nbhd}]
This follows by the proof of result (19.31) in Theorem 19.3 in \cite{Kras}.
\end{proof}

\begin{proof}[Proof of Remark \ref{rmk:fpunique}]
This follows by Theorem 19.7 in \cite{Kras}.
\end{proof}

\begin{proof}[Proof of Lemma \ref{lem:bias:fp}]
Part (c) follows by the proof of display (19.50) on p. 310 in \cite{Kras} where, in our notation, $x_0 = h$, $x_n = h_k$, $P_n = \Pi_k$, $P^{(n)} = I - \Pi_k$, $T = \mb T$, and $T'(x_0) = \mb D_h$. Note that Assumption \ref{a:fp:bias}(a) implies their condition $\| T'(x_0) - P_n T'(x_0)\| \to 0$ as $n \to \infty$.
Part (b) then follows from the inequality:
\[
  \left\| \frac{h}{\|h\|} - \frac{h_k}{\|h_k\|} \right\|  \leq \frac{2}{\|h\|} \|h - h_k\| \,.
\]
Finally, part (a) follows from the fact that $\big| \|h\| - \|h_k\| \big| = O(\tau_k)$ and continuous differentiability of $x \mapsto x^{1-\beta}$ at each $x > 0$.
\end{proof}

The next lemma presents some bounds on the estimators which are used in the proof of Lemmas \ref{lem:fphat:exist} and \ref{lem:var:fp}.

\begin{lemma}\label{lem:fp:matcgce}
(a) Let Assumptions \ref{a:fp:exist}(b) and \ref{a:fp:var} hold. Then:
\[
 \sup_{v \in \mb R^k : \|  v\|_{\mf G} \leq c} \| \wh{\mf G}^{-1}\wh{\mf T}v - \mf G^{-1} \mf T v\|_{\mf G}  = o_p(1) \,.
\]
(b) Moreover:
\[
 \sup_{v \in \mb R^k : \| v' b^k- h\| \leq \varepsilon} \|\wh{\mf G}^{-1} \wh{\mf T}  v - \mf G^{-1} \mf T v \|_{\mf G} = O_p(\nu_{n,k})
\]
where $\nu_{n,k}$ is from display (\ref{e:nudef}).
\end{lemma}

\begin{proof}[Proof of Lemma \ref{lem:fp:matcgce}]
By definition of $\wh{\mf G}^o$, $\wh{\mf T}^o$, and $\mf T^o$, we have 
\[
 \sup_{v \in \mb R^k : \|  v\|_{\mf G} \leq c} \| \wh{\mf G}^{-1}\wh{\mf T}v - \mf G^{-1} \mf T v\|_{\mf G}  = \sup_{v \in \mb R^k : \|  v\| \leq c} \| (\wh{\mf G}^o)^{-1}\wh{\mf T}^ov - \mf T^o v\|\,.
\]
Whenever $\|\mf I - \wh{\mf G}^o\| < 1$ (which it is wpa1 by Assumption \ref{a:fp:var}), for any $v \in \mb R^k$ we have:
\begin{align}
 (\wh{\mf G}^o)^{-1} \wh{\mf T}^ov - \mf T^ov  
 & = \wh{\mf T}^ov - \mf T^ov - (\wh{\mf G}^o)^{-1}(\wh{\mf G}^o - \mf I) \mf T^ov - (\wh{\mf G}^o)^{-1}(\wh{\mf G}^o - \mf I) (\wh{\mf T}^ov - \mf T^ov) \label{e:gtineq}
\end{align}
Part (a) follows by the triangle inequality and Assumption \ref{a:fp:var}, noting that $\sup_{v \in \mb R^k : \|v\| \leq c} \| \mf T^o v\|  \leq \sup_{\psi : \|\psi\| \leq c} \| \mb T \psi\| < \infty$ holds for each $c$ by Assumption \ref{a:fp:exist}(b).

Part (b) follows similarly by definition of $\wh{\mf G}^o$, $\wh{\mf T}^o$, $\mf T^o$, and $\nu_{n,k}$ in display (\ref{e:nudef}).
\end{proof}

\begin{proof}[Proof of Lemma \ref{lem:fphat:exist}]
Let $\varepsilon$, $K$ and $N_k$ be as in Lemma \ref{lem:fp:exist}. Also define the sets $N = \{ \psi \in L^2 : \|\psi - h\| < \varepsilon\}$, $\Gamma = \{ \psi \in L^2 : \| \psi - h\| = \varepsilon\}$, $\Gamma_k = \{ \psi \in B_k : \| \psi - h\| = \varepsilon\}$, $\mf N_k = \{ v \in \mb R^k : v'b^k(x) \in  N_k\}$, and $\boldsymbol \Gamma_k = \{ v \in \mb R^k : v'b^k(x) \in \Gamma_k\}$.

Let $\gamma(I - \mb T;\Gamma)$ denote the rotation of the field $(I - \mb T)\psi$ on $\Gamma$. Assumption \ref{a:fp:exist} implies that $|\gamma(I - \mb T;\Gamma)| = 1$; see result (5) on p. 300 of \cite{Kras}. Also notice that 
\begin{equation} \label{e:fpindex:ineq:1}
 \sup_{\psi \in \Gamma} \|\mb T \psi - \Pi_k \mb T \psi\| < \inf_{\psi \in \Gamma} \|\psi - \mb T \psi\|
\end{equation} holds for all $k$ sufficiently large by Assumption \ref{a:fp:bias}(b) (note that $\inf_{\psi \in \Gamma} \|\psi - \mb T \psi\| > 0$, otherwise $\mb T$ would have a fixed point on $\Gamma$, contradicting the definition of $\ol N$ in the proof of Lemma \ref{lem:fp:exist}). Result (2) on p. 299 of \cite{Kras} then implies that whenever (\ref{e:fpindex:ineq:1}) holds we have $|\gamma(I - \Pi_k \mb T;\Gamma)| = |\gamma(I - \mb T;\Gamma)|= 1$. Result (3) on p. 299 of \cite{Kras} then implies that $|\gamma(I - \Pi_k \mb T|_{B_k}; \Gamma_k)|= 1$ whenever (\ref{e:fpindex:ineq:1}) holds. Finally, by isomorphism, we have that $|\gamma( \mf I - {\mf G}^{-1} {\mf T} ; \boldsymbol \Gamma_k) | = 1$ whenever (\ref{e:fpindex:ineq:1}) holds.

We now show that the inequality: 
\begin{equation} \label{e:fpindex:ineq}
 \sup_{v \in \boldsymbol \Gamma_k} \|(\wh{\mf G}^{-1} \wh{\mf T} - {\mf G}^{-1} {\mf T})v\|_{\mf G}  <  \inf_{ \psi \in \Gamma_k} \| \psi - \Pi_k \mb T \psi\|
\end{equation}
holds wpa1. The left-hand side is $o_p(1)$ by Lemma \ref{lem:fp:matcgce}(a). For the right-hand side, we claim that  $\liminf_{k \to \infty} \inf_{ \psi \in \Gamma_k} \| \psi - \Pi_k \mb T \psi\| > 0$. Suppose the claim is false. Then there exists a subsequence $\{\psi_{k_l} : l \geq 1\}$ with $\psi_{k_l} \in \Gamma_{k_l}$ such that $\psi_{k_l} - \Pi_{k_l} \mb T \psi_{k_l} \to 0$. Since $\mb T$ is compact, there exists a convergent subsequence $\{ \mb T\psi_{k_{l_j}} : j \geq 1\}$. Let $\psi^* = \lim_{j \to \infty} \mb T \psi_{k_{l_j}}$. Then: 
\begin{align*}
 \| \psi_{k_{l_j}} - \psi^*\| & \leq \| \psi_{k_{l_j}} - \Pi_{k_{l_j}} \mb T \psi_{k_{l_j}} \| + \| \Pi_{k_{l_j}} \mb T \psi_{k_{l_j}} -\Pi_{k_{l_j}} \psi^*\| + \| \Pi_{k_{l_j}} \psi^* - \psi^*\|  \to 0
\end{align*}
as $j \to \infty$, where the first term vanishes by definition of $\psi_{k_l}$, the second vanishes by definition of $\psi^*$, and the third vanishes by Assumption \ref{a:fp:bias}(b). Therefore, $\psi^* \in \Gamma$. Moreover, by continuity of $\mb T$ and definition of $\psi^*$:
\[
 \| \mb T \psi^* - \psi^*\| \leq \| \mb T \psi^* - \mb T \psi_{k_{l_j}}\| + \| \mb T \psi_{k_{l_j}} - \psi^*\| \to 0
\]
as $j \to \infty$, hence $\psi^* \in \Gamma$ is a fixed point of $\mb T$. But this contradicts the fact that $h$ is the unique fixed point of $\mb T$ in $\ol N = N \cup \Gamma$ (cf. the proof of Lemma \ref{lem:fp:exist}). This proves the claim.

Result (2) on p. 299 of \cite{Kras} then implies that whenever (\ref{e:fpindex:ineq:1}) and (\ref{e:fpindex:ineq}) hold (which they do wpa1), we have $\gamma( \mf I - \wh{\mf G}^{-1} \wh{\mf T} ; \boldsymbol \Gamma_k) = \gamma( \mf I - {\mf G}^{-1} {\mf T} ; \boldsymbol \Gamma_k)$. Therefore, $|\gamma( \mf I - \wh{\mf G}^{-1} \wh{\mf T} ; \boldsymbol \Gamma_k)| = 1$ also holds wpa1 and hence, by result (1) on p. 299 of \cite{Kras}, $\wh{\mf G}^{-1} \wh{\mf T}$ has at least one fixed point $\hat v \in \mf N_k$. We have therefore shown that $\hat h(x) = b^k(x)'\hat v$ is well defined wpa1 and $\|\hat h - h\| < \varepsilon$ wpa1. Consistency of $\hat h$ follows by repeating the preceding argument with any positive $\varepsilon' < \varepsilon$.
\end{proof}

\begin{proof}[Proof of Remark \ref{rmk:nbhd:2}]
Fix any positive $\varepsilon' < \varepsilon$ and let $A = \{ \psi \in L^2: \varepsilon' \leq \| \psi - h\| \leq \varepsilon\}$, $A_k = \{ \psi \in B_k : \varepsilon' \leq \| \psi - h\| \leq \varepsilon\}$ and $\mf A_k = \{v \in \mb R^k :  v'b^k(x) \in A_k\}$. Clearly $\mb T$ has no fixed point in $A$. Moreover, similar arguments to the proof of result (19.31) in Theorem 19.3 in \cite{Kras} imply that $A_k$ contains no fixed points of $\Pi_k \mb T$ for all $k$ sufficiently large. By similar arguments to the proof of Lemma \ref{lem:fphat:exist} we may deduce that $\liminf_{k \to \infty} \inf_{ \psi \in A_k} \| \psi - \Pi_k \mb T \psi\| =: c^* > 0$. Then for any $v \in \mf A_k$, we have $\|v - \wh{\mf G}^{-1} \wh{\mf T}v\| \geq c^* - o_p(1)$ where the $o_p(1)$ term holds uniformly over $\mf A_k$ by Lemma \ref{lem:fp:matcgce}(a). Therefore, $\|v - \wh{\mf G}^{-1} \wh{\mf T}v\| \geq c^*/2$ holds for all $v \in \mf A_k$ wpa1. On the other hand, any fixed point $\hat v$ of $\wh{\mf G}^{-1} \wh{\mf T}$ with $b^k(x)'\hat v \in N_k$ necessarily has $ \|\hat v - \wh{\mf G}^{-1} \wh{\mf T}\hat v\| = 0$. Therefore, no such fixed point $\hat v$ belongs to $\mf A_k$ wpa1.
\end{proof}

\begin{proof}[Proof of Lemma \ref{lem:var:fp}]
We first prove part (c).  The Fr\'echet derivative of $\Pi_k \mb T|_{B_k}$ at $h$ is $\Pi_k \mb D_h |_{B_k}$. This may be represented on $(\mb R^k,\langle \cdot,\cdot\rangle_{\mf G})$ by the matrix $\mf G^{-1} \mf D_h$ where $\mf D_h = \mb E[b^k(X_t) \beta G_{t+1}^{1-\gamma} h(X_t)^{\beta-1} b^k(X_{t+1})']$. By Lemma \ref{lem:fphat:exist}, $\hat v$ (equivalently, $\hat h$) is well defined wpa1. Therefore, wpa1 we have:
\[
  (\mf I - \mf G^{-1} \mf D_{h})(v_k - \hat v)  = \mf G^{-1} \mf T \hat v  - \wh{\mf G}^{-1} \wh{\mf T} \hat v - \big( \mf G^{-1} \mf T \hat v  - \mf G^{-1} \mf T v_k - \mf G^{-1} \mf D_{h} (\hat v - v_k) \big) 
\]
Note that $\|\mf G^{-1} \mf T \hat v  - \wh{\mf G}^{-1} \wh{\mf T} \hat v\|_{\mf G} = O_p(\nu_{n,k})$ by Lemma \ref{lem:fp:matcgce}(b) and consistency of $\hat h$. Therefore,
\begin{align} \label{e:var:fp:0}
  \| (\mf I - \mf G^{-1} \mf D_{h})(v_k - \hat v)\|_{\mf G} 
 & \leq O_p(\nu_{n,k}) + \| \mf G^{-1} \mf T \hat v  - \mf G^{-1} \mf T v_k - \mf G^{-1} \mf D_{h} (\hat v - v_k) \|_{\mf G}\,.
\end{align}
By isomorphism, we have $\| (\mf I - \mf G^{-1} \mf D_{h})(v_k - \hat v)\|_{\mf G} = \|(I - \Pi_k \mb D_h)(h_k - \hat h)\|$. Assumptions \ref{a:fp:exist}(c) and \ref{a:fp:bias}(a) together imply that $(I - \Pi_k \mb D_h)^{-1}$ exists for all $k$ sufficiently large and the norms $\|(I - \Pi_k \mb D_h)^{-1}\|$ are uniformly bounded (for all $k$ sufficiently large). Therefore, 
\begin{equation} \label{e:var:fp:1}
 \| (\mf I - \mf G^{-1} \mf D_{h})(v_k - \hat v)\|_{\mf G} \geq \mr{const} \times \|h_k - \hat h\|
\end{equation}
holds  for all $k$ sufficiently large. Also notice that:
\begin{align}
 & \| \mf G^{-1} \mf T \hat v  - \mf G^{-1} \mf T v_k - \mf G^{-1} \mf D_{h} (\hat v - v_k) \|_{\mf G} \notag \\
 & = \| \Pi_k \mb T \hat h - \Pi_k \mb T h_k - \Pi_k \mb D_{h}(\hat h - h_k)\| \notag \\
 & \leq \| \mb T \hat h - \mb T h - \mb D_h(\hat h - h) - (\mb T h_k - \mb T h - \mb D_h( h_k - h)) \| \notag \\
 & \leq \| \mb T \hat h - \mb T h - \mb D_h(\hat h - h) \| +\|\mb T h_k - \mb T h - \mb D_h( h_k - h)\| \notag \\
 & =  o(1) \times ( \| \hat h - h_k\| + \|h_k - h\|) + o(1) \times \|h - h_k\|  \label{e:var:fp:2}
\end{align}
where the first inequality is because $\Pi_k$ is a (weak) contraction on $L^2$ and the final line is by Assumption \ref{a:fp:exist}(c).
Substituting (\ref{e:var:fp:1}) and (\ref{e:var:fp:2}) into (\ref{e:var:fp:0}) and rearranging, we obtain:
\[
 (1-o(1))\times  \| h_k - \hat h\| \leq O_p(\nu_{n,k}) + o_p(\tau_k) \,.
\] 
Parts (a) and (b) follow by similar arguments to the proof of Lemma \ref{lem:bias:fp}.
\end{proof}

\subsection{Proofs for Appendix \ref{ax:inf}}

\begin{proof}[Proof of Proposition \ref{p:asydist:L:1}]
First note that:
\begin{align*}
 \sqrt n (\hat L - L) & = \sqrt n \left( \log \hat \rho - \log \rho - \frac{1}{n} \sum_{t=0}^{n-1} \log m(X_t,X_{t+1}) + \mb E[\log m(X_t,X_{t+1})] \right) \\
 & = \frac{1}{\sqrt n} \sum_{t=0}^{n-1} ( \rho^{-1} \psi_{\rho,t} - \psi_{lm,t} ) + o_p(1)
\end{align*}
where the second line is by display (\ref{e:ale:1}) and a delta-method type argument. The result now follows from the joint convergence in the statement of the proposition.
\end{proof}

\begin{proof}[Proof of Proposition \ref{p:asydist:L:2a}]
Similar arguments to the proof of Proposition \ref{p:asydist:L:1} yield:
\begin{align*}
 \sqrt n (\hat L - L) & = \frac{1}{\sqrt n} \sum_{t=1}^n \Big( \rho^{-1} \psi_{\rho,t} + \rho^{-1} \phi^*_{k,t} \phi_{k,t+1} \big( m_t(\hat \alpha)  - m_t(\alpha_0) \big) \\
 & \quad \quad  - \big(\log m_t(\hat \alpha) - \log m_t(\alpha_0) \big)  - \psi_{lm,t} \Big) + o_p(1)\,.
\end{align*} 
By similar arguments to the proof of Theorem \ref{t:asydist:2a}, we may deduce:
\begin{align*}
 & \frac{1}{\sqrt n } \sum_{t=0}^{n-1} \Big( \rho^{-1} \phi^*_{k,t} \phi_{k,t+1} \big( m_t(\hat \alpha)  - m_t(\alpha_0) \big)   - \big(\log m_t(\hat \alpha) - \log m_t(\alpha_0) \big) \Big) \\
 & = D_{\alpha,lm} \sqrt n (\hat \alpha - \alpha_0) + o_p(1)
\end{align*}
where
\[
  D_{\alpha,lm}  = \mb E \left[\left( \frac{\phi^*(X_t) \phi(X_{t+1})}{\rho} - \frac{1}{m(X_t,X_{t+1},\alpha)} \right) \frac{\partial m(X_t,X_{t+1},\alpha)}{\partial \alpha'} \right] \,.
\]
Substituting into the expansion for $\hat L$ and using Assumption \ref{a:parametric}(a) yields:
\[
 \sqrt n (\hat L - L) = \frac{1}{\sqrt n} \sum_{t=1}^n \left( \rho^{-1} \psi_{\rho,t} + D_{\alpha,lm} \psi_{\alpha,t} - \psi_{lm,t} \right) + o_p(1)\,.
\]
The result follows by the joint CLT assumed in the statement of the proposition.
\end{proof}

\begin{proof}[Proof of Theorem \ref{t:eff:main}]
We prove part (1) first. We first characterize the tangent space as in pp. 878--880  of \cite{BickelKwon2001} (their arguments trivially extend to $\mb R^d$-valued Markov processes). Let $Q_2$ denote the stationary distribution of $(X_t,X_{t+1})$. Consider the tangent space $\mc H_0 = \{ h(X_t,X_{t+1}) : \mb E[h(X_t,X_{t+1})^2] < \infty$ and $\mb E[ h(X_t,X_{t+1})|X_t = x] = 0$ almost surely$\}$ endowed with the $L^2(Q_2)$ norm. Take any bounded  $h \in \mc H_0$ and consider the one-dimensional parametric model which we identify with the collection of transition probabilities $\{ P_1^{\tau,h} : |\tau| \leq 1\}$ where each transition probability $P_1^{\tau,h}$ is dominated by $P_1$ (the true transition probability) and is given by:
\[
 \frac{\mr d P_1^{\tau,h}(x_{t+1}|x_t)}{\mr d P_1(x_{t+1}|x_t)} = e^{\tau h(x_t,x_{t+1})-A(\tau,x_t)}
\]
where:
\[
 A(\tau,x_t) = \log\left( \int e^{\tau h(x_t,x_{t+1})} P_1(\mr d x_{t+1}|x_t) \right) \,.
\]

For each $\tau$ we define the linear operator $\mb M^{(\tau,h)}$ on $L^2$ by:
\[
 \mb M^{(\tau,h)} \psi(x_t) = \int m(x_t,x_{t+1}) \psi(x_{t+1})  P_1^{\tau,h}(\mr d x_{t+1}|x_t) \,.
\]
Observe that:
\begin{align}\label{e:eff:pf1}
 (\mb M^{(\tau,h)}-\mb M)\psi(x_t) = \int  m(x_t,x_{t+1}) \psi(x_{t+1}) \left( e^{\tau h(x_t,x_{t+1})-A(\tau,x_t)}-1 \right)P_1 (\mr d x_{t+1}|x_t)\,.
\end{align}
is a bounded linear operator on $L^2$ (since $\|\mb M\| < \infty$ and $h$ is bounded). By Taylor's theorem:
\begin{equation} \label{e:eff:pf2}
 e^{\tau h(x_t,x_{t+1})-A(\tau,x_t)}-1 = \tau h(x_t,x_{t+1}) + O(\tau^2)
\end{equation}
where the $O(\tau^2)$ term is uniform in $(x_t,x_{t+1})$. It now follows by boundedness of $h$ that $\|\mb M^{(\tau,h)}-\mb M\| = O(\tau)$. Similar arguments to the proof of Lemma \ref{lem:exist} imply that there exists $\epsilon > 0$  and $\bar \tau > 0$ such that the largest eigenvalue $\rho_{(\tau,h)}$ of $\mb M^{(\tau,h)}$ is simple and lies in the interval $(\rho - \epsilon,\rho + \epsilon)$ for each $\tau < \bar \tau$. Taking a perturbation expansion of $\rho_{(\tau,h)}$ about $\tau = 0$ (see, for example, equation (3.6) on p. 89 of \cite{Kato} which also applies in the infinite-dimensional case, as made clear in Section VII.1.5 of \cite{Kato}):
\begin{align}
 \rho_{(\tau,h)} - \rho & =  \langle (\mb M^{(\tau,h)}-\mb M) \phi, \phi^* \rangle + O(\tau^2) \notag \\
 & = \tau \mb E[ m(X_t,X_{t+1}) h(X_t,X_{t+1}) \phi(X_{t+1}) \phi^* (X_t)] + O(\tau^2) \notag \\
 & = \tau \int m(x_t,x_{t+1})\phi(x_{t+1})\phi^*(x_t) h(x_t,x_{t+1})\mr dQ_2(x_t,x_{t+1}) + O(\tau^2) \label{e:eff:pf3}
\end{align}
under the normalization $\langle \phi,\phi^*\rangle = 1$, where the second line is by (\ref{e:eff:pf1}) and (\ref{e:eff:pf2}). Expression (\ref{e:eff:pf3}) shows that the derivative of $\rho_{(\tau,h)}$ at $\tau = 0$ is $\tilde \psi_\rho = m(x_t,x_{t+1})\phi(x_{t+1})\phi^*(x_t)$. 

As bounded  functions are dense in $\mc H_0$, we have shown that $\rho$ is differentiable relative to $\mc H_0$ with derivative $\tilde \psi_\rho$. The efficient influence function for $\rho$ is the projection of $\tilde \psi_\rho$ onto $\mc H_0$, namely:
\[
 \tilde \psi_\rho(x_t,x_{t+1}) - \mb E[\tilde \psi_\rho(X_t,X_{t+1}) |X_t = x_t] = \psi_\rho(x_t,x_{t+1})
\]
because $\mb E[\tilde \psi_\rho(X_t,X_{t+1}) |X_t = x_t] = \phi^*(x_{t}) \mb M \phi(x_t) = \rho \phi(x_t) \phi^*(x_t)$. It follows that $V_\rho = \mb E[\psi_\rho(X_t,X_{t+1})^2]$ is the efficiency bound for $\rho$. A similar argument shows that $h'(\rho) \psi_\rho$ is the efficient influence function for $h(\rho)$.

We now prove part (2). By linearity, the efficient influence function for $L$ is:
\[
 \psi_L = \rho^{-1} \psi_\rho - \psi_{\log m}
\]
where $\psi_{\log m}$ is the efficient influence function for $\mb E[\log m(X_t,X_{t+1})]$. It is well known that:
\[
 \psi_{\log m}(x_0,x_1) = l(x_0,x_1) + \sum_{t=0}^\infty \Big( \mb  E[l(X_{t+1},X_{t+2})|X_1 = x_1] - \mb  E[l(X_t,X_{t+1})|X_0 = x_0]  \Big) \label{e-phim}
\]
where $l(x_t,x_{t+1}) =  \log m(x_t,x_{t+1})$
(see, e.g., \cite{GreenwoodWefelmeyer1995}). It may be verified using the telescoping property of the above sum that $V_L = \mb E[ \psi_L(X_0,X_1)^2]$.
\end{proof}

\begin{proof}[Proof of Lemma \ref{lem:expansion}]
Take $k \geq K$ from Lemma \ref{lem:exist} and work on the sequence of events upon which (\ref{e:kato:half}) holds, so that $\hat \rho$, $\hat c$ and $\hat c^*$ are uniquely defined by Lemma \ref{lem:exist:hat}.

Normalize  $\hat c$, $\hat c^*$, $c_k$, and $c_k^*$ so that $\| \hat c\|_{\mf G} = 1$, $\|c_k\|_{\mf G} = 1$, $\hat c ' \mf G \hat c^* = 1$ and $c_k' \mf G  c_k^* = 1$. Let $\mf P = c_k^{\phantom *} c_k^{* \prime} \mf G$ and $\wh{\mf P} = \hat c \hat c^{*\prime} \mf G$. We then have $\tr{\wh{\mf P}}=1$, $\tr{\mf P} = 1$, $\hat \rho = \tr{\wh{\mf P} \wh{\mf G}^{-1} \wh{\mf M}}$, $\rho_k = \tr{{\mf P} {\mf G}^{-1} {\mf M}}$, $\wh{\mf G}^{-1} \wh{\mf M}\wh{\mf P} = \hat \rho \wh{\mf P} $ and $\mf G^{-1} \mf M \mf P = \mf P \mf G^{-1} \mf M = \rho_k \mf P$. Now observe that:
\begin{align*}
 \hat \rho - \rho_k & = \tr{\wh{\mf P} \wh{\mf G}^{-1} \wh{\mf M}} - \tr{{\mf P} {\mf G}^{-1} {\mf M}} \\
 & = \tr{(\wh{\mf P} - \mf P) \wh{\mf G}^{-1} \wh{\mf M}} + \tr{{\mf P} (\wh{\mf G}^{-1} \wh{\mf M} - {\mf G}^{-1} {\mf M})} \,.
\end{align*}
By addition and subtraction of terms, we have:
\begin{align}
 & \tr{(\wh{\mf P} - \mf P) \wh{\mf G}^{-1} \wh{\mf M}} \notag \\
 & = \hat \rho - \hat \rho \tr{\mf P \wh{\mf P}} + \tr{\mf P \wh{\mf G}^{-1} \wh{\mf M} (\wh{\mf P} - \mf I)} \notag \\
 & = \hat \rho \tr{\mf P (\mf I -  \wh{\mf P})} + \tr{\mf P \wh{\mf G}^{-1} \wh{\mf M} (\wh{\mf P} - \mf I)} \notag \\
 & = (\hat \rho - \rho_k) \tr{\mf P (\mf I -  \wh{\mf P})} + \rho_k \tr{\mf P (\mf I -  \wh{\mf P})} + \tr{\mf P \wh{\mf G}^{-1} \wh{\mf M} (\wh{\mf P} - \mf I)} \notag \\
 & = (\hat \rho - \rho_k) \tr{\mf P (\mf I -  \wh{\mf P})} + \tr{\mf P \mf G^{-1} \mf M (\mf I -  \wh{\mf P})} + \tr{\mf P \wh{\mf G}^{-1} \wh{\mf M} (\wh{\mf P} - \mf I)} \notag \\
 & = (\hat \rho - \rho_k) \tr{\mf P (\mf I -  \wh{\mf P})} + \tr{\mf P (\wh{\mf G}^{-1} \wh{\mf M} - \mf G^{-1} \mf M)(\wh{\mf P} - \mf I)}  \label{e:pgbd1:00}
\end{align}
where:
\begin{equation} \label{e:pgbd1:0}
  |\tr{\mf P (\mf I -  \wh{\mf P})}| = |c_k^{*\prime} \mf G (c_k - \wh{\mf P} c_k) | \leq \|c_k^*\|_{\mf G} \| c_k - \wh{\mf P} c_k\|_{\mf G} \,.
\end{equation}
By the proof of Proposition 4.2 of \cite{Gobetetal} (setting $\wh{\mf P} = P_\varepsilon$, $\wh{\mf G}^{-1}\wh{\mf M} = T_\varepsilon$, ${\mf G}^{-1}{\mf M} = T$ and $\Gamma $ from the proof of Lemma \ref{lem:exist} as the boundary of $B(\kappa,\rho)$ in their notation) and similar arguments to the proof of Lemma \ref{lem:var}:
\begin{equation} \label{e:pgbd1}
 \| c_k - \wh{\mf P} c_k\|_{\mf G} \lesssim \|(\wh{\mf G}^{-1} \wh{\mf M} - \mf G^{-1} \mf M) c_k\|_{\mf G} = O_p(\eta_{n,k})\,.
\end{equation}
Moreover,
\[
 \| c_k^*\|_{\mf G} = \| \mf P c_k\|_{\mf G} \leq \| \mf P\|_{\mf G} \leq \left\| \frac{-1}{2 \pi \mr i} \int_\Gamma \frac{\mc R(\mf G^{-1} \mf M,z)}{\rho_k - z }\, \mr dz \right\|_{\mf G}
\]
\cite[expression (6.19), p. 178]{Kato} and which is $O(1)$ by display (\ref{e:resbd:k}) and the fact that $\rho_k \to \rho$. By displays (\ref{e:pgbd1:0}) and (\ref{e:pgbd1}) and the fact that $\hat \rho - \rho_k = O_p(\eta_{n,k})$ (by Lemma \ref{lem:var}), we obtain:
\begin{equation} \label{e:pgbd2}
 (\hat \rho - \rho_k) \tr{\mf P (\mf I -  \wh{\mf P})} = O_p(\eta_{n,k}^2) \,.
\end{equation}
Moreover:
\begin{align}
 |\tr{\mf P (\wh{\mf G}^{-1} \wh{\mf M} - \mf G^{-1} \mf M)(\wh{\mf P} - \mf I)} |   
 & = | c_k^{*\prime} \mf G(\wh{\mf G}^{-1} \wh{\mf M} - \mf G^{-1} \mf M)(\wh{\mf P} - \mf I) c_k| \notag \\
 & \leq \|c_k^*\|_{\mf G} \|\wh{\mf G}^{-1} \wh{\mf M} - \mf G^{-1} \mf M\|_{\mf G} \| c_k - \wh{\mf P} c_k\|_{\mf G} \notag \\
 & = O_p( \eta_{n,k,1} +  \eta_{n,k,2}) \times O_p(\eta_{n,k}) \label{e:pgbd3}
\end{align}
by Lemma \ref{lem:matcgce}(b) and display (\ref{e:pgbd1}). It follows by (\ref{e:pgbd1:00}), (\ref{e:pgbd2}) and (\ref{e:pgbd3}) that:
\[
 \hat \rho - \rho_k =  \tr{{\mf P} (\wh{\mf G}^{-1} \wh{\mf M} - {\mf G}^{-1} {\mf M})} + O_p( \eta_{n,k,1} +  \eta_{n,k,2} ) \times O_p(\eta_{n,k}) + O_p(\eta_{n,k}^2)\,.
\]
Finally,
\begin{align*}
 \tr{{\mf P} (\wh{\mf G}^{-1} \wh{\mf M} - {\mf G}^{-1} {\mf M})} 
 & = c_k^{*\prime} \mf G (\wh{\mf G}^{-1} \wh{\mf M} - \mf G^{-1} \mf M) c_k \\
 & = \tilde c_k^{*\prime}  ((\wh{\mf G}^o)^{-1} \wh{\mf M}^o - \mf M^o) \tilde c_k \\
 & = \tilde c_k^{*\prime}  ( \wh{\mf M}^o - \wh{\mf G}^o \mf M^o) \tilde c_k + O_p( \eta_{n,k,1} \times (\eta_{n,k,1} + \eta_{n,k,2})) 
\end{align*}
by Lemma \ref{lem:matcgce}(a) and the fact that $\| \tilde c_k^*\| = \|c_k^*\|_{\mf G} = O(1)$. The result follows by noting that 
\[
 \tilde c_k^{*\prime}  ( \wh{\mf M}^o - \wh{\mf G}^o \mf M^o ) \tilde c_k =  c_k^{*\prime}  ( \wh{\mf M} - \rho_k \wh{\mf G}  )  c_k 
\]
and that $\eta_{n,k}$ is of at least as small order as $\eta_{n,k,1}$ and $\eta_{n,k,2}$ (cf. Lemma \ref{lem:matcgce}(a)).
\end{proof}

{\small \singlespacing
\putbib
}
\end{bibunit}

\begin{bibunit}

\newpage
\clearpage
\pagenumbering{arabic}\renewcommand{\thepage}{\arabic{page}}
\setcounter{equation}{0}
\renewcommand{\theequation}{OA.\arabic{equation}}

\begin{center}
{\Large Online Appendix for 

Nonparametric Stochastic Discount Factor Decomposition}

\vskip 24pt
{\large Timothy M. Christensen}

\vskip 8pt
{\large May 19, 2017}

\end{center}

\vskip 8pt

This Online Appendix contains material to support the paper ``Nonparametric Stochastic Discount Factor Decomposition''. Appendix \ref{s:mc:supp} presents additional simulation evidence. Appendix \ref{ax:id} provides further details on the relation between the identification and existence conditions in Section \ref{s:id} and the identification and existence conditions in \cite{HS2009} and \cite{BHS}. Appendix \ref{ax:proofs:supp} presents proofs of results in Appendix \ref{ax:suff} of the supplementary material and this online appendix.

\section{Additional Monte Carlo evidence}\label{s:mc:supp}

This section presents additional simulation results using a cubic B-spline basis of dimension $k = 8$ for the Monte Carlo design described in Section \ref{s:mc} of the main text.  The knots of the B-splines were placed evenly at the empirical quantiles of the data. As with the results obtained using Hermite polynomials, the simulation results were reasonably insensitive to the dimension of the sieve space.

Tables \ref{tab:mc1:b} and \ref{tab:mc2:b} present bias and RMSE of the estimators across simulations. Figures \ref{fig:mc:sfig1:bs}--\ref{fig:mc:sfig5:bs} present (pointwise) confidence intervals for $\phi$, $\phi^*$ and $\chi$ computed across simulations of different sample sizes.

\begin{table}[p]
\small
\centering 
\begin{tabular}{cc|cc|ccc} 
\hline \hline
 & & \multicolumn{2}{c|}{Power Utility} & \multicolumn{3}{c}{Recursive Preferences} \\
 & $n$ &$\hat \phi$  &$\hat \phi^*$& $\hat \phi$& $\hat \phi^*$& $\hat \chi$ \\ \hline 
\multirow{4}{*}{Bias} 	& 400 & 0.0144 & 0.0141 & 0.0009 & 0.0241 & 0.0116 \\
 						& 800 & 0.0113 & 0.0132 & 0.0011 & 0.0190 & 0.0086 \\
 						& 1600& 0.0078 & 0.0101 & 0.0010 & 0.0145 & 0.0057 \\ 
 						& 3200& 0.0049 & 0.0068 & 0.0009 & 0.0128 & 0.0034 \\ \hline 
\multirow{4}{*}{RMSE} 	& 400 & 0.1106 & 0.1334 & 0.0283 & 0.3479 & 0.0988 \\ 
 						& 800 & 0.0851 & 0.1043 & 0.0270 & 0.3151 & 0.0734 \\ 
 						& 1600& 0.0650 & 0.0814 & 0.0235 & 0.2747 & 0.0547 \\ 
 						& 3200& 0.0500 & 0.0627 & 0.0222 & 0.1702 & 0.0414 \\   \hline \hline
\end{tabular} \vskip 4pt
\parbox{5.0in}{\caption{\label{tab:mc1:b} \small
Simulation results for $\hat \phi$, $\hat \phi^*$ and $\hat \chi$ with a cubic B-spline sieve of dimension $k =8$.}} 
\end{table}

\begin{table}[p]
\small
\centering
\begin{tabular}{cc|ccc|cccc}
\hline \hline
 & & \multicolumn{3}{c|}{Power Utility} & \multicolumn{4}{c}{Recursive Preferences} \\
 & $n$ & $\hat \rho$ & $\hat y$    & $\hat L$    & $\hat \rho$ & $\hat y$    & $\hat L$ & $\hat \lambda$ \\ \hline 
\multirow{4}{*}{Bias} 	& 400 & 0.0036 &-0.0030 & 0.0030 & 0.0010 &-0.0009 & 0.0031 & 0.0028  \\ 
						& 800 & 0.0027 &-0.0024 & 0.0024 & 0.0011 &-0.0011 & 0.0027 & 0.0019 \\
						& 1600& 0.0019 &-0.0017 & 0.0017 & 0.0010 &-0.0009 & 0.0020 & 0.0013 \\
						& 3200& 0.0012 &-0.0011 & 0.0011 & 0.0007 &-0.0006 & 0.0013 & 0.0008 \\ \hline 
\multirow{4}{*}{RMSE} 	& 400 & 0.0345 & 0.0330 & 0.0272 & 0.0154 & 0.0130 & 0.0305 & 0.0348 \\ 
 						& 800 & 0.0254 & 0.0244 & 0.0206 & 0.0155 & 0.0133 & 0.0244 & 0.0209 \\ 
 						& 1600& 0.0190 & 0.0182 & 0.0157 & 0.0163 & 0.0136 & 0.0208 & 0.0153 \\ 
 						& 3200& 0.0142 & 0.0135 & 0.0118 & 0.0148 & 0.0123 & 0.0165 & 0.0110 \\ \hline \hline
\end{tabular} \vskip 4pt
\parbox{5.0in}{\caption{\label{tab:mc2:b} \small
Simulation results for $\hat \rho$, $\hat y$, $\hat L$ and $\hat \lambda$ with a cubic B-spline sieve of dimension $k =8$.}} 
\end{table}

\begin{figure}[p]
\centering
\begin{subfigure}{.5\textwidth}
  \centering
  \includegraphics[width=\linewidth]{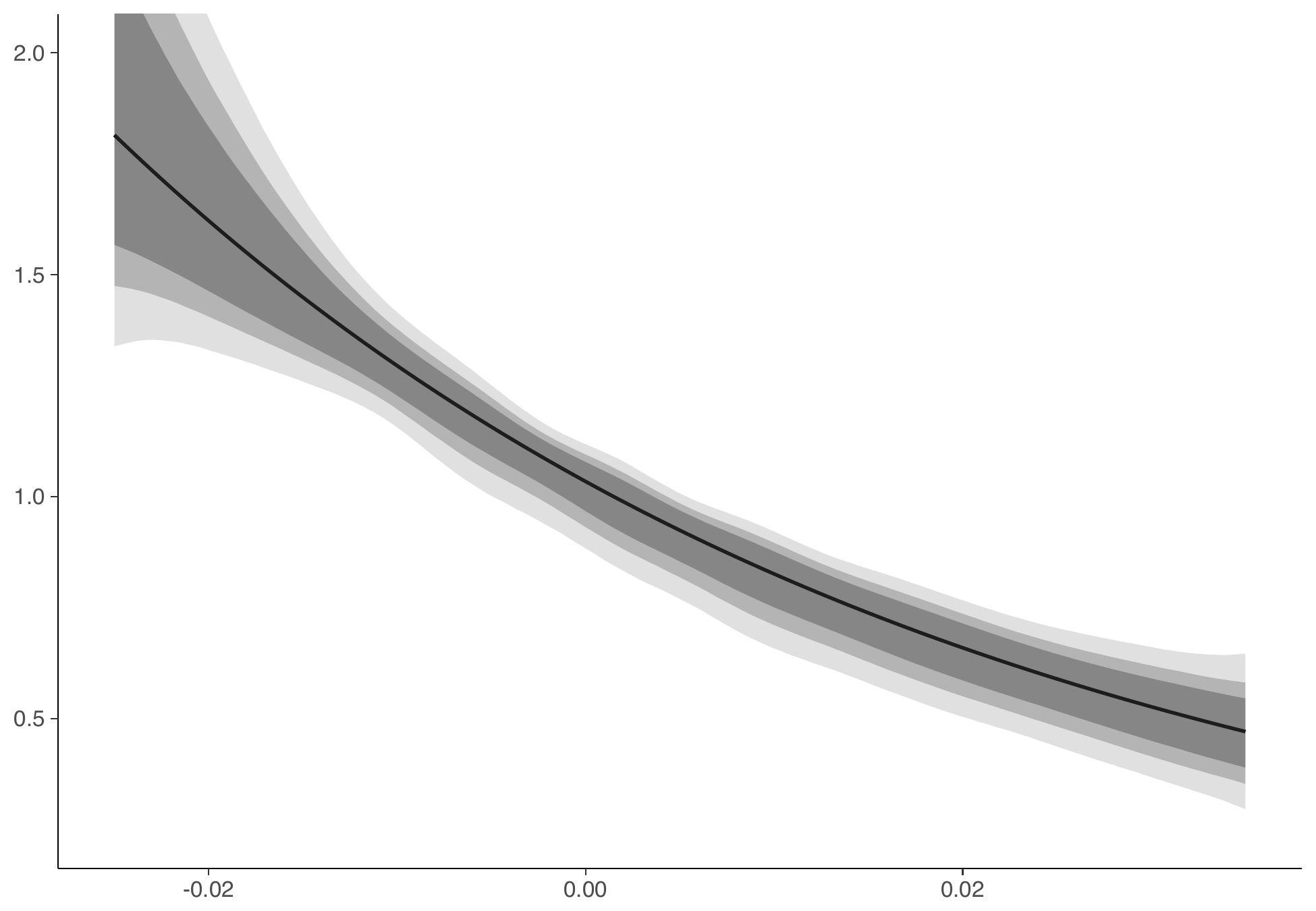}
  \caption{\footnotesize $\hat \phi(x)$ for power utility}
  \label{fig:mc:sfig1:bs}
\end{subfigure}%
\begin{subfigure}{.5\textwidth}
  \centering
  \includegraphics[width=\linewidth]{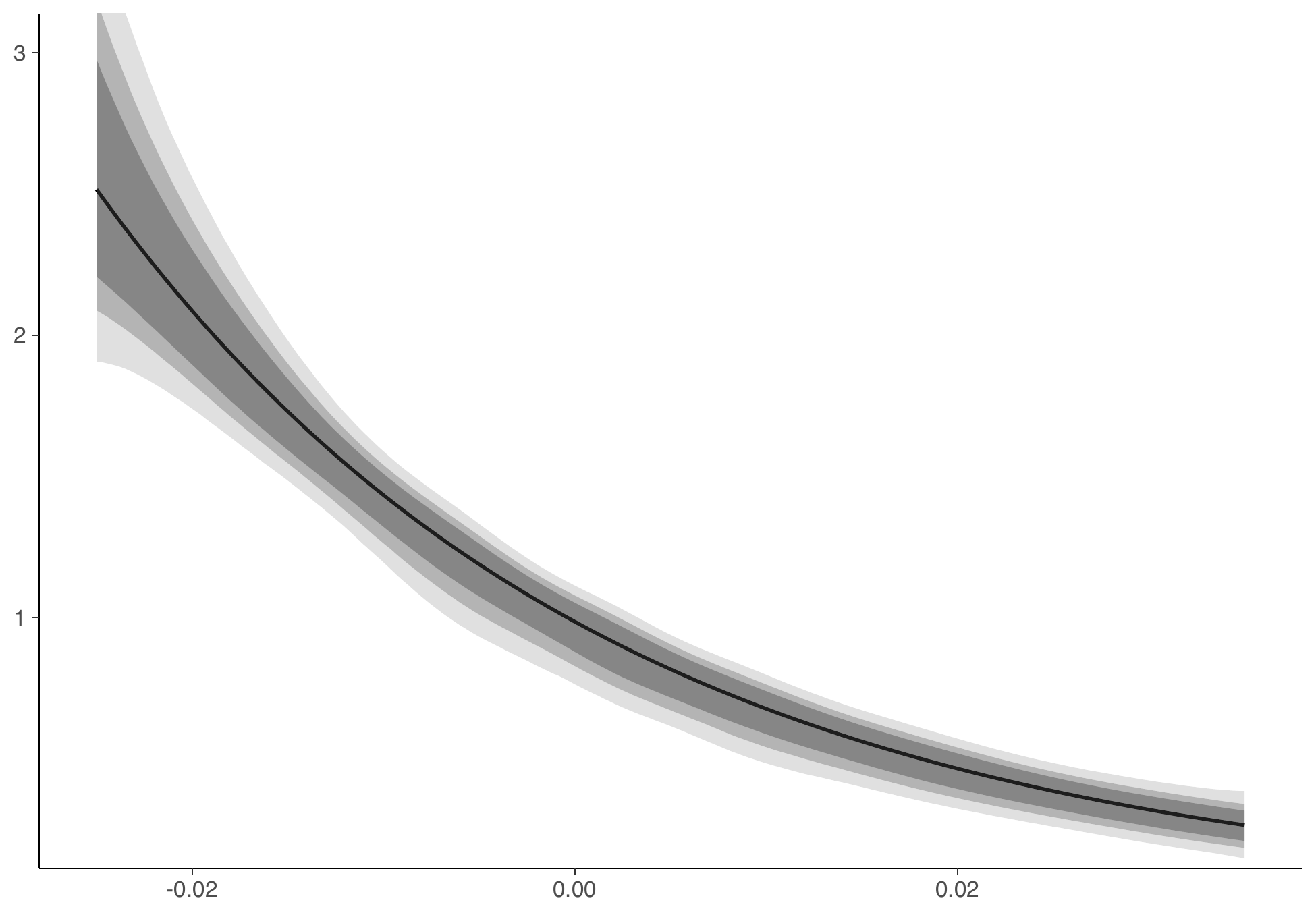}
  \caption{\footnotesize $\hat \phi^*(x)$ for power utility}
  \label{fig:mc:sfig2:bs}
\end{subfigure} \\[10pt]
\begin{subfigure}{.5\textwidth}
  \centering
  \includegraphics[width=\linewidth]{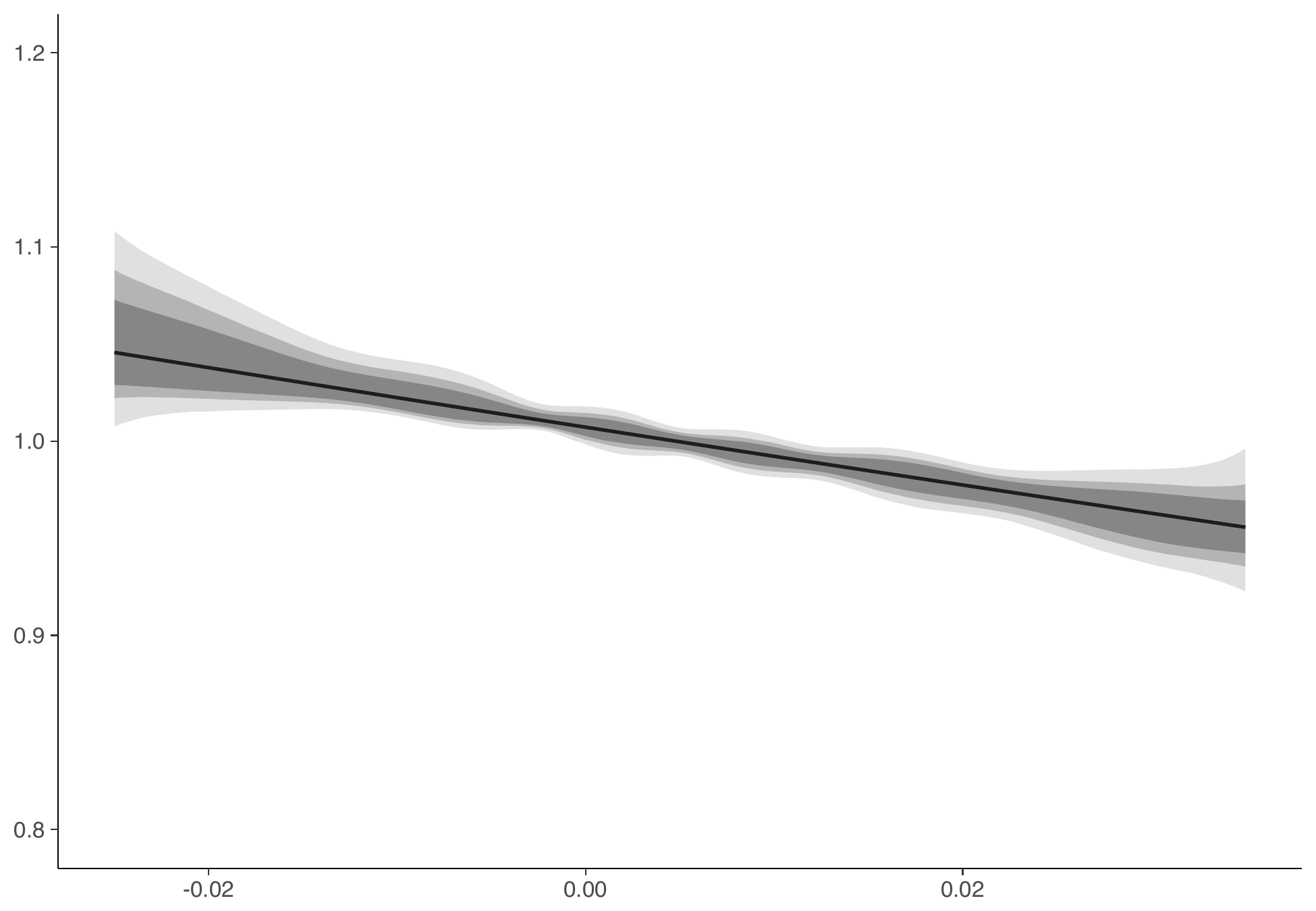}
  \caption{\footnotesize $\hat \phi(x)$ for recursive preferences}
  \label{fig:mc:sfig3:bs}
\end{subfigure}%
\begin{subfigure}{.5\textwidth}
  \centering
  \includegraphics[width=\linewidth]{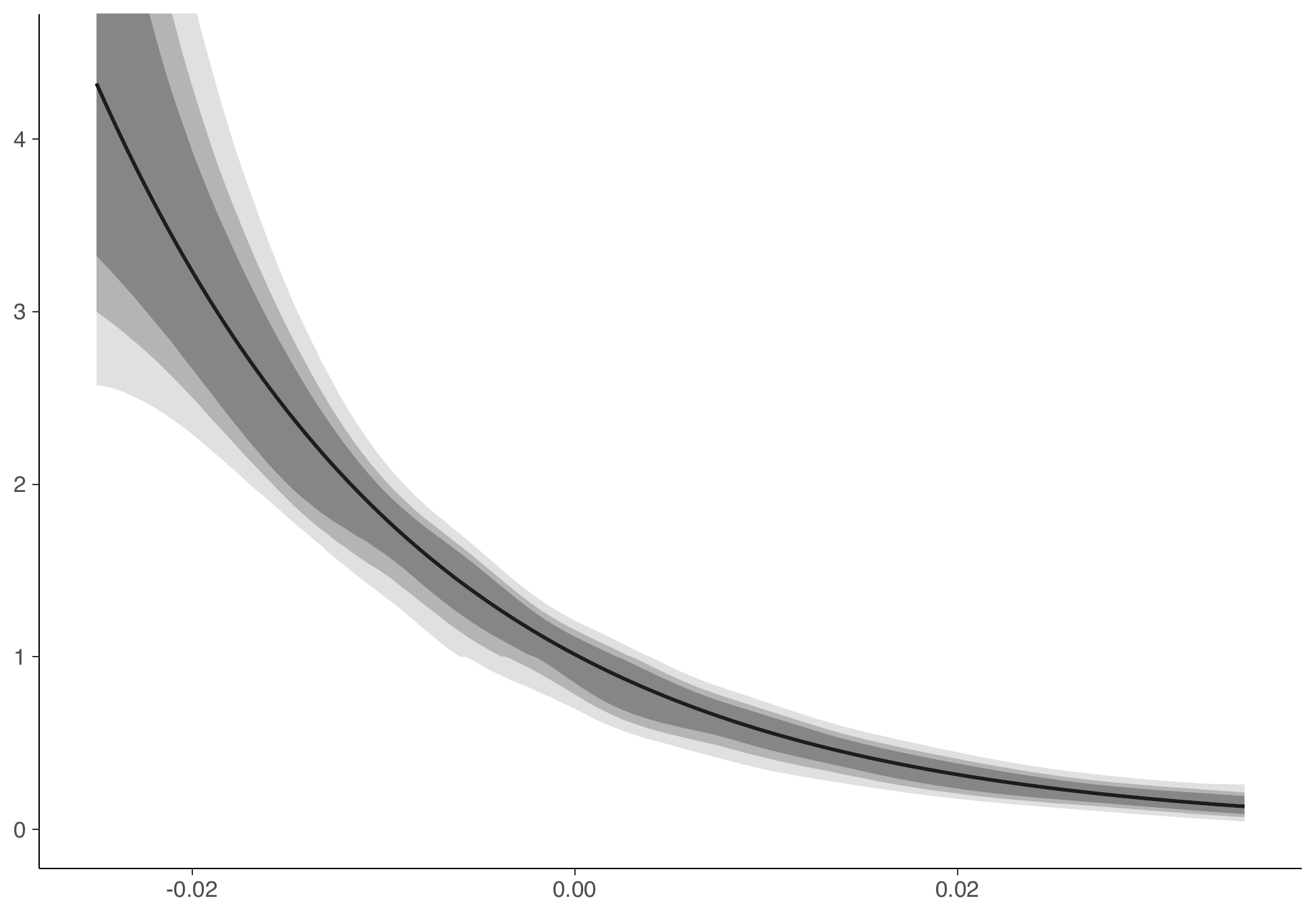}
  \caption{\footnotesize $\hat \phi^*(x)$ for recursive preferences}
  \label{fig:mc:sfig4:bs}
\end{subfigure} \\[10pt]
\begin{subfigure}{.5\textwidth}
  \centering
  \includegraphics[width=\linewidth]{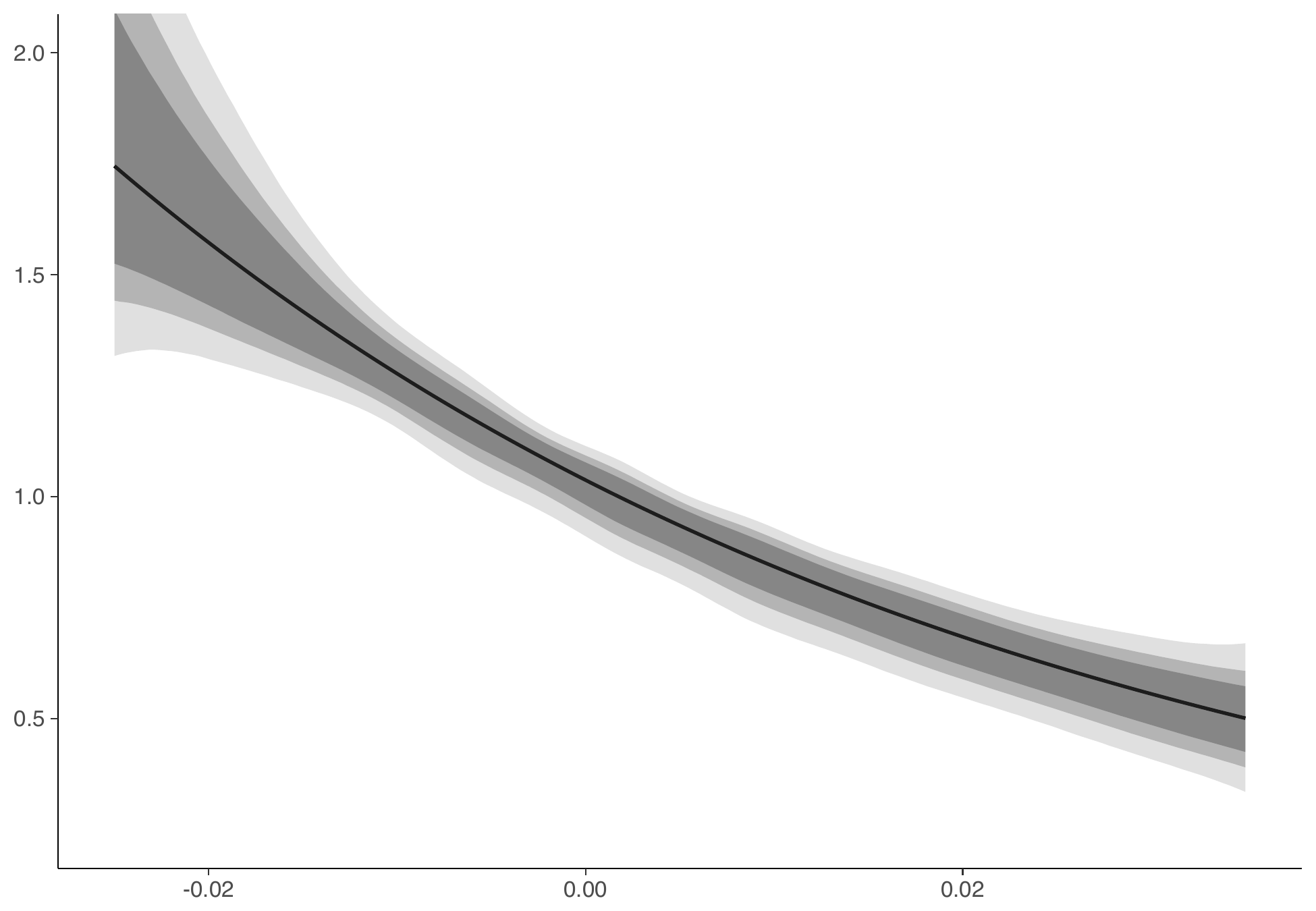}
  \caption{\footnotesize $\hat \chi(x)$ for recursive preferences}
  \label{fig:mc:sfig5:bs}
\end{subfigure}
\begin{center}
\parbox{5.0in}{\caption{ \small Simulation results for a cubic B-spline basis with $k=8$. Panels (a)--(d) display pointwise 90\% confidence intervals for $\phi$ and $\phi^*$ across simulations (light, medium and dark correspond to $n= 400$, $800$, and $1600$ respectively; the true $\phi$ and $\phi^*$ plotted as solid lines). Panel (e) displays results for the positive eigenfunction $\chi$ of the continuation value operator}}
\label{fig:mc:bs}
\end{center}
\end{figure}

\section{Additional results on identification} \label{ax:id}

In this appendix we discuss separately existence and identification, and compare the conditions in the present paper with the stochastic stability conditions in \cite{HS2009} (HS hereafter) and \cite{BHS} (BHS hereafter).

\subsection{Identification}

\begin{assumption} \label{a:id:1}
Let the following hold:
\begin{enumerate}
\item[(a)] $\mb M$ is bounded
\item[(b)] There exists positive functions $\phi,\phi^* \in L^2$ and a positive scalar $\rho$ such that $(\rho,\phi)$ solves (\ref{e:pev}) and $(\rho,\phi^*)$ solves (\ref{e:pev:star})
\item[(c)] $\mb M \psi$ is positive for each non-negative  $\psi \in L^2$ that is not identically zero.
\end{enumerate}
\end{assumption}

Note that no compactness or power-compactness condition appears in Assumption \ref{a:id:1}.

\begin{proposition}\label{p:a:id}
Let Assumption \ref{a:id:1} hold. Then: the functions $\phi$ and $\phi^*$ are the unique solutions (in $L^2$) to (\ref{e:pev}) and (\ref{e:pev:star}), respectively.
\end{proposition}

We now compare the identification results with those in HS and BHS. Some of HS's conditions related to the generator of the semigroup of conditional expectation operators $\wt{\mb E}[\cdot|X_t = x]$ under the change of conditional probability induced by $M_t^P$, namely:
\begin{equation} \label{e:twist}
 \wt{\mb E} [\psi(X_{t+\tau})|X_t = x] := \mb E \bigg[ \frac{M_{t+ \tau}^P}{M_t^P} \psi(X_{t+\tau}) \bigg| X_t = x\bigg] \,.
\end{equation}
In discrete-time environments, both  multiplicative functionals and semigroups are indexed by non-negative integers. Therefore, the ``generator'' in discrete-time is just the single-period distorted conditional expectation operator $\psi \mapsto \wt{\mb E}[\psi(X_{t+1})|X_t = \cdot\,]$.

The following are discrete-time versions of Assumptions 6.1, 7.1, 7.2, 7.3, and 7.4 in HS.

\begin{condition}\label{c:HS} \begin{enumerate}
\item[(a)] $\{ M^P_t : t \in T\}$ is a positive multiplicative functional
\item[(b)] There exists a probability measure $\hat \varsigma$ such that 
\[
 \int \wt{\mb E}[\psi(X_{t+1})|X_t = x]\,\mr d \hat \varsigma(x) = \int \psi(x)\,\mr d \hat \varsigma(x)
\]
for all bounded measurable $\psi : \mathcal X \to \mb R$
\item[(c)] For any $\Lambda \in \mathscr X$ with $\hat \varsigma(\Lambda) > 0$, 
\[
 \wt{\mb E}\left[ \left. \sum_{t=1}^\infty \ind \{X_t \in \Lambda\} \right| X_0 = x \right] > 0
\]
for all $x \in \mathcal X$
\item[(d)] For any $\Lambda \in \mathscr X$ with $\hat \varsigma(\Lambda) > 0$, 
\[
 \wt{\mb P} \left( \left. \sum_{t=1}^\infty \ind \{X_t \in \Lambda\} = \infty \right| X_0 = x \right) = 1
\]
for all $x \in \mathcal X$, where $$\wt{\mb P} (\{X_s\}_{s=0}^t \in A | X_0 = x) = \int {\mb E}[(M_t^P/M_0^P) \ind \{\{X_s\}_{s=0}^t \in A\} |X_0 = x]\,\mr d \hat \varsigma(x)$$ for each $A \in \mathcal F_t$.
\end{enumerate}
\end{condition}

\bigskip

Condition \ref{c:HS}(a) is satisfied by construction of $M^P$ in (\ref{e:pctc}). For Condition \ref{c:HS}(b), let $\phi$ and $\phi^*$ be as in Assumption \ref{a:id:1}(b) and normalize $\phi^*$ such that $\mb E[\phi(X_t)\phi^*(X_t)] = 1$. Under this normalization we can define a probability measure $\hat \varsigma$ by $\hat \varsigma(A) = \mb E[\phi(X_t) \phi^*(X_t) \ind\{X_t \in A\}]$ for all $A \in \mathscr X$. Proposition \ref{p:lr} below shows that this probability measure is precisely the measure used to define the unconditional expectation $\wt{\mb E}$ in the long-run approximation (\ref{e:lrr}). Recall that $Q$ is the stationary distribution of $X$. We then have:
\begin{align*}
 & \int \wt{\mb E}[\psi(X_{t+1})|X_t = x]\,\mr d \hat \varsigma(x) \\
 & = \int \mb E \left[ \left. \rho^{-1} m(X_t,X_{t+1}) \frac{\phi(X_{t+1})}{\phi(X_t)} \psi(X_{t+1}) \right| X_t = x \right] \phi(x) \phi^*(x)\,\mr d Q(x) \\
 & = \rho^{-1} \mb E \left[ \phi^*(X_t) (\mb M(\phi \psi) ( X_t)) \right] \\
 & = \rho^{-1} \mb E \left[ ((\mb M^* \phi^*)(X_{t+1})) \phi(X_{t+1}) \psi(X_{t+1})  \right] \\
 & = \mb E[\phi^*(X_{t+1}) \phi(X_{t+1}) \psi(X_{t+1})]  = \int \psi(x) \,\mr d \hat \varsigma(x)\,.
\end{align*}
Therefore, Condition \ref{c:HS}(b) is satisfied. A similar derivation is reported for continuous-time semigroups in an preliminary 2005 draft of HS with $Q$ replaced by an arbitrary measure.

For Condition \ref{c:HS}(c), note that $\hat \varsigma(\Lambda) > 0$ implies $Q(\Lambda ) > 0$ under our construction of $\hat \varsigma$. Therefore, $\hat \varsigma(\Lambda) > 0$ implies $\phi(x) \ind\{x \in \Lambda\}$ is positive on a set of positive $Q$ measure. Moreover, by definition of $\wt{\mb E}$ we have:
\begin{eqnarray*}
 \wt{\mb E}\left[ \left. \sum_{t=1}^\infty \ind\{X_t \in \Lambda\} \right| X_0 = x \right] & = &  \frac{1}{\phi(x)} \sum_{t=1}^\infty \rho^{-t} \mathbb M_t (\phi(\cdot) \ind\{ \cdot \in \Lambda\})(x) \\
 & \geq & \frac{1}{\phi(x)} \sum_{t=1}^\infty \lambda^{-t} \mathbb M_t(\phi(\cdot) \ind\{ \cdot \in \Lambda\})(x) 
\end{eqnarray*}
for any $\lambda \geq r(\mb M)$ where $r(\mb M)$ denotes the spectral radius of $\mb M$. 
Assumption \ref{a:id:1}(c) implies $\mb M$ is irreducible and, by definition of irreducibility, $\sum_{t=1}^\infty \lambda^{-t} \mathbb M_t(\phi(\cdot) \ind\{ \cdot \in \Lambda\})(x)  > 0$ (almost everywhere) holds for $\lambda > r(\mb M)$. Therefore, Assumption \ref{a:id:1}(c) implies Condition \ref{c:HS}(c), up to the ``almost everywhere'' qualification.

Part (d) is a Harris recurrence condition which does not translate clearly in terms of the operator $\mathbb M$. When combined with existence of an invariant measure and irreducibility (Condition \ref{c:HS}(b) and (c), respectively), it ensures both uniqueness of $\hat \varsigma$ as the invariant measure for the distorted expectations as well as $\phi$-ergodicity, i.e., 
\begin{equation}\label{e:HScgce}
 \lim_{\tau \to \infty} \sup_{0 \leq \psi \leq \phi} \left| \wt {\mb E}\left[\left. \frac{\psi(X_{t+\tau})}{\phi(X_{t+\tau})}\right|X_t = x\right] - \int \frac{\psi(x)}{\phi(x)}\,\mr d \hat \varsigma(x) \right| = 0 
\end{equation}
(almost everywhere) where the supremum is taken over all measurable $\psi$ such that $0 \leq \psi \leq \phi$ \cite[Proposition 14.0.1]{MeynTweedie2009}. Result (\ref{e:HScgce}) is a discrete-time version of Proposition 7.1 in HS, which they use to establish identification of $\phi$. Assumption \ref{a:id:1} alone is not enough to obtain a convergence result like (\ref{e:HScgce}). On the other hand, the conditions in the present paper assume existence of $\phi^*$ whereas no positive eigenfunction of the adjoint of $\mb M$ is guaranteed under the conditions in HS. Indeed, for non-stationary environments it is not even clear how to restrict the class of functions appropriately to define an adjoint (for instance, HS do not appear to restrict $\phi$ to belong to a Banach space). This suggests the Harris recurrence condition is of a very different nature from Assumption \ref{a:id:1}.

BHS assume that $X$ is ergodic under the $\wt{\mb P}$ probability measure, for which Conditions \ref{c:HS}(b)--(d) are sufficient. Also notice that Condition \ref{c:HS}(a) is satisfied by construction in BHS.

The identification results in HS and the proof of proposition 3.3 in BHS shows that uniqueness is established in the space of functions $\psi$ for which $\wt{\mb E}[\psi(X_t)/\phi(X_t)]$ is finite, where $\wt{\mb E}$ denotes expectation under the stationary distribution corresponding to (\ref{e:twist}). Under Assumption \ref{a:id:1}, their result establishes identification in the space of functions $\psi$ for which 
\[
 \wt{\mb E}[\psi(X_t)/\phi(X_t)] = \mb E[ \psi(X_t) \phi^*(X_t)]
\]
is finite. The right-hand side is finite for all $\psi \in L^2$ (by Cauchy-Schwarz). So in this sense the identification result in HS and BHS applies to a larger class of functions than our result.

\subsection{Existence}

We obtain the following existence result by replacing Assumption \ref{a:id:1}(b)(c) by the slightly stronger quasi-compactness and positivity conditions in Assumption \ref{a:id:0}. The following result is essentially Theorems 6 and 7 of \cite{Sasser}.\footnote{I thank an anonymous referee for bringing Theorems 6 and 7 of \cite{Sasser} to my attention. Theorems 6 and 7 of \cite{Sasser} replace Assumption \ref{a:id:0}(a) in Proposition \ref{p:a:exist} by the condition that  $\mb M$ is \emph{quasi-positive}, i.e. for each non-negative $\psi$ and $\psi^*$ in $L^2$ that are not identically zero there exists $\tau \in T$ such that $\langle \psi^*, \mb M_\tau \psi \rangle > 0$. Notice that quasi-compactness also requires that $r(\mb M) > 0$. Assumption \ref{a:id:0}(a) is sufficient for these two conditions (i.e. quasi-positivity and $r(\mb M) > 0$). The condition $r(\mb M) > 0$ together with power-compactness of $\mb M$ (Assumption \ref{a:id:0}(b)) is sufficient for quasi-compactness.}  Say that $\mb M$ is \emph{quasi-compact} if $\mb M$ is bounded and there exists $\tau \in T$ and a bounded linear operator $\mb V$ such that $\mb M_\tau - \mb V$ is compact and $r(\mb V) < r(\mb M)^\tau$. Quasi-compactness of $\mb M$ is implied by Assumption \ref{a:id:0}.

\medskip

\begin{proposition}\label{p:a:exist}
Let Assumption \ref{a:id:0}(a) hold and let $\mb M$ be quasi-compact. Then:
\begin{enumerate}
\item[(a)] There exists positive functions $\phi,\phi^* \in L^2$ and a positive scalar $\rho$ such that $(\rho,\phi)$ solves (\ref{e:pev}) and $(\rho,\phi^*)$ solves (\ref{e:pev:star}).
\item[(b)] The functions $\phi$ and $\phi^*$ are the unique solutions (in $L^2$) to (\ref{e:pev}) and (\ref{e:pev:star}), respectively.
\item[(c)] The eigenvalue $\rho$ is simple and isolated and it is the largest eigenvalue of $\mb M$.
\end{enumerate}
\end{proposition}

\bigskip

A similar existence result to part (a) was presented in a 2005 preliminary version of HS. For that result, HS assumed that $r(\mb M)$ was positive and that the (continuous-time) semigroup of operators had an element which was compact. The further properties of $\rho$ that we establish in part (c) of Proposition \ref{p:a:exist} are essential to our derivation of the large-sample theory. A similar proposition was derived under different conditions in \cite{Christensen-idpev}.

HS establish existence of $\phi$ in possibly non-stationary, continuous-time environments by appealing to the theory of ergodic Markov processes. Equivalent conditions for discrete-time environments are now presented and compared with our identification conditions. As with the identification conditions, we use analogues of generators and resolvents for discrete-time semigroups where appropriate.

\medskip

\begin{condition}\label{c:HS:ex}
\begin{enumerate}
\item[(a)] There exists a function $V : \mathcal X \to \mb R$ with $V \geq 1$ and a finite constant $\underline a> 0 $ such that $\mb M V(x) \leq \underline a V(x)$ for all $x \in \mathcal X$
\item[(b)] There exists a measure $\nu$ on $(\mathcal X, \mathscr X)$ such that $\mb J\ind\{ \cdot \in \Lambda\}(x) > 0$ for any $\Lambda \in \mathscr X$ with $\nu(\Lambda) > 0$, where $\mb J$ is given by
\[
 \mb J \psi(x) = \sum_{t=0}^\infty a^{-(t+1)} \frac{\mb M_t (V \psi)(x)}{V(x)}
\]
for $a > \underline a$
\item[(c)] The operators $\mb J$ and $\mb K$ are bounded, where $\mb K$ is given by
\[
 \mb K \psi(x) =  \sum_{t=0}^\infty \lambda^{-t}((\mb J - s \otimes \nu)^t \psi)(x)
\]
where $s: \mathcal X \to \mb R_+$ is such that $\int s \, \mr d \nu > 0$ and $\mb J \psi(x) \geq s(x) \int \psi\,\mr d \nu$ for all $\psi \geq 0$ ($s$ exists by part (b)), $(s \otimes \nu) \psi(x) := s(x) \int \psi \,\mathrm d \nu$, and $\lambda\in \sigma(\mb J)$. 
\end{enumerate}
\end{condition}

\medskip

HS show that $\mb K s$ is a positive eigenfunction of $\mb M$ under the preceding conditions (see their Lemma D.3). Condition \ref{c:HS:ex}(b) is satisfied under Assumption \ref{a:id:0} with $\nu = Q$ whenever $a > r(\mb M)$. To see this, take $\Lambda \in \mathscr X$ with $Q(\Lambda) > 0$ and observe that: 
\[
 \sum_{t=1}^\infty a^{-t} \mb M_t (V(\cdot) \ind\{ \cdot \in \Lambda\})(x) \geq \sum_{t=1}^\infty a^{-t} \mb M_t \ind\{ \cdot \in \Lambda\} > 0
\] 
(almost everywhere) where the first inequality is by positivity and the second is by irreducibility. It follows that $\mb J \ind\{ \cdot \in \Lambda\}(x) > 0$ (almost everywhere). This verifies part (b), up to the ``almost everywhere'' qualification. 

On the other hand, Conditions \ref{c:HS:ex}(a)(c)  seem quite different from the conditions of Proposition \ref{p:a:exist}. For instance, Assumption \ref{a:id:0} does not presume existence of the function $V$ but imposes a quasi-compactness condition. HS do not restrict the function space for $\mb M$ ex ante so there is no notion of a bounded or power-compact operator on the space to which $\phi$ belongs. The requirement that $\mb K$ be bounded (or the sufficient conditions for this provided in HS) do not seem to translate clearly in terms of the operator $\mb M$. 

\subsection{Long-run pricing}

We now present a version of the long-run pricing approximation of HS that holds under our existence and identification conditions. We impose the normalization $\mb E[\phi(X_t)\phi^*(X_t)] = 1$ and define the operator $(\phi \otimes \phi^*) : L^2 \to L^2$ by: 
\[
 (\phi \otimes \phi^*) \psi(x) = \phi(x) \int \phi^* \psi\,\mr d Q\,.
\]

\medskip

\begin{proposition}\label{p:lr}
Let Assumption \ref{a:id:0} hold. Then: there exists $c > 0$ such that:
\[
 \|\rho^{-\tau}\mb M_\tau - (\phi \otimes \phi^*)\| = O( e^{-c\tau})
\]
as $\tau \to \infty$.
\end{proposition}

\bigskip

Proposition \ref{p:lr} is similar to Proposition 7.4 in HS. Proposition \ref{p:lr} establishes convergence of $\rho^{-\tau} \mb M_\tau$ to $(\phi \otimes \phi^*)$, with the approximation error vanishing exponentially in the payoff horizon $n$. A similar proposition (without the rate of convergence) was reported in a 2005 draft of HS. There, HS assumed directly that the distorted conditional expectations converged to an unconditional expectation characterized by $\phi$, $\phi^*$, and an arbitrary measure. Proposition \ref{p:lr} shows that in stationary environments the unconditional expectation $\wt{\mb E}[\psi(X_t)/\phi(X_t)]$ appearing in the long-run approximation (\ref{e:lrr}) is characterized by $\phi$, $\phi^*$ and $Q$, namely:
\[
 \wt {\mb E} \left[ \frac{\psi(X_t)}{\phi(X_t)} \right] = \mb E[\psi(X_t) \phi^*(X_t)]\,.
\]

\section{Proofs of results in Appendices \ref{ax:suff} and \ref{ax:id}} \label{ax:proofs:supp}

\subsection{Proofs for Appendix \ref{ax:est:mat}}

\begin{proof}[Proof of Lemma \ref{lem:beta:1}]
Lemma 2.2 of \cite{ChenChristensen-reg} gives the bound $\|\wh{\mf G}^o - \mf I\| = O_p( \xi_k (\log n)/\sqrt n)$. We first prove that $\|\wh{\mf M}^o - \mf  M^o \| = O_p ( \xi_k^{1+2/r} (\log n)/\sqrt n )$. 

Let $\{T_n : n \geq 1\}$ be a sequence of positive constants to be defined below. Let  $\tilde b^k = \mf G^{-1/2} b^k$ be the orthogonalized basis functions and let $\Xi_{t,n} = n^{-1} \tilde b^k(X_t) m(X_t,X_{t+1}) \tilde b^k(X_{t+1})'$. Write:
\begin{align*}
 \wh{\mf M}^o - \mf M^o & = \sum_{t=0}^{n-1} \Xi_{t,n}^{trunc} + \sum_{t=0}^{n-1} \Xi_{t,n}^{tail} \quad \mbox{where} \\
 \Xi_{t,n}^{trunc} & = \Xi_{t,n} \ind\{ \|\Xi_{t,n}\| \leq T_n/n\}  - \mb E[\Xi_{t,n}\ind\{ \|\Xi_{t,n} \| \leq T_n/n\} ] \\
 \Xi_{t,n}^{tail}  & = \Xi_{t,n} \ind\{ \|\Xi_{t,n}\| > T_n/n\}  - \mb E[\Xi_{t,n}\ind\{ \|\Xi_{t,n} \| > T_n/n\} ] \,.
\end{align*}
Note $\mb E[\Xi_{t,n}^{trunc}] = 0$ and $\|\Xi_{t,n}^{trunc}\| \leq 2 n^{-1} T_n$ by construction. Let $S^{k-1} = \{u \in \mb R^k : \|u\| = 1\}$. For any $u,v \in S^{k-1}$ and any $0 \leq t,s \leq n-1$, we have:
\begin{align*}
  |u'\mb E[ \Xi_{t,n}^{trunc} (\Xi_{s,n}^{trunc})' ]v| 
 & \lesssim \frac{\xi_k^2}{n^2} \mb E[|u'\tilde b^k(X_t) m(X_t,X_{t+1}) m(X_s,X_{s+1}) \tilde b^k(X_s)'v|] \\
 & \leq \frac{\xi_k^2}{n^2} \mb E[|m(X_t,X_{t+1})|^r]^{2/r}  \times \mb E[|(u'\tilde b^k(X_t))|^q]^{1/q} \times \mb E[|(v'\tilde b^k(X_s))|^q]^{1/q} \\
 & \lesssim \frac{\xi_k^2}{n^2} \mb E[|(u'\tilde b^k(X_t))|^q]^{1/q} \times \mb E[|(v'\tilde b^k(X_s))|^q]^{1/q} 
\end{align*}
where the second line is by H\"older's inequality choosing $q$ such that $1 = \frac{2}{r} + \frac{2}{q}$ and the third is because $\mb E[|m(X_t,X_{t+1})|^r] < \infty$. Since $E[(\tilde b^k(X_0)'u)^2] = \|u\|^2 = 1$  for any $u \in S^{k-1}$, we have:
\[
  \mb E[|(u'\tilde b^k(X_t))|^q]^{1/q} \leq ( \xi_k^{q-2} \mb E[(u'\tilde b^k(X_t))^2])^{1/q} = \xi_k^{1-2/q}
\]
and so:
\[
 \|\mb E[ \Xi_{t,n}^{trunc} (\Xi_{s,n}^{trunc})' ]\| \lesssim  \sup_{u,v \in  S^{k-1}}|u'\mb E[ \Xi_{t,n}^{trunc} (\Xi_{s,n}^{trunc})' ]v|  = O(\xi_k^{2+4/r}/n^2)\,.
\]
The same argument gives $\|\mb E[ (\Xi_{t,n}^{trunc})' \Xi_{s,n}^{trunc} ]\| = O(\xi_k^{2+4/r}/n^2)$. By Corollary 4.2 of \cite{ChenChristensen-reg}:
\[
 \bigg\| \sum_{t=0}^{n-1} \Xi_{t,n}^{trunc}  \bigg\| = O_p ( \xi_k^{1+2/r}(\log n)/\sqrt n )
\]
provided $T_n (\log n)/n = o(\xi_k^{1+2/r}/\sqrt n)$.

Now consider the remaining term. If $m$ is bounded we can set $\Xi_{t,n}^{tail} \equiv 0$ by taking $T_n = C\xi_k^2$ for sufficiently large $C$. Otherwise, by the triangle and Jensen inequalities:
\begin{align*}
 \mb E \bigg[ \bigg\| \sum_{t=0}^{n-1} \Xi_{t,n}^{tail} \bigg\| \bigg] 
 & \leq 2n  \mb E[ \| \Xi_{t,n}\|\ind\{ \|\Xi_{t,n}\| > T_n/n \}  ] \\
 & \leq \frac{2n^r}{T_n^{r-1}}  \mb E[ \| \Xi_{t,n}\|^r \ind\{ \|\Xi_{t,n}\| > T_n/n \}  ]
  \leq \frac{2\xi_k^{2r}}{T_n^{r-1}} \mb E[|m(X_0,X_1)|^r]\,.
\end{align*}
By Markov's inequality:
\[
 \bigg\| \sum_{t=0}^{n-1} \Xi_{t,n}^{tail}  \bigg\| = O_p ( \xi_k^{2r}/T_n^{r-1} )\,.
\]
choosing $T_n$ so that $\xi_k^{2r}/T_n^{r-1} \asymp \xi_k^{1+2/r}(\log n)/\sqrt n$, we obtain:
\[
 \bigg\| \sum_{t=0}^{n-1} \Xi_{t,n}^{tail} \bigg\| = O_p ( \xi_k^{1+2/r}(\log n)/\sqrt n )\,.
\]
The condition $T_n (\log n)/n = o(\xi_k^{1+2/r}/\sqrt n)$ is, with this choice of $T_n$, equivalent to the condition $(\xi_k (\log n)/\sqrt n )^{(r-2)/(r-1)} = o(1)$, which holds because $\xi_k (\log n)/\sqrt n = o(1)$ and $r > 2$. We have therefore shown that $\|\wh{\mf M}^o - \mf  M^o \| = O_p ( \xi_k^{1+2/r} (\log n)/\sqrt n )$.

Result (1)  now follows from Lemma \ref{lem:matcgce}(b), noting that 
\[
 \|(\wh{\mf G}^o)^{-1} \wh{\mf M}^o - \mf M^o \| = O_p ( \xi_k^{1+2/r} (\log n)/\sqrt n )
\]
which is $o_p(1)$ under the condition $\xi_k^{1+2/r} (\log n)/\sqrt n = o(1)$. Result (2) follows from Result (1) and definition of the operator norm. Result (3) is immediate from the fact that $\|\wh{\mf G}^o - \mf I\| = O_p( \xi_k (\log n)/\sqrt n)$ and $\|\wh{\mf M}^o - \mf  M^o \| = O_p ( \xi_k^{1+2/r} (\log n)/\sqrt n )$.
\end{proof}

\begin{proof}[Proof of Lemma \ref{lem:rho:1}]
Similar arguments to the proof of Lemmas 4.8 and 4.12 of \cite{Gobetetal} give the bounds $\|\wh{\mf G}^o - \mf I\| = O_p ( \xi_k \sqrt{ k/n} )$,  $\|(\wh{\mf G}^o - \mf I)\tilde c_k\| = O_p ( \xi_k /\sqrt n )$, and $\|\tilde c_k^{* \prime}(\wh{\mf G}^o - I)\| = O_p ( \xi_k /\sqrt n )$. We first establish analogous bounds for $\wh{\mf M}^o$.

Let $u_1,\ldots,u_k$ be an orthonormal basis for $\mb R^k$. Then: 
\begin{align*}
 \mb E[ \|\wh{\mf M}^o - \mf M^o\|^2] & \leq  \sum_{l=1}^k \mb E[ \|(\wh{\mf M}^o - \mf M^o)u_l\|^2] \\
 & = \sum_{l=1}^k \sum_{j=1}^k \mr{Var} \left[ \frac{1}{n} \sum_{t=1}^n (\tilde b_{kj}(X_t)^2 m(X_t,X_{t+1})^2 (\tilde b^k(X_{t+1})' u_l) \right] \,.
\end{align*}
Now, by the covariance inequality for rho-mixing processes:
\begin{align*}
  \mb E[ \|\wh{\mf M}^o - \mf M^o\|^2] & \leq \frac{C}{n} \sum_{l=1}^k \sum_{j=1}^k  \mb E \left[ \tilde b_{kj}(X_t)^2 m(X_t,X_{t+1})^2 (\tilde b^k(X_{t+1})' u_l)^2 \right] \\
 & \leq \frac{C \xi_k^2}{n} \sum_{l=1}^k \mb E \left[ m(X_t,X_{t+1})^2 (\tilde b^k(X_{t+1})' u_l)^2 \right] 
\end{align*}
where the constant $C$ depends only on the rho-mixing coefficients. By H\"older's inequality:
\begin{align*}
 \mb E [ m(X_t,X_{t+1})^2 (\tilde b^k(X_{t+1})' u_l)^2 ] 
 & \leq \mb E[|m(X_0,X_1)|^r]^{2/r} \times \mb E[ (\tilde b^k(X_0)' u_l)^{\frac{2r}{r-2}}]^{\frac{r-2}{r}} \\
 & \leq \mb E[|m(X_0,X_1)|^r]^{2/r} \times \xi_k^{4/r} \times \mb E[ (\tilde b^k(X_0)' u_l)^2]^{\frac{r-2}{r}} \lesssim \xi_k^{4/r}
\end{align*}
since $\mb E[|m(X_0,X_1)|^r] < \infty$ and $\|u_l\| = 1$. Substituting into the above, we obtain
\[
 \mb E[ \|\wh{\mf M}^o - \mf M^o\|^2] \lesssim  \xi_k^{2+4/r} k / n
\]
which, by Markov's inequality, yields $\|\wh{\mf M}^o - \mf  M^o \| = O_p ( \xi_k^{1+2/r} \sqrt{k/n} ) $. Similar arguments give $\|(\wh{\mf M}^o - \mf M^o)\tilde c_k\| = O_p ( \xi_k^{1+2/r}/\sqrt n )$ and $\|\wt c_k^{* \prime} (\wh{\mf M}^o - \mf M^o)\| = O_p ( \xi_k^{1+2/r}/\sqrt n )$.

Result (1) now follows from Lemma \ref{lem:matcgce}(b), noting that 
\[
 \|(\wh{\mf G}^o)^{-1} \wh{\mf M}^o - \mf M^o \| = O_p ( \xi_k^{1+2/r}\sqrt{k/ n} )
\]
which is $o_p(1)$ under the condition $\xi_k^{1+2/r} \sqrt{k/ n} = o(1)$. 

For result (2), note that whenever $\| \wh{\mf G}^o - \mf I\| \leq \frac{1}{2}$, we have $\|(\wh{\mf G}^o)^{-1}\| \leq 2$ and hence:
\begin{align*}
 \| ((\wh{\mf G}^o)^{-1} \wh{\mf M}^o - \mf M^o) \tilde c_k\| 
 & \leq  \| (\wh{\mf G}^o)^{-1} (\wh{\mf M}^o - \mf M^o) \tilde c_k\| + \| ((\wh{\mf G}^o)^{-1} - \mf I)\mf M^o \tilde c_k\| \\
 & \leq 2 \|  (\wh{\mf M}^o - \mf M^o)\tilde c_k\| + 2 \rho_k \| (\wh{\mf G}^o - \mf I )\tilde c_k\| \,.
\end{align*} 
The result for $\tilde c_k$ follows form the bounds $\|(\wh{\mf G}^o - \mf I)\tilde c_k\| = O_p ( \xi_k /\sqrt n )$ and $\|(\wh{\mf M}^o - \mf M^o)\tilde c_k\| = O_p ( \xi_k^{1+2/r}/\sqrt n )$. The result for $\tilde c_k^*$ follows similarly.

Result (3) is immediate from the fact that $\|\wh{\mf G}^o - \mf I\| = O_p( \xi_k \sqrt{k/ n})$ and $\|\wh{\mf M}^o - \mf  M^o \| = O_p ( \xi_k^{1+2/r} \sqrt{k/ n} )$.
\end{proof}

\begin{proof}[Proof of Lemma \ref{lem:beta:2}]
The proof will follow by the same arguments as the proof of results (1)--(3) in Lemma \ref{lem:beta:1}, provided we show that $\|\wh{\mf M}^o - \mf M^o\| = O_p(\xi_k^{1+2/r} (\log n)/\sqrt n)$ also holds in this case. First write:
\begin{align*}
 \wh{\mf M}^o - \mf M^o  
  & = \left(  \frac{1}{n} \sum_{t=0}^{n-1} \tilde b^k(X_t) \Big(  m(X_t,X_{t+1};\hat \alpha) - m(X_t,X_{t+1};\alpha_0) \Big) \tilde b^{k}(X_{t+1}) \right) \\
  & \quad + \left( \frac{1}{n} \sum_{t=0}^{n-1} \tilde b^k(X_t) m(X_t,X_{t+1};\alpha_0) \tilde b^{k}(X_{t+1}) -  \mf M^o  \right)   
   =: \wh \Delta_{1,k} + \wh \Delta_{2,k}
\end{align*}
where $\|\wh \Delta_{2,k}\| = O_p ( \xi_k^{1+2/r} (\log n)/\sqrt n )$ by the proof of Lemma \ref{lem:beta:1}. For $\wh \Delta_{1,k}$, condition (a) implies that $\hat \alpha \in N$ wpa1. Whenever $\hat \alpha \in N$ we may take a mean value expansion (valid by condition (b)) to obtain:
\begin{align*}
 \|  \wh \Delta_{1,k} \| 
 & = \left\|  \frac{1}{n} \sum_{t=0}^{n-1} \tilde b^k(X_t)  \tilde b^{k}(X_{t+1})'  \left( \frac{\partial m(X_t,X_{t+1};\tilde \alpha)}{\partial \alpha'} (\hat \alpha - \alpha_0) \right) \right\|  \quad \mbox{wpa1}
\end{align*}
for $\tilde \alpha$ in the segment between $\hat \alpha$ and $\alpha_0$. Therefore, wpa1 we have:
\begin{align*}
 \|  \wh \Delta_{1,k} \| & = \sup_{u,v \in S^{k-1}} \left| \frac{1}{n} \sum_{t=0}^{n-1} (u'\tilde b^k(X_t) ) (v' \tilde b^{k}(X_{t+1})) \left( \frac{\partial m(X_t,X_{t+1};\tilde \alpha)}{\partial \alpha'} (\hat \alpha - \alpha_0) \right) \right| \\
 & \leq \xi_k \times \left( \sup_{u \in S^{k-1}} \frac{1}{n} \sum_{t=0}^{n-1} |u'\tilde b^k(X_t) | \times \bar m(X_t,X_{t+1}) \right) \times \|\hat \alpha - \alpha_0 \| \\
 & \leq \xi_k \times \left( \sup_{u \in S^{k-1}} u' \wh{\mf G}^o u \right)^{1/2} \times \left( \frac{1}{n} \sum_{t=0}^{n-1} \bar m(X_t,X_{t+1})^2 \right)^{1/2} \times \|\hat \alpha - \alpha_0 \| 
\end{align*}
where the first line is because $\|\mf A\| = \sup_{u,v \in S^{k-1}} |u'\mf A v|$ and the second and third lines are by condition (b) and the H\"older and Cauchy-Schwarz inequalities. Finally, notice that $\sup_{u \in S^{k-1}} u' \wh{\mf G}^o u = \| \wh{\mf G}^o \| = 1 + o_p(1)$ by the proof of Lemma \ref{lem:beta:1}, and $\frac{1}{n} \sum_{t=0}^{n-1} \bar m(X_t,X_{t+1})^2  = O_p(1)$ by the ergodic theorem and condition (b). Therefore $\|\wh \Delta_{1,k}\| = O_p( \xi_k/\sqrt n)$ and so $\|\wh{\mf M}^o - \mf M^o\| = O_p(\xi_k^{1+2/r} (\log n)/\sqrt n)$, as required.
\end{proof}

\begin{proof}[Proof of Lemma \ref{lem:beta:3}]
The proof will follow by the same arguments as the proof of results (1)--(3) in Lemma \ref{lem:beta:1}, provided we show that 
\[
 \|\wh{\mf M}^o - \mf M^o\| = O_p \left( \frac{\xi_k^{1+2/r}(\log n)}{\sqrt n} + \frac{ \xi_k^{2-\frac{2s-v}{2sv}} \sqrt{ k \log k} }{\sqrt n}  \right)\,.
\] As in the proof of Lemma \ref{lem:beta:2}, it suffices to bound:
\[
 \wh \Delta_{1,k} :=   \frac{1}{n} \sum_{t=0}^{n-1} \tilde b^k(X_t) \Big(  m(X_t,X_{t+1};\hat \alpha) - m(X_t,X_{t+1};\alpha_0) \Big) \tilde b^{k}(X_{t+1}) \,.
\]
Let $h_\alpha(x_0,x_1) =  m(x_0,x_1; \alpha) - m(x_0,x_1;\alpha_0)$ and let:
\begin{align*}
 h_\alpha^{trunc}(x_0,x_1) & = h_\alpha(x_0,x_1) \ind\{ \| \tilde b^k(x_0)\| \| \tilde b^k(x_1)\| E(x_0,x_1) \leq T_n\} \\
 h_\alpha^{tail}(x_0,x_1) & = h_\alpha(x_0,x_1) \ind\{ \| \tilde b^k(x_0)\| \| \tilde b^k(x_1)\| E(x_0,x_1) > T_n\} 
\end{align*}
where $\{T_n : n \geq 1\}$ be a sequence of positive constants to be defined below. Then:
\begin{align*}
 \| \wh \Delta_{1,k}\|  \leq & \,\sup_{\alpha \in \mc A} \left\| \frac{1}{n} \sum_{t=0}^{n-1} \tilde b^k(X_t) h^{trunc}_\alpha(X_t,X_{t+1}) \tilde b^{k}(X_{t+1}) - \mb E[ \tilde b^k(X_t) h^{trunc}_\alpha(X_t,X_{t+1}) \tilde b^{k}(X_{t+1})  ] \right\| \\
 & + \sup_{\alpha \in \mc A} \left\| \frac{1}{n} \sum_{t=0}^{n-1} \tilde b^k(X_t) h^{tail}_\alpha(X_t,X_{t+1})  \tilde b^{k}(X_{t+1})\right\| \\
 & + \sup_{\alpha \in \mc A} \left\| \mb E[ \tilde b^k(X_t) h^{tail}_\alpha(X_t,X_{t+1}) \tilde b^{k}(X_{t+1})  ] \right\|  + \left\| \mb E[ \tilde b^k(X_t) h_{\hat \alpha}(X_t,X_{t+1}) \tilde b^{k}(X_{t+1})] \right\| \\
  =: &\, \wh \Delta_{1,k,1} + \wh \Delta_{1,k,2} + \wh \Delta_{1,k,3} + \wh \Delta_{1,k,4} \,.
\end{align*} 
Let $\mc H_{n,k} = \{ (c_0'\tilde b^k(x_0))(c_1'\tilde b^k(x_1))h^{trunc}_\alpha(x_0,x_1) : c_0,c_1 \in S^{k-1}, \alpha \in \mc A\}$ where $S^{k-1}$ is the unit sphere in $\mb R^k$. Then:
\[
 \wh \Delta_{1,k,1} \leq n^{-1/2} \times {\textstyle \sup_{h \in \mc H_{n,k}} | \mc Z_n(h)|}
\]
by definition of the operator norm, where $\mc Z_n$ is the centered empirical process on $\mc H_{n,k}$. By Theorem 2 of \cite{DoukhanMassartRio}:
\begin{equation} \label{e:dmr:bd}
 \mb E[{\textstyle \sup_{h \in \mc H_{n,k}} | \mc Z_n(h)|}] = O  \left( \varphi(\sigma_{n,k}) + \frac{T_n q \varphi^2(\sigma_{n,k})}{\sigma_{n,k}^2 \sqrt n} + \sqrt n T_n \beta_q \right) 
\end{equation}
where $q \in \{1,2,\ldots\}$, $\sigma_{n,k} \geq \sup_{h \in \mc H_{n,k}} \|h\|_{2,\beta}$ for the norm $\|\cdot\|_{2,\beta}$ defined on p. 400 of \cite{DoukhanMassartRio}, and $\varphi(\sigma)$ is the bracketing entropy integral:
\[
 \varphi(\sigma) = \int_0^\sigma \sqrt{\log N_{[\,\,]}(u,\mc H_{n,k},\|\cdot\|_{2,\beta})} \, \mr d u \,.
\]
Exponential $\beta$-mixing and Lemma 2 of \cite{DoukhanMassartRio} (with $\phi(x) = x^{v}$) imply: 
\begin{equation} \label{e:beta:norm}
 \|\cdot\|_{2,\beta} \leq C  \|\cdot\|_{2v} \quad \mbox{on $L^{2v}$}
\end{equation}
for any $v > 1$, where the constant $C < \infty$  depends only on $v$ and the $\beta$-mixing coefficients. Taking $1 < v < 2s$, by H\"older's inequality and condition (a) we have:
\[
 \sup_{h \in \mc H_{n,k}} \|h\|_{2,\beta} \leq C  \sup_{h \in \mc H_{n,k}} \|h\|_{2v} \leq C \xi_k^{2-\frac{2s-v}{2sv}} \|E\|_{4s} \,.
\]
We therefore take $\sigma_{n,k} = C \xi_k^{2-\frac{2s-v}{2sv}} \|E\|_{4s}$.

To bound the bracketing entropy, define $\mc H_{n,k}^* = \{ b_0(x_0)b_1(x_1)h(x_0,x_1) : b_0,b_1 \in \mc B_k^*, h \in \mc H_n^*\}$ where $\mc B_k^* = \{ (c'\tilde b^k)/\xi_k : c \in S^{k-1}\}$ and $\mc H_{n}^* = \{ h^{trunc}_\alpha/E : \alpha \in \mc A\}$. For $\mc B_k^*$, note that $|c_0' \tilde b^k(x)/\xi_k - c_1' \tilde b^k(x)/\xi_k| \leq  (\xi_k^{-1} \|\tilde b^k(x)\|) \times  \|c_0 - c_1\|$ where $\| (\|\tilde b^k(x)\|/\xi_k) \|_p \leq (k/\xi_k^2)^{1/p}$ for any $p > 2$. By Theorem 2.7.11 of \cite{vdVW} and Lemma 2.5 of \cite{vdG}:
\[
 N_{[\,\,]}(u,\mc B_k^*,\|\cdot\|_p) \leq N\bigg( \frac{u}{2(k/\xi_k^2)^{1/p}},S^{k-1},\|\cdot\| \bigg) \leq \bigg( \frac{8 (k/\xi_k^2)^{1/p} }{u} + 1\bigg)^k \,.
\]
It follows by Lemma 9.25(ii) in \cite{Kosorok} that:
\begin{align}
 N_{[\,\,]}(3u,\mc H_{n,k}^*,\|\cdot\|_{p}) 
 & \leq  \bigg( \frac{8 (k/\xi_k^2)^{1/p} }{u} + 1\bigg)^{2k} N_{[\,\,]}(u,\mc H_n^*,\|\cdot\|_p)\,. \label{e:bracket:euclid}
\end{align}
Let $[f_l,f_u]$ be a $\varepsilon$-bracket for $\mc H_{n,k}^*$ under the $L^{\frac{4sv}{2s-v}}$ norm. Then $[ \xi_k^2 E f_l, \xi_k^2 E f_u]$ is a $\xi_k^2 \| E \|_{4s}\varepsilon$-bracket for $\mc H_{n,k}$ under the $L^{2v}$ norm, because $\| \xi_k^2 E (f_u - f_l)\|_{2v} \leq \xi_k^2 \| E \|_{4s} \| f_u - f_l\|_{\frac{4sv}{2s-v}}$. Taking $p = \frac{4sv}{2s-v}$ in display (\ref{e:bracket:euclid}) and using the fact that truncation of $\mc M^*$ doesn't increase its bracketing entropy, we obtain:
\begin{align}
  N_{[\,\,]} ( u ,\mc H_{n,k}, \|\cdot\|_{2v}) 
 & \leq N_{[\,\,]} \Big(  \frac{u}{\xi_k^2 \| E \|_{4s}} ,\mc H_{n,k}^*, \|\cdot\|_{\frac{4sv}{2s-v}} \Big) \notag \\
 & \leq \Big( \frac{24 \|E\|_{4s} \xi_k^{2-\frac{2s-v}{2sv}} k^{\frac{2s-v}{4sv}}}{u} + 1\Big)^{2k} N_{[\,\,]}\Big( \frac{u}{3 \xi_k^2 \|E\|_{4s}},\mc M^*,\|\cdot\|_{\frac{4vs}{2s-v}}\Big) \,. \label{e:entropy:Hnk}
\end{align} 
Now, by displays (\ref{e:beta:norm}) and (\ref{e:entropy:Hnk}) and condition (b):
\begin{align*}
 \varphi(\sigma) & = \int_0^\sigma \sqrt{\log N_{[\,\,]}(u,\mc H_{n,k},\|\cdot\|_{2,\beta})} \, \mr d u \\
 & \leq \int_0^\sigma \sqrt{\log N_{[\,\,]}(u/C,\mc H_{n,k},\|\cdot\|_{2v})} \, \mr d u \\
 & \lesssim k^{1/2} \int_0^\sigma \sqrt{\log \Big(1+ 24 C \|E\|_{4s} \xi_k^{2-\frac{2s-v}{2sv}} k^{\frac{2s-v}{4sv}} /u \Big)} \, \mr d u  +  ( \xi_k^2 \|E\|_{4s})^\zeta  \frac{\sigma^{1-\zeta}}{1-\zeta} \\
 & \lesssim  \|E\|_{4s} \xi_k^{2-\frac{2s-v}{2sv}} k^{\frac{1}{2} + \frac{2s-v}{4sv}}  \int_0^{\sigma/(24 C \|E\|_{4s} \xi_k^{2-\frac{2s-v}{2sv}} k^{\frac{2s-v}{4sv}})} \!\!\!\!\sqrt{ \log (1 + 1/u )} \, \mr d u + (\xi_k^2 \|E\|_{4s})^\zeta  \frac{\sigma^{1-\zeta}}{1-\zeta}  \,.
\end{align*}
Since $\sigma_{n,k} = C \xi_k^{2-\frac{2s-v}{2sv}} \|E\|_{4s}$, we obtain:
\begin{align*}
 \varphi( \sigma_{n,k} ) & \lesssim \|E\|_{4s} \xi_k^{2-\frac{2s-v}{2sv}} k^{\frac{1}{2} + \frac{2s-v}{4sv}}  \int_0^{\frac{1}{24} k^{-\frac{2s-v}{4sv}}} \!\!\!\! \sqrt{ \log (1 + 1/u )} \, \mr d u + (\xi_k^2 \|E\|_{4s})^\zeta  (\xi_k^{2-\frac{2s-v}{2sv}} \|E\|_{4s})^{1-\zeta} \\
 & \lesssim \|E\|_{4s} \xi_k^{2-\frac{2s-v}{2sv}} \sqrt{k \log k} + \|E\|_{4s} \xi_k^{2-\frac{2s-v}{2sv} + \zeta\frac{2s-v}{2sv}} 
\end{align*}
since $\int_0^\delta \sqrt{\log(1+1/u)}\, \mr du = O(\delta\sqrt{-\log \delta})$ as $\delta \to 0^+$. If $\xi_k^{\zeta \frac{2s-v}{2sv}} \lesssim \sqrt {k \log k}$ then the first term dominates and we obtain $\varphi(\sigma_{n,k}) = O(\xi_k^{2-\frac{2s-v}{2sv}} \sqrt{ k \log k})$. It follows by display (\ref{e:dmr:bd}) that:
\[
 \wh \Delta_{1,k,1} = O_p \bigg( \frac{ \xi_k^{2-\frac{2s-v}{2sv}} \sqrt{ k \log k} }{\sqrt n} + \frac{T_n q  k \log k}{ n} +  T_n \beta_q \bigg) \,.
\]

By Markov's inequality we may deduce $\wh \Delta_{1,k,2} = O_p ( \xi_k^{8s}/T_n^{4s-1}) $ and  $\wh \Delta_{1,k,3} = O ( \xi_k^{8s}/T_n^{4s-1})$. Choosing $T_n$ so that:
\[
 \frac{\xi_k^{8s}}{T_n^{4s-1}} \asymp \frac{ \xi_k^{2-\frac{2s-v}{2sv}} \sqrt{ k \log k} }{\sqrt n} 
\]
and $q = C_0 \log n $ for sufficiently large $C_0$ ensures, in view of the condition $\log n = O(\xi_k^{1/3})$, that $\wh \Delta_{1,k,1}$, $\wh \Delta_{1,k,2}$, and $\wh \Delta_{1,k,3}$ are all $O_p(\xi^{2-\frac{2s-v}{2sv}} \sqrt{(k \log k)/n})$.
For the remaining term, by condition (c) we have:
\begin{align*}
 \wh \Delta_{1,k,4} = \|\Pi_k ({\mb M}^{(\hat \alpha)} - {\mb M})|_{B_k}\|  \leq \ell^*(\hat \alpha) 
 & = \frac{1}{\sqrt n } \sqrt n \dot \ell_{\alpha_0}^*[\hat \alpha - \alpha_0] + O( \|\hat \alpha - \alpha_0\|_{\mc A}^2 )  = O_p(n^{-1/2}) 
\end{align*}
which is of smaller order. 
\end{proof}

\subsection{Proofs for Appendix \ref{ax:est:mat:fp}}

\begin{proof}[Proof of Lemma \ref{lem:beta:4}]
Lemma 2.2 of \cite{ChenChristensen-reg} gives the bound $\|\wh{\mf G}^o - \mf I\| = O_p( \xi_k (\log n)/\sqrt n)$. Let $\{T_n : n \geq 1\}$ be a sequence of positive constants to be defined and let: 
\begin{align*}
 G^{trunc}_{t+1} & = G_{t+1}^{1-\gamma} \ind\{ \| \tilde b^k(x_t) \| \|\tilde b^k(x_{t+1})\|^\beta |G_{t+1}^{1-\gamma}| \leq T_n\} \\
 G^{tail}_{t+1} & = G_{t+1}^{1-\gamma} \ind\{ \| \tilde b^k(x_t) \| \|\tilde b^k(x_{t+1})\|^\beta |G_{t+1}^{1-\gamma}| > T_n\}  \,.
\end{align*}
We then have:
\begin{align*}
  \sup_{v : \|  v\| \leq c} \| \wh{\mf T}^ov - \mf T^o v\|  
 &  \leq \sup_{v  : \|  v\| \leq c} \left\| \frac{1}{n} \sum_{t=0}^{n-1} \tilde b^k(X_t) G^{trunc}_{t+1} |\tilde b^{k}(X_{t+1})'v|^\beta - \mb E[ \tilde b^k(X_t) G^{trunc}_{t+1} |\tilde b^{k}(X_{t+1})'v|^\beta  ] \right\| \\
 & \quad + \sup_{v  : \|  v\| \leq c} \left\| \frac{1}{n} \sum_{t=0}^{n-1} \tilde b^k(X_t) G^{tail}_{t+1} |\tilde b^{k}(X_{t+1})'v|^\beta \right\| \\
 & \quad  + \sup_{v : \|  v\| \leq c} \left\| \mb E[ \tilde b^k(X_t) G^{tail}_{t+1} |\tilde b^{k}(X_{t+1})'v|^\beta  ] \right\| \quad =: \quad \wh T_1 + \wh T_2 + \wh T_3  \,.
\end{align*} 
Let $\mc H_{n,k} = \{ w' \tilde b^k(x_0) G_1^{trunc} |\tilde b^k(x_1)'v|^\beta : v \in \mb R^k, \|v\| \leq c, w \in S^{k-1} \}$. Then:
\[
 \wh T_1 \leq n^{-1/2} \times {\textstyle \sup_{h \in \mc H_{n,k}} | \mc Z_n(h) |}
\]
where $\mc Z_n$ is the centered empirical process on $\mc H_{n,k}$. Each $h \in \mc H_{n,k}$ is uniformly bounded by $c^\beta T_n$. Therefore, by Condition (a) and Theorem 2 of \cite{DoukhanMassartRio}:
\begin{equation} \label{e:dmr:bd:fp}
 \mb E[{\textstyle \sup_{h \in \mc H_{n,k}} | \mc Z_n(h)|}] = O  \left( \varphi(\sigma_{n,k}) + \frac{c^\beta T_n q \varphi^2(\sigma_{n,k})}{\sigma_{n,k}^2 \sqrt n} + \sqrt nc^\beta T_n \beta_q \right) 
\end{equation}
where $q \in \{1,2,\ldots\}$, $\sigma_{n,k} \geq \sup_{h \in \mc H_{n,k}} \|h\|_{2,\beta}$ for the norm $\|\cdot\|_{2,\beta}$ defined on p. 400 of \cite{DoukhanMassartRio}, and $\varphi(\sigma)$ is the bracketing entropy integral:
\[
 \varphi(\sigma) = \int_0^\sigma \sqrt{\log N_{[\,\,]}(u,\mc H_{n,k},\|\cdot\|_{2,\beta})} \, \mr d u \,.
\]
To calculate $\sigma_{n,k}$, by (\ref{e:beta:norm}) and H\"older's inequality we have:
\[
 \sup_{h \in \mc H_{n,k}} \|h\|_{2,\beta} \leq C  \sup_{h \in \mc H_{n,k}} \|h\|_{2s} \leq C c^\beta  \|G^{1-\gamma}\|_{2s} \xi_k^{1+\beta}
\]
where $\|G^{1-\gamma}\|_{2s}$ is finite by condition (b). Set $\sigma_{n,k} =  C c^\beta  \|G^{1-\gamma}\|_{2s} \xi_k^{1+\beta}$.

To bound the bracketing entropy, first fix $q > 2$ and let $w_1,\ldots,w_{N_1}$ be a $\varepsilon$-cover for $S^{k-1}$ and $v_1,\ldots,v_{N_2}$ be a $\varepsilon^{1/\beta}$-cover for $\{ v \in \mb R^k : \|v\| \leq c\}$. For any $ w \in S^{k-1}$ and $v \in \{ v \in \mb R^k : \|v\| \leq c\}$ there exist $v_i \in \{v_1,\ldots,v_{N_1}\}$ and $w_j \in \{w_1,\ldots,w_{N_2}\}$ such that:
\begin{align*}
 & w_j' \tilde b^k(x_0) G_1^{trunc} |\tilde b^k(x_1)'v_i|^\beta - \varepsilon \Big( (1+c^\beta)\| \tilde b^k(x_0) \| \| \tilde b^k(x_1)\|^\beta | G_1^{trunc}| \Big) \\
 & \leq w' \tilde b^k(x_0) G_1^{trunc} |\tilde b^k(x_1)'v|^\beta \\
 & \leq w_j' \tilde b^k(x_0) G_1^{trunc} |\tilde b^k(x_1)'v_i|^\beta + \varepsilon \Big( (1+c^\beta)\| \tilde b^k(x_0) \| \| \tilde b^k(x_1)\|^\beta | G_1^{trunc}| \Big) 
\end{align*}
where:
\[
 \left\|  2 \varepsilon \Big( (1+c^\beta)\| \tilde b^k(x_0) \| \| \tilde b^k(x_1)\|^\beta | G_1^{trunc}| \Big)\right\|_{2s} \leq 2 \varepsilon (1+c^\beta)\| G^{1-\gamma}\|_{2s} \xi_k^{1+\beta}  = \varepsilon C_0  \xi_k^{1+\beta}\,.
\]
where $C_0 = 2(1+c^\beta)\| G^{1-\gamma}\|_{2s}$. Therefore, given a $\varepsilon$-cover of $S^{k-1}$ and a $\varepsilon^{1/\beta}$-cover for $\{ v \in \mb R^k : \|v\| \leq c\}$ we can construct $\varepsilon C_0  \xi_k^{1+\beta}$-brackets for $\mc H_{n,k}$ under the $L^{2s}$ norm, and so by Lemma 2.5 of \cite{vdG}:
\[
 N_{[\,\,]}\big(u ,\mc H_{n,k}, \|\,\cdot\,\|_{2s} \big) \leq \Big(\frac{4C_0 \xi_k^{1+\beta} }{u} + 1\Big)^k \Big(\frac{4c(C_0 \xi_k^{1+\beta})^{1/\beta} }{u^{1/\beta}} + 1\Big)^k \,.
\]
By (\ref{e:beta:norm}) and the above display:
\begin{align*}
 \varphi(\sigma) & = \int_0^\sigma \sqrt{\log N_{[\,\,]}(u,\mc H_{n,k},\|\cdot\|_{2,\beta})} \, \mr d u \\
 & \leq \int_0^\sigma \sqrt{\log N_{[\,\,]}(u/C,\mc H_{n,k},\|\cdot\|_{2s})} \, \mr d u \\
 & \leq k^{1/2} \left( \int_0^\sigma \sqrt{\log \big(1+ 4 C C_0 \xi_k^{1+\beta}/u \big)} \, \mr d u  + \int_0^\sigma \sqrt{\log \big(1+ 4 c (C C_0 \xi_k^{1+\beta}/u)^{1/\beta} \big)}  \, \mr d u \right) \,.
\end{align*}
Since $\sigma_{n,k} = C c^\beta  \|G^{1-\gamma}\|_{2s} \xi_k^{1+\beta}$, by a change of variables we obtain $\varphi(\sigma_{n,k}) = O(\xi_k^{1+\beta} \sqrt k)$. Substituting into (\ref{e:dmr:bd:fp}):
\[
 \wh T_1 = O_p \bigg( \frac{ \xi_k^{1+\beta} \sqrt k }{\sqrt n} + \frac{T_n q  k }{ n} +  T_n \beta_q \bigg) \,.
\]

By Markov's inequality we may deduce $\wh T_2 = O_p ( \xi_k^{(1+\beta)2s}/T_n^{2s-1}) $ and  $\wh T_3 = O ( \xi_k^{(1+\beta)2s}/T_n^{2s-1})$. Choosing $T_n$ so that $\xi_k^{(1+\beta)2s}/{T_n^{2s-1}} \asymp \xi_k^{1+\beta} \sqrt{ k / n}$ and $q = C_0 \log n $ for large enough $C_0$ ensures, in view of the condition $(\log n)^{(2s-1)/(s-1)}k/n = o(1)$, that  $\wh T_1$, $\wh T_2$, and $\wh T_3$ are all $O_p(\xi_k^{1+\beta} \sqrt{k/n})$.

The expression for $\nu_{n,k}$ now follows from display (\ref{e:gtineq}) and the rates for $\wh{\mf G}^o$ and $\wh{\mf T}^o$.
\end{proof}

\subsection{Proofs for Appendix \ref{ax:id}}

\begin{proof}[Proof of Proposition \ref{p:a:id}]
We first show that any positive eigenfunction of $\mb M$ must have eigenvalue $\rho$. Suppose that there is some positive $\psi \in L^2$ and scalar $\lambda$ such that $\mb M \psi(x) = \lambda \psi(x)$. Then we obtain:
\[
 \lambda \langle \phi^*, \psi \rangle = \langle \phi^*, \mb M \psi \rangle = \langle \mb M^* \phi^*, \psi \rangle = \rho \langle \phi^*, \psi \rangle 
\]
with $\langle \phi^*, \psi \rangle  > 0$ because $\phi^*$ and $\psi$ are positive, hence $\lambda = \rho$. A similar argument shows that any positive eigenfunction of $\mb M^*$ must correspond to the eigenvalue $\rho$. 

It remains to show that $\phi$ and $\phi^*$ are the unique eigenfunctions (in $L^2$) of $\mb M$ and $\mb M^*$ with eigenvalue $\rho$. We do this in the following three steps. Let $F = \{ \psi \in L^2 : \mb M \psi = \rho \psi\}$. We first show that if $\psi \in F$ then the function $|\psi|$ given by $|\psi|(x) = |\psi(x)|$ also is in $F$. In the second step we show that $\psi \in F$ implies $\psi = |\psi|$ or $\psi = - |\psi|$. Finally, in the third step we show that $F = \{ s \phi : s \in \mb R\}$.

For the first step, first observe that $F \neq \{0\}$ because $\phi \in F$ by Assumption \ref{a:id:1}(b). Then by Assumption \ref{a:id:1}(c), for any $\psi \in F$ we have $\mb M |\psi| \geq |\mb M \psi| = \rho |\psi|$ and so $\mb M |\psi| - \rho |\psi| \geq 0$ (almost everywhere). On the other hand,
\[
 \langle \phi^*,\mb M |\psi| - \rho |\psi|\rangle =  \langle \mb M^* \phi^*,|\psi| \rangle - \rho \langle \phi^*,|\psi| \rangle = 0
\] 
which implies that $\mb M |\psi| = \rho |\psi|$ and hence $|\psi| \in F$. 

For the second step, take any $\psi \in F$ that is not identically zero. Suppose that $\psi = |\psi|$ on a set of positive $Q$ measure (otherwise we can take $- \psi$ in place of $\psi$). We will prove by contradiction that this implies $|\psi| = \psi$. Assume not, i.e. $|\psi| \neq \psi$ on a set of positive $Q$ measure. Then $|\psi| - \psi \geq 0$ (almost everywhere) and $|\psi| - \psi \neq 0$. But by step 1 we also have that $\mb M(|\psi| - \psi) = \rho(|\psi| - \psi)$. Then for any $\lambda > r(\mb M)$ we have
\[
 \frac{(\rho/\lambda)}{1-(\rho/\lambda)}(|\xi| - \xi) = \sum_{n\geq 1} \left(\frac{\rho}{\lambda}\right)^n (|\xi| - \xi) = \sum_{n \geq 1}\lambda^{-n} \mb M^n (|\xi| - \xi) > 0 
\]
(almost everywhere) by Assumption \ref{a:id:1}(c). Therefore, $|\psi| > \psi $ (almost everywhere). This contradicts the fact that $\psi = |\psi|$ on a set of positive $Q$ measure. A similar proof shows that if $-\psi = |\psi|$ holds on a set of positive $Q$ measure then $-\psi = |\psi|$.

For the third step we use an argument based on the Archimedean axiom (see, e.g., p. 66 of \cite{Schaefer1974}). Take any positive $\psi \in F$ and define the sets $S_+ = \{s \in \mb R : \psi \geq s \phi\}$ and $S_- = \{s \in \mb R : \zeta \leq s \phi\}$ (where the inequalities are understood to hold almost everywhere). It is easy to see that $S_+$ and $S_-$ are convex and closed. We also have  $(-\infty,0] \subseteq S_+$ so $S_+$ is nonempty. Suppose $S_-$ is empty. Then $\psi > s \phi$ on a set of positive measure for all $s \in (0,\infty)$. By step 2 we therefore have $\psi > s \phi$ (almost everywhere). But then because $L^2$ is a lattice we must have $\|\psi\| \geq s \|\phi\|$ for all $s \in (0,\infty)$ which is impossible because $\psi \in L^2$. Therefore $S_-$ is nonempty. Finally, we show that $\mb R = S_+ \cup S_-$. Take any $s \in \mb R$. Clearly $\psi - s \phi \in F$. By Claim 2 we know that either: $\psi - s \phi \geq 0$ (almost everywhere) which implies $s  \in S_+$ or $\psi - s \phi \leq 0$ (almost everywhere) which implies $s \in S_-$. Therefore $\mb R = S_+ \cup S_-$. The Archimedean axiom implies that the intersection $S_+ \cap S_-$ must be nonempty. Therefore $S_+ \cap S_- = \{s^*\}$  (the intersection must be a singleton else $\psi = s \phi$ and $\psi = s'\phi$ with $s \neq s'$) and so $\psi = s^* \phi$ (almost everywhere). This completes the proof of the third step.

A similar argument implies that $\phi^*$ is the unique positive eigenfunction of $\mb M^*$.
\end{proof}

\begin{proof}[Proof of Proposition \ref{p:a:exist}]
Assumption \ref{a:id:0}(a) implies that $r(\mb M) > 0$ (see Proposition IV.9.8 and Theorem V.6.5 of \cite{Schaefer1974}). The result now follows by Theorems 6 and 7 of \cite{Sasser} with $\rho = r(\mb M)$. That $\rho$ is isolated follows from the discussion on p. 1030 of \cite{Sasser}.
\end{proof}

\begin{proof}[Proof of Proposition \ref{p:lr}]
Consider the operator $\ol{\mb M} = \rho^{-1} \mb M$ with $\rho = r(\mb M)$. Proposition \ref{p:a:exist} implies that $\{1\} = \{  \lambda \in \sigma(\ol{\mb M}) : |\lambda| = 1\}$. Further, since $\mb M$ is power compact it has discrete spectrum \cite[Theorem 6, p. 579]{DunfordSchwartz}. We therefore have $\sup\{ |\lambda| : \lambda  \in \sigma(\ol{\mb M}) , \lambda \neq 1\} < 1$ and hence $\ol{\mb M} = (\phi \otimes \phi^*) + \mb V$ where $r(\mb V) < 1$ and $\ol{\mb M}$, $(\phi \otimes \phi^*)$ and $\mb V$ commute (see, e.g., p. 331 of \cite{Schaefer1974} or pp. 1034-1035 of \cite{Sasser}). Since these operators commute, a simple inductive argument yields:
\[
 \mb V^\tau = (\overline{\mb M} - (\phi \otimes \phi^*))^\tau 
 = \overline{\mb M}^\tau - (\phi \otimes \phi^*) = \rho^{-\tau} \mb M_\tau - (\phi \otimes \phi^*)
\]
for each $\tau \in T$. By the Gelfand formula, there exists $\epsilon > 0$ such that:
\begin{equation} \label{gelfand}
 \lim_{\tau \to \infty}\|\mb V^\tau \|^{1/\tau} = r(\mb V) \leq 1-\epsilon
\end{equation}
Let $\{\tau_k : k \geq 1\} \subseteq T$ be the maximal subset of $T$ for which $\|\mb V^{\tau_k}\| > 0$. If this subsequence is finite then the proof is complete. If this subsequence is infinite, then by expression (\ref{gelfand}),
\[
 \limsup_{\tau_k \to \infty} \frac{\log \|\mb V^{\tau_k}\|}{\tau_k} < 0\,.
\]
Therefore, there exists a finite positive constant $c$ such that for all $\tau_k$ large enough, we have:
\[
 \log \|\mb V^{\tau_k}\| \leq -c \tau_k
\]
and hence:
\[
 \| \rho^{-\tau_k} \mb M_{\tau_k} - (\phi \otimes \phi^*) \| \leq e^{-c \tau_k}
\]
as required.
\end{proof}

{\small \singlespacing
\putbib
}
\end{bibunit}

\end{document}